\documentclass[journal,twocolumn]{IEEEtran}

\IEEEoverridecommandlockouts

\usepackage[cmex10]{amsmath} 
\usepackage{amsfonts,acronym,amssymb,amsthm,dsfont}

\usepackage{subcaption}
\usepackage{graphicx} 
\usepackage{url}
\usepackage{cite}
\usepackage{hyperref}
\usepackage{braket}
\usepackage{physics}
\usepackage{enumerate}
\usepackage{lipsum}
\usepackage{diagbox}
\usepackage[dvipsnames]{xcolor}
\usepackage{mathrsfs}
\usepackage{bm}
\usepackage{stackengine}
\usepackage{float}
\usepackage{dblfloatfix}

\makeatletter
\def\blfootnote{\xdef\@thefnmark{}\@footnotetext}
\makeatother

\allowdisplaybreaks
\allowbreak

\newcommand{\calA}{\mathcal{A}}
\newcommand{\bbA}{\mathbb{A}}

\newcommand{\calB}{\mathcal{B}}

\newcommand{\calC}{\mathcal{C}}
\newcommand{\bbC}{\mathbb{C}}

\newcommand{\calD}{\mathcal{D}}
\newcommand{\bbD}{\mathbb{D}}

\newcommand{\calE}{\mathcal{E}}
\newcommand{\bbE}{\mathbb{E}}

\newcommand{\calF}{\mathcal{F}}

\newcommand{\calG}{\mathcal{G}}

\newcommand{\calH}{\mathcal{H}}

\newcommand{\calI}{\mathcal{I}}
\newcommand{\bbI}{\mathbb{I}}

\newcommand{\calJ}{\mathcal{J}}

\newcommand{\calL}{\mathcal{L}}
\newcommand{\bbL}{\mathbb{L}}

\newcommand{\calM}{\mathcal{M}}
\newcommand{\bbM}{\mathbb{M}}

\newcommand{\calN}{\mathcal{N}}
\newcommand{\bbN}{\mathbb{N}}

\newcommand{\calP}{\mathcal{P}}
\newcommand{\bbP}{\mathbb{P}}

\newcommand{\calQ}{\mathcal{Q}}

\newcommand{\calR}{\mathcal{R}}
\newcommand{\bbR}{\mathbb{R}}

\newcommand{\calS}{\mathcal{S}}

\newcommand{\calT}{\mathcal{T}}
\newcommand{\bbT}{\mathbb{T}}

\newcommand{\calU}{\mathcal{U}}

\newcommand{\calV}{\mathcal{V}}

\newcommand{\calX}{\mathcal{X}}

\newcommand{\calZ}{\mathcal{Z}}

\newcommand\ubar[1]{\stackunder[1.1pt]{$#1$}{\rule{1.2ex}{.08ex}}}

\newcommand{\dent}{\mathds{h}}
\newcommand{\brk}[1]{\ensuremath{\big[{#1}\big]}}

\def\ot{\otimes}
\DeclareMathOperator{\tra}{Tr}
\DeclareMathOperator{\D}{D}
\DeclareMathOperator{\I}{I}

\DeclareMathOperator{\F}{F}
\DeclareMathOperator{\Pu}{P}

\newcommand{\den}[2]{\ensuremath{\ket{#1}{\hspace{-1.6mm}}\bra{#2}}}

\newtheorem{theorem}{Theorem}
\newtheorem{corollary}{Corollary}
\newtheorem{definition}{Definition}

\newtheorem{remark}{Remark}

\newtheorem{lemma}{Lemma}[]

\newcommand{\indic}[1]{\ensuremath{\mathds{1}}}
\newcommand{\card}[1]{\ensuremath{\left\lvert{#1}\right\rvert}}   

\newcommand{\sbra}[2]{\ensuremath{\left[{#1}{\,:\,}{#2}\right]}}%
\newcommand{\sbr}[1]{\ensuremath{\left[{#1}\right]}}%
\newcommand{\pr}[1]{\ensuremath{\left({#1}\right)}}%
\newcommand{\br}[1]{\ensuremath{\left\{{#1}\right\}}}%

\newcommand{\Squad}{\hspace{0.5em}}
\acrodef{ACDIS}[ACDIS]{Adaptive Communication Decision and Information Systems}
\acrodef{AEP}{Asymptotic Equipartition Property}
\acrodef{AoA}{Angle of Arrival}
\acrodef{AWGN}{Additive White Gaussian Noise}
\acrodef{AVC}[AVC]{Arbitrarily Varying Channel}
\acrodef{BER}{Bit-Error-Rate}
\acrodef{BEC}{Binary Erasure Channel}
\acrodef{BPSK}{Binary Phase-Shift Keying}
\acrodef{BSC}{Binary Symmetric Channel}
\acrodef{RV}{Random Variable}
\acrodefplural{RV}{Random Variables}
\acrodef{MP}{Multi-Parent}
\acrodef{JTE}{Joint Typical Encoder}
\acrodef{BICM}[BICM]{Bit-Interleaved Coded-Modulation}
\acrodef{CDF}[CDF]{Cumulative Distribution Function}
\acrodef{CGF}[CGF]{Cumulant Generating Function}
\acrodef{CLT}[CLT]{Central Limit Theorem}
\acrodef{CSI}[CSI]{Channel State Information}
\acrodef{DMC}[DMC]{Discrete Memoryless Channel}
\acrodef{DMS}[DMS]{Discrete Memoryless Source}
\acrodef{ERM}[ERM]{Empirical Risk Minimization}
\acrodef{FER}[FER]{Frame Error Rate}
\acrodef{ICA}[ICA]{Independent Component Analysis}
\acrodef{iid}[i.i.d.]{independent and identically distributed}
\acrodef{IoT}[IoT]{Internet of Things}
\acrodef{KKT}[KKT]{Karush-Kuhn Tucker}
\acrodef{LASSO}[LASSO]{Least Absolute Shrinkage and Selection Operator}
\acrodef{LPD}[LPD]{Low Probability of Detection}
\acrodef{LDPC}[LDPC]{Low-Density Parity-Check}
\acrodef{LLMS}[LLMS]{Linear Least Mean Square}
\acrodef{LMS}[LMS]{Least Mean Square}
\acrodef{MAC}[MAC]{multiple-access channel}
\acrodef{QMAC}[QMAC]{Quantum Multiple-Access Channel}
\acrodef{MGF}[MGF]{Moment Generating Function}
\acrodef{MLC}[MLC]{Multi-Level Coding}
\acrodef{MLE}[MLE]{Maximum Likelihood Estimate}
\acrodef{MIMO}[MIMO]{Multiple-Input Multiple-Output}
\acrodef{MISO}{Multiple-Input Single-Output}
\acrodef{MSD}[MSD]{Multi-Stage Decoding}
\acrodef{MMSE}[MMSE]{Minimum Mean-Square Error}
\acrodef{PAC}[PAC]{Probably Approximately Correct}
\acrodef{PCA}[PCA]{Principal Component Analysis}
\acrodef{PDF}[PDF]{Probability Density Function}
\acrodefplural{PDF}{Probability Density Functions}
\acrodef{PMF}[PMF]{Probability Mass Function}
\acrodefplural{PMF}{Probability Mass Functions}
\acrodef{PPM}[PPM]{Pulse Position Modulation}
\acrodef{PSD}{Power Spectral Density}
\acrodef{PSK}{Phase Shift Keying}
\acrodef{QKD}{Quantum Key Distribution}
\acrodef{ROC}{Receiver Operating Characteristic}
\acrodef{CVQKD}{Continuous-Variable \ac{QKD}}
\acrodef{QPSK}{Quadrature Phase-Shift Keying}
\acrodef{RV}{random variable}
\acrodef{SIMO}{Single-Input Multiple-Output}
\acrodef{SNR}{Signal-to-Noise Ratio}
\acrodef{SVM}[SVM]{Support Vector Machine}
\acrodef{POVM}{Positive Operator-Valued Measure}
\acrodefplural{POVM}{Positive Operator-Valued Measures}
\acrodef{wrt}[w.r.t.]{with respect to}
\acrodef{WSS}{Wide Sense Stationary}
\acrodef{RHS}{Right Hand Side}
\acrodef{LHS}{Left Hand Side}
\acrodef{CPTP}{Completely Positive and Trace Preserving}
\acrodef{ADSI}[ADSI]{Action-Dependent State Information}

\allowdisplaybreaks
\begin{document}

\title{Keyless Covert Communication Over Quantum MACs with General Message Sets}

\author{
\IEEEauthorblockN{Hassan ZivariFard and Xiaodong Wang}\\
\thanks{The authors are with the Department of Electrical Engineering, Columbia University, New York, NY 10027. This work is supported in part by the U.S. Office of Naval Research (ONR) under grant N000142112155. E-mails: \{hz2863, xw2008\}@columbia.edu. Part of this work is presented at the 2025 IEEE International Symposium on Information Theory~\cite{2M3Tx_ISIT}. Part of this work is submitted to the 2026 IEEE Information Theory Workshop.
}
\vspace{-1.0cm}
}
\maketitle
\date{}
\begin{abstract}
We study covert classical communication over quantum \acp{MAC} with general message sets. 
Specifically, we consider a fully quantum \ac{MAC} with arbitrary message sets and an arbitrary number of transmitters. 
We demonstrate the feasibility of achieving a positive covert rate over this channel and establish general one-shot and asymptotic achievable rate regions. For classical-quantum \acp{MAC} with general message sets, we establish the covert capacity,  when the transmitters are restricted to deterministic encoding. Our result recovers, as a special case, known results for classical communication over classical \acp{MAC} with general message sets, covert communication of a classical message over a classical channel with two transmitters, and classical communication over quantum \acp{MAC}. We provide three examples of \acp{MAC} to which our results can be applied, either directly or indirectly, to achieve positive covert rates. Specifically, we first study covert communication over a finite-dimensional \ac{MAC} with a helper. We then analyze a classical Gaussian \ac{MAC} with a helper and derive its covert capacity. Finally, we extend the analysis to a single-mode bosonic \ac{MAC} with a helper and show that positive covert rates can also be achieved in this setting. 
To the best of our knowledge, this is the first work to achieve positive-rate covert communication over both classical and quantum \acp{MAC}.
\end{abstract}

\section{Introduction}
\label{sec:Intro}
The objective of covert communication is to render the transmission of messages undetectable~\cite{LPD_on_AWGN,Reliable_Deniable_Comm,LPD_by_Resolvability,LPD_over_DMC}. In a point-to-point classical \ac{DMC}, it is well established that it is possible to reliably and covertly transmit at most $O(\sqrt{n})$ bits over $n$ channel uses \cite{LPD_by_Resolvability,LPD_over_DMC}, provided that the encoder and decoder share a secret key of size $O(\sqrt{n})$ bits \cite{LPD_by_Resolvability}. Also, covert communication over discrete memoryless \acp{MAC} is studied in \cite{MAC_LPD}, where each transmitter shares a secret key with the receiver. The authors show that each transmitter can transmit on the order of $\sqrt{n}$ reliable and covert bits per $n$ channel uses. 
We note that positive covert communication rates can be achieved for point-to-point classical \acp{DMC} \cite{LPD_over_DMC} only when the symbol $x_0\in\calX$, transmitted by the sender in the no-communication mode, is redundant, which implies that the distribution induced on the warden's channel observation by $x_0\in\calX$ can also be induced using the channel input symbols $\{x\in\calX:x\ne x_0\}$. Building on this result, it has been demonstrated that positive covert rates can also be achieved over classical channels under various scenarios: (i) when a friendly jammer is present and there is a secret key shared between the legitimate terminals \cite{UninformedJammer,ISIT22,ISIT21,MyDissertation}, (ii) when the transmitter has access to \ac{CSI} \cite{Covert_With_State,Keyless22}, (iii) when the transmitter has access to \ac{ADSI} \cite{Action_Covert}, (iv) when the warden has uncertainty about the statistical characteristics of its channel and the transmitter and the receiver share a secret key \cite{Lee15,Deniable_ITW14}, (v) and when there is a cooperative user who knows either the message or the transmitter's codeword \cite{ISIT22,MyDissertation,Action_Covert}. 

Existing works on covert communication over quantum channels show that optimal rates obey the square root law.  
Specifically, covert communication over bosonic channels is considered in \cite{Bash_15,LPD_on_bosonic,Wang23}, where it is shown that the covert capacity follows the square-root law, and that achieving this law requires a secret key shared between the transmitter and the receiver. Additionally, \cite{Wang16,Covert_Quantum16,Bullock_25} establishes that the covert capacity of classical–quantum point-to-point channels also follows the square-root law, provided that the transmitter and receiver share a secret key of sufficiently high rate. We note that the availability of entangled qubits at both the transmitter and the receiver does not enable a positive covert rate \cite{Entangled_Bosonic} when there is also a secret key of infinite rate shared between the legitimate terminals. However, very recent results show that a keyless positive covert rate can be achieved over quantum state-dependent channels, where the channel state is modeled as the transmitter sharing an entangled state with the channel \cite{ITW24}.

In this paper, we study keyless covert communication over quantum \acp{MAC} with general message sets and, in contrast to previous works, demonstrate the achievability of positive covert rates in this setting. In our model, a set of transmitters, each with access to a subset of messages, cooperate to communicate covertly at a positive rate with a legitimate receiver, see Fig.~\ref{fig:System_Model}. 
We present a one-shot achievable covert rate region for this problem, which relies on various R\'{e}nyi mutual information error exponents. This region is then extended to the asymptotic setting, where we show that our achievable rate region recovers, as a special case, several known results: communication over classical \acp{MAC} with general message sets \cite{Han_MAC79, GunduzSimeone10, Romero17}, covert communication over classical \acp{MAC} when both transmitters send the same message to the receiver covertly \cite{MyDissertation}, and communication over quantum \ac{MAC} \cite{QMAC, Entanglement_MAC}. We further show that our achievable rate region is optimal for classical-quantum \acp{MAC} with general message sets when the transmitters are limited to use deterministic encoders. We present several examples that illustrate the feasibility of positive-rate covert communication, where our general results are applied either directly or indirectly. We first study covert communication over a classical Gaussian \ac{MAC} with a helper--modeled as a three-user \ac{MAC} in which one transmitter has access to the private messages of the other two and assists them in covertly communicating with a legitimate receiver. We derive the covert capacity of this channel. As a second example, we study covert communication over a single-mode bosonic \ac{MAC} with a helper, for which we present an achievable rate region. 
To the best of our knowledge, our results are the first fundamental limit of covert communication with positive rates over both classical and quantum \acp{MAC}.

Our achievability scheme combines random codebook generation, superposition coding \cite{ElGamalKim}, simultaneous pinching \cite{QIT_Hayashi}, and channel resolvability. We begin by identifying a special message hierarchy for \acp{MAC} with general message sets, which allows the corresponding reliability \cite{GunduzSimeone10} and resolvability regions to be described using a limited number of auxiliary \acp{RV}. We then propose a general procedure for characterizing achievable rate regions under any common message structure. This is achieved by transforming the original message set into the identified special hierarchy via the introduction of additional virtual users \cite{GunduzSimeone10}. For the reliability analysis, we leverage simultaneous pinching to establish the existence of a simultaneous decoder in the one-shot regime. For the covert analysis, we employ channel resolvability and pinching techniques \cite{QIT_Hayashi, Hayashi02}, enabling the resolvability analysis required to guarantee covertness.

Of particular relevance to this paper, classical \acp{MAC} with two transmitters and two private messages, where ``private'' refers to messages being unknown to the other transmitter, not secure messages, and one common message, was first studied in \cite{SlepianWol_MAC}, where the authors established the capacity region of this \ac{MAC} and further conjectured the capacity region of a \ac{MAC} with three transmitters and seven messages, including three private messages, three pairwise common messages, and one common message shared by all transmitters. However, Prelov later provided a counterexample that invalidates this conjecture \cite{Prelov84}. The correct generalization of classical \acp{MAC} with general message sets was presented by Han in \cite{Han_MAC79}, where the capacity characterization requires one auxiliary \ac{RV} per message, with all auxiliaries being independent. A later work by G\"{u}nd\"{u}z and Simeone \cite{GunduzSimeone10} introduced a new scheme for characterizing the capacity region of \acp{MAC} with general message sets, requiring only a limited number of auxiliary \acp{RV} and far fewer defining inequalities to specify the achievable rate region compared to the scheme introduced in \cite{Han_MAC79}. Adopting a graph-based approach, similar to that in \cite{GunduzSimeone10}, considerably simplifies the numerical evaluation of the capacity region; however, because it employs correlated codewords across transmitters, the resulting characterization has a higher complexity of codebook construction compared to \cite{Han_MAC79}, as discussed in \cite{Romero17}. Romero and Varanasi in \cite{Romero17} studied \acp{MAC} with general message sets from an order-theoretic perspective and showed that the schemes in \cite{Han_MAC79} and \cite{GunduzSimeone10} are the two extremes of a wide range of achievability schemes, each with polymatroidal structure. In this paper, since we demonstrate the feasibility of achieving positive covert rates using our results through specific examples, and because fewer auxiliary random variables are desirable for numerical evaluation, we adopt and extend the scheme of \cite{GunduzSimeone10}. 

Some related works that do not consider the covertness constraint are discussed in the following. Secure communication over quantum \acp{MAC} is studied in \cite{QMAC_Security}. The problem of communication over \ac{MAC} with cribbing is studied in \cite{WillemsCribbing}, for classical channels, and in \cite{UziCribbing}, for quantum channels. Also, secure communication and channel resolvability over classical \acp{MAC} with cribbing is studied in \cite{Helal_Cribbibg}. Secure communication over quantum state-dependent channels is studied in \cite{AshnuHayashi2020}. The problem of transmitting classical information over three-receiver quantum broadcast channels via simultaneous pinching is studied in \cite{Salek25}.

\textit{Notation:}
Let $\bbN$ and $\bbN^+$ denote the set of integers and the set of positive integers, respectively, and let $\bbR$ and $\bbC$ denote the set of real numbers and the set of complex numbers, respectively. Define $\bbR^+ \triangleq \br{x \in \bbR : x \geq 0}$ as the set of non-negative real numbers, and $\bbR^{++} \triangleq \bbR^+ \setminus {0}$ as the set of strictly positive real numbers. For any $x, y\in\bbR$, define $\sbra{x}{y} \triangleq [\lfloor x\rfloor, \lceil y \rceil]\cap \bbN$ and $[x] \triangleq \sbra{1}{x}$. For a set of indices $\calI\subset\bbN$, $M_\calI$ denotes $\br{M_i}_{i\in\calI}$. We denote the set of positive semi-definite operators on a finite-dimensional Hilbert space $\calH$ by $\calP(\calH)$ and denote the set of quantum states by $\calD(\calH)\triangleq\{\rho\in\calP(\calH):\tra[\rho]=1\}$. 
We use $\ket{\tau}$ to denote a pure quantum state and use $\tau$ to denote the corresponding density matrix. For a finite set $\calJ \triangleq \br{j_1, j_2, \dots, j_k}$ with $k \in \bbN^+$ and quantum states $\rho_{j_1}, \rho_{j_2}, \dots, \rho_{j_k}$, we use the notation $\bigotimes\limits_{j \in \calJ} \rho_j$ to denote the product state $\rho_{j_1} \otimes \rho_{j_2} \otimes \cdots \otimes \rho_{j_k}$. 
We also denote the space of the bounded linear operators on $\calH$ with $\calL(\calH)$. The identity operator on some Hilbert space $\calH$ is denoted by $\bbI$. For $\rho,\sigma \in \calD(\calH)$, the fidelity distance \cite{Jozsa94,Uhlmann76} between these two quantum states is defined as $\F(\rho,\sigma)\triangleq\lVert\sqrt{\rho}\sqrt{\sigma}\rVert_1$, where $\lVert\rho\rVert_1\triangleq\tra\left[\sqrt{\rho^\dagger\rho}\right]$. The purified distance \cite{Rastegin02,Tomamichel10} between $\rho$ and $\sigma$ is then given by $\Pu(\rho,\sigma)\triangleq\sqrt{1-\F^2(\rho,\sigma)}$. Moreover, for $\rho,\sigma \in \calD(\calH)$ and $\alpha \in (-1,0)\cup(0,\infty)$, the sandwiched R\'{e}nyi relative entropy \cite{Wilde14} is defined as
\begin{subequations}\label{eq:Sand_alpha_diver}
\begin{align}
    \ubar{\D}_{1+\alpha}\big(\rho\lVert\sigma\big)\triangleq\frac{1}{\alpha}\log\tra\left[\left(\sigma^{-\frac{\alpha}{2(1+\alpha)}}\rho\sigma^{-\frac{\alpha}{2(1+\alpha)}}\right)^{1+\alpha}\right].
\end{align}Also, the quantum relative entropy is defined as $\bbD(\rho\lVert\sigma)=\tra\sbr{\rho\pr{\log\rho-\log\sigma}}$. 
Consider the following classical-quantum states
\begin{align*}
    \rho_{UVB}\triangleq\sum_{u,v}p_{UV}(u,v)\den{u}{u}_U\otimes\den{v}{v}_V\otimes\rho_{B\lvert u,v},\\
    \rho_{U-V-B}\triangleq\sum_{u,v}p_{UV}(u,v)\den{u}{u}_U\otimes\den{v}{v}_V\otimes\rho_{B\lvert v},
\end{align*}where $\rho_{B\lvert v}$ is the marginal of $\rho_{UVB}$ and the notation $\rho_{U-V-B}$ is to emphasize that we have a Markov chain $U-V-B$, i.e., the systems $B$ and $U$ are independent given the system $V$. We define the R\'{e}nyi mutual information quantities as
\begin{align}
&\ubar{\I}_{1-\alpha}\pr{UV;B}_{\rho_{UVB}\lvert \rho_{UV}}\triangleq\min_{\sigma_B\in\calH_b} \ubar{\D}_{1+\alpha}\pr{\rho_{UVB}\lVert\rho_{UV}\otimes\sigma_B},\label{eq:Renyi_MI}\\ 
&\ubar{\I}_{1-\alpha}\pr{U;B\lvert V}_{\rho_{UVB}\lvert \rho_{UV}}\triangleq\nonumber\\
&\qquad\min_{\sigma_{U-V-B}:\sigma_{UV}=\rho_{UV}} \ubar{\D}_{1+\alpha}\pr{\rho_{UVB}\lVert\sigma_{U-V-B}},\label{eq:Renyi_CMI}
\end{align}where $\alpha \in (-1,0)\cup(0,\infty)$. 
\end{subequations}

The remainder of this paper is organized as follows. Section~\ref{sec:Problem_Statement} introduces the problem statement. Section~\ref{sec:Special_Message_Hierarchy} presents the \ac{MAC} with a special message hierarchy, and Section~\ref{sec:Main_Results} provides the main results. Section~\ref{sec:General_Message_Hierarchy} extends these results to \acp{MAC} with a general message hierarchy. Section~\ref{sec:Examples} presents illustrative examples, and Section~\ref{sec:Conclusion} offers concluding remarks. Proofs of the results are deferred to the appendices.
\begin{figure}[t!]
\centering
\includegraphics[width=9.0cm]{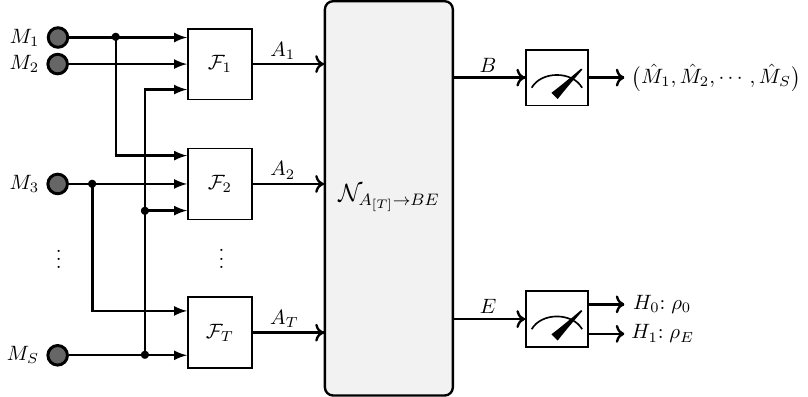}
\caption{Covert communication over a quantum \ac{MAC} with general message sets. 
}
\label{fig:System_Model}
\vspace{-0.7cm}
\end{figure}
\section{Problem Statement}
\label{sec:Problem_Statement}
Consider $S$ message sources indexed by $s \in [S]$, and denote the message set of the $s$-th source as $\calM_s$. A set of messages $\{M_s\}_{s\in[S]}$, where each $M_s \in \calM_s$, is to be transmitted over a quantum channel $\calN_{A_{[T]}\to BE}$, which is a completely positive and trace-preserving linear map taking an operator $\rho_{A_{[T]}}\in\calD\pr{\calH_1\otimes\cdots\otimes\calH_T}$ to an operator $\rho_{BE}\in\calD(\calH_b\otimes\calH_e)$. 
Each input to the channel corresponds to the output of an encoder, and the channel output $B$ is received by a legitimate receiver while the channel output $E$ is received by an adversary, which we refer to as warden, as illustrated in Fig.~\ref{fig:System_Model}. 
We assume that the messages are assigned to the encoders via a fixed, arbitrary mapping, and denote the set of the message source indices mapped to the $t^{\text{th}}$ transmitter~by 
\begin{align}
    \calI_t\triangleq\left\{{s_1^t},\dots,{s_{i_t}^t}\right\},\label{eq:Messages_at_TXt}
\end{align}where $t\in[T]$, $i_t\in\bbN^+$ is the number of the message sources accessible to Transmitter~$t$, and $s_j^t\in[S]$, for $j\in[i_t]$, is the index of the $j^{\text{th}}$ message source known by Transmitter~$t$. Without loss of generality, we assume that $\calI_t\ne\calI_{t'}$, for $t\ne t'$; otherwise, the two transmitters with the exact same set of message sources can be merged. Similarly, we assume that no two message sources are available for the same set of transmitters; otherwise, these message sources can be merged.

\subsection{Hypothesis Testing and Covertness Constraint}
Let $\phi_{0,t}\in\calD(\calH_t)$, for $t\in[T]$, denote the innocent input state transmitted by Transmitter~$t$ in the no-communication mode. The corresponding state induced at the warden's output in the one-shot regime, where the channel is used once, is given by
\begin{align}
    \rho_0\triangleq\tra_B\calN_{A_{[T]}\to BE}\left(\bigotimes_{t\in[T]}\phi_{0,t}\right).\label{eq:No_Comm}
\end{align}
Similarly, in the asymptotic regime, where the channel is used $n$ times independently, the corresponding state induced at the warden's output is
\begin{align}
    \rho_0^{\otimes n}\triangleq\tra_{B^n}\calN_{A_{[T]}^n\to B^nE^n}^{\otimes n}\left(\bigotimes_{t\in[T]}\phi_{0,t}^{\otimes n}\right).\label{eq:rho0_Asym}
\end{align}
Let $\tau_E$ denote the state induced at the warden's output by our coding scheme in the one-shot regime, and let $\tau_{E^n}$ denote the state induced at the warden's output by our coding scheme in the asymptotic regime. The warden's objective is to distinguish between the no-communication state, namely $\rho_0$ (or $\rho_0^{\otimes n}$ in the asymptotic regime), referred to as the null hypothesis $H_0$, and the communication state, namely $\tau_E$ (or $\tau_{E^n}$ in the asymptotic regime), referred to as the alternative hypothesis $H_1$. The warden may fail either by declaring communication when none occurs (false alarm) or by failing to detect an ongoing transmission (missed detection). Let $p_{\textup{FA}}= p(\textup{choose } H_1 \mid H_0 \textup{ is true})$ and $p_{\textup{MD}} = p(\textup{choose } H_0 \mid H_1 \textup{ is true})$ denote the probabilities of false alarm and missed detection, respectively. Assuming equal prior probabilities, $p(H_0)=p(H_1)=\frac{1}{2}$, the warden's average probability of error in the one-shot regime is $P_w=\frac{p_{\textup{FA}}+p_{\textup{MD}}}{2}$, and similarly, in the asymptotic regime, $P_w^{(n)}=\frac{p_{\textup{FA}}+p_{\textup{MD}}}{2}$. An uninformed detector that simply guesses between the two hypotheses achieves $P_w=\frac{1}{2}$ (or $P_w^{(n)}=\frac{1}{2}$) in the asymptotic regime. The goal of covert communication is to design a sequence of codes that forces the warden's optimal detector to perform arbitrarily close to this uninformed detector. 
By the Helstrom bound, the minimum achievable detection error probability is determined by the trace distance between the corresponding states \cite[Sec.~9.1.4]{Wilde_Book}:
\begin{subequations}\label{eq:pe_warden}
\begin{align}
\min P_w&=\frac{1}{2}\pr{1-\frac{1}{2}\left\lVert\tau_E-\rho_0\right\rVert_1},\label{eq:pe_warden_one_shot}\\
\min P_w^{(n)}&=\frac{1}{2}\pr{1-\frac{1}{2}\left\lVert\tau_{E^n}-\rho_0^{\otimes n}\right\rVert_1}.
\label{eq:pe_warden_asympt}
\end{align}
\end{subequations}
Note that the expressions in \eqref{eq:pe_warden} can be generalized to the case of non-uniform prior probabilities $p(H_0)$ and $p(H_1)$, see \cite[Sec.~II.B]{UninformedJammer}.

\subsection{One-Shot Regime} 
\begin{definition}
\label{defi:Encoder_Decoder}
A $(2^{R_1},\dots,2^{R_S},1)$ code for the quantum channel $\calN_{A_{[T]}\to BE}$ consists of the following:
\begin{itemize}
    \item message sets $\calM_s\triangleq\left[\left\lfloor2^{R_s}\right\rfloor\right]$, for $s\in[S]$;
    \item a encoding map at each transmitter, which is a quantum channel  $\calF^{(t)}:\calM_{\calI_t}\to\calD(\calH_t)$, for $t\in[T]$, that maps a set of messages $M_{\calI_t}\in\calM_{\calI_t}$ to a channel input $\rho_{A_t}\in\calD(\calH_t)$;
    \item a set of decoding \acp{POVM} $\left\{\calD^{\left(m_{[S]}\right)}_B\right\}_{m_{[S]}\in\calM_{[S]}}$, which maps a channel observation $\rho_B\in\calD(\calH_b)$ to $\hat{M}_{[S]}\in\calM_{[S]}$.
\end{itemize}
\end{definition}
The code is known by all the terminals, and the objective is to design a reliable and covert code. 
From Definition~\ref{defi:Encoder_Decoder} the input and output of the legitimate receiver's channel are given respectively by
\begin{subequations}\label{eq:Pe_rho0_CC_CSG}
\begin{align}
    \rho_{A_t}&=\calF^{(t)}(M_{\calI_t}),\quad\text{for}\quad t\in[T],\\
    \rho_B&=\tra_E\calN_{A_{[T]}\to BE}\left(\rho_{A_{[T]}}\right),\label{eq:Legi_Output}
\end{align}where the input states $\rho_{A_t}$, for $t\in[T]$, may in general be entangled across different transmitters. Therefore, from Definition~\ref{defi:Encoder_Decoder} and \eqref{eq:Legi_Output} the probability of error is defined as
\begin{align}
    P_e&\triangleq\bbP\left\{M_{[S]}\ne\hat{M}_{[S]}\right\}\nonumber\\
    &=\frac{1}{\card{\calM_{[S]}}}\sum_{m_{[S]}\in\calM_{[S]}}\tra\left[\left(\bbI-\calD^{\left(m_{[S]}\right)}_B\right)(\rho_B)\right].\label{eq:probaility_error}
\end{align}The code $(2^{R_1},\dots,2^{R_S},1)$ is reliable if
\begin{align}
    P_e&\le\epsilon.\label{eq:Pe}
\end{align}
\end{subequations}

Now from \eqref{eq:pe_warden_one_shot}, since the warden's average probability of error is only a function of the trace norm between $\tau_E$ and $\rho_0$, we define the covertness metric as,
    \begin{align}
        \left\lVert\tau_E- \rho_0\right\rVert_1&\le\delta.\label{eq:Covertness_Metric}
    \end{align}The metric in \eqref{eq:Covertness_Metric} indicates that the state induced at the output of the warden should approximate the state induced at the output of the warden when communication is not happening, i.e., communication is covert.

\begin{definition}[One-Shot Code]
\label{defi:Code}
    An $S$-tuple $(R_1,\dots,R_S)$ is said to be achievable for the quantum channel $\calN_{A_{[T]}\to BE}$ if there exists a sequence of $\epsilon$-reliable and $\delta$-covert $(2^{R_1},\dots,2^{R_S},1)$ codes that satisfy \eqref{eq:Pe} and \eqref{eq:Covertness_Metric}, respectively. 
\end{definition}
\subsection{Asymptotic Regime} 
We now define the problem in the asymptotic regime, where the channel is utilized $n$ times independently as $n$ approaches infinity. The corresponding $\left(2^{nR_1},\dots,2^{nR_S}, n\right)$ code for the channel $\calN_{A_{[T]}\to BE}^{\otimes n}$ is similarly defined as Definition~\ref{defi:Encoder_Decoder}, with the message set $\calM_s\triangleq\left[\left\lfloor{2^{nR_s}}\right\rfloor\right]$, for $s\in[S]$, encoding map at each transmitter $\calF_n^{(t)}:\calM_{\calI_t}\to\calD(\calH_t^n)$, for $t\in[T]$, and quantum states $\rho_{A_{[T]}^n}\in\calD(\calH_1^n\otimes\calH_2^n\otimes\cdots\otimes\calH_T^n),\rho_{B^n}\in\calD(\calH_b^n),\rho_{E^n}\in\calD(\calH_e^n)$. Then, the input and output of the legitimate receiver's channel are given respectively by
\begin{subequations}\label{eq:rho0_Pe_Asym}
\begin{align}
    \rho_{A_t^n}&=\calF_n^{(t)}(M_{\calI_t}),\quad\text{for}\quad t\in[T],\\
    \rho_{B^n}&=\tra_{E^n}\calN_{A_{[T]}^n\to B^nE^n}\left(\rho_{A_{[T]}^n}\right).\label{eq:Legi_Output_Asymp}
\end{align}Also, the probability of error is
\begin{align}
    P_e^{(n)}&\triangleq\bbP\left\{M_{[S]}\ne\hat{M}_{[S]}\right\}\nonumber\\
    &=\frac{1}{\card{\calM_{[S]}}}\sum_{m_{[S]}\in\calM_{[S]}}\tra\left[\left(\bbI-\calD^{\left(m_{[S]}\right)}_{B^n}\right)\left(\rho_{B^n}\right)\right].\label{eq:Pe_Asym}
\end{align}A sequence of codes $\left(2^{nR_1},\dots,2^{nR_S}, n\right)$ is reliable if
\begin{align}
    \lim_{n\to\infty}P_e^{(n)}&=0.\label{eq:Pe_Asymp}
\end{align}
\end{subequations}
It follows from \eqref{eq:pe_warden_asympt} that a sufficient condition for covert communication is
    \begin{align}
        \lim_{n\to\infty}\left\lVert\tau_{E^n}-\rho_0^{\otimes n}\right\rVert_1&=0,\label{eq:CC_CSK_Metric_Asymp}
    \end{align}which implies that the warden's optimal detector becomes asymptotically no better than random guessing.

\begin{definition}[Asymptotic Code]
\label{defi:Asymp_CC}
    An $S$-tuple $(R_1,\dots,R_S)$ is said to be achievable for the quantum channel $\calN_{A_{[T]}^n\to B^nE^n}$ if there exists a sequence of $\left(2^{nR_1},\dots,2^{nR_S},n\right)$ codes such that \eqref{eq:Pe_Asymp} and \eqref{eq:CC_CSK_Metric_Asymp} hold. The covert capacity region, denoted by $\calC_{\textup{CC-qMAC}}$, is defined as the closure of all achievable covert rate tuples.
\end{definition}

\subsection{Special Case: Classical Regime} 
\label{sec:Classical_Definition}
An important special case of the problem illustrated in Fig.~\ref{fig:System_Model} corresponds to covert communication over classical \acp{MAC} with general message sets. Here the channel is $\big(\calA_{[T]}, W_{BE|A_{[T]}},\calB,\calE\big)$, where $\calA_{[T]}\triangleq\calA_1\times\calA_2\times\dots\times\calA_T$ is the channel input alphabet, $W_{BE|A_{[T]}}$ is the channel law, and $\calB$ and $\calE$ are the channel output alphabets for the legitimate receiver and the warden, respectively. Then, for $t\in[T]$, the encoder $\calF^{(t)}$ takes as input a set of messages $M_{\calI_t}\in\calM_{\calI_t}$, and outputs the channel input $A_t\in\calA_t$; and the decoder $\calD_B$ takes as input the channel output $B$ and outputs an estimation of the messages, i.e.,  $\hat{M}_{[S]}$. 

In the classical case, $\rho_0$ in \eqref{eq:No_Comm} becomes the distribution induced at the warden's channel output when communication is not happening, and $\tau_E$ becomes the distribution induced by our code design. Then, for the covertness metrics defined in \eqref{eq:Covertness_Metric} and \eqref{eq:CC_CSK_Metric_Asymp}, we use the total variation distance instead of the trace distance.

Under the above problem setup, Definitions~\ref{defi:Encoder_Decoder} to \ref{defi:Asymp_CC} have their counterparts for the classical channels. In particular, the classical covert capacity region $\calC_{\textup{CC-MAC}}$ can be similarly defined.

\subsection{Cooperative Jamming} 
\label{sec:Cooperative_Jamming}
In general, a positive covert communication rate is achievable when the state induced at the warden's output in the no-communication mode can also be induced by non-innocent channel input states. The set of states that can be induced at the warden's output is determined solely by the channel (and, in our setting, by the mapping between messages and transmitters). Consequently, for channels in which the no-communication output state cannot be reproduced using non-innocent channel inputs, one may instead modify the no-communication output state through cooperative jamming \cite{UninformedJammer,ISIT21,ISIT22}.

Since our setup involves multiple transmitters and our schemes and proofs do not depend on the specific realization of the state induced at the warden's output in the no-communication mode, the framework can be naturally extended to a cooperative-jamming setting. In such a setting, during the no-communication mode, a subset of transmitters sends random states that convey no information to the legitimate receiver but serve to confuse the warden. This additional flexibility, however, comes at the cost of a non-vanishing energy expenditure in the no-communication mode.

Specifically, suppose that, in the no-communication mode, a subset $\calJ\subseteq[T]$ of the transmitters performs cooperative jamming. For each, $j\in\calJ$, Transmitter~$j$ privately generates a random density operator $\sigma_j\in\calD(\calH_{A_j})$ according to an arbitrary probability measure on $\calD(\calH_{A_j})$. For the asymptotic regime, Transmitter~$j$ generates a product state $\sigma_j^{\otimes n}\triangleq\sigma_{j,1}\otimes\sigma_{j,2}\otimes\dots\otimes\sigma_{j,n}$, where, for each $i\in[n]$, the density operator $\sigma_{j,i}\in\calD(\calH_{A_j})$ is generated according to an arbitrary probability measure on $\calD(\calH_{A_j})$. The realization of $\sigma_j$ (or $\sigma_j^{\otimes n}$ in the asymptotic regime) is known only to Transmitter~$j$. The warden's innocent state in \eqref{eq:No_Comm} and \eqref{eq:rho0_Asym} is then obtained by averaging over the randomness of the transmitted jamming states. Specifically,
\begin{subequations}\label{eq:No_Comm_Jamming}
\begin{align}
    \rho_0&\triangleq\bbE\sbr{\tra_B\calN_{A_{[T]}\to BE}\left(\bigotimes\limits_{j\in\calJ}\sigma_{j}\otimes\bigotimes\limits_{j\in\calJ^c}\phi_{0,j}\right)},\label{eq:No_Comm_Jamming_Oneshot}\\
    \rho_0^{\otimes n}&\triangleq\bbE\sbr{\tra_{B^n}\calN_{A^n_{[T]}\to B^nE^n}\left(\bigotimes\limits_{j\in\calJ}\sigma_{j}^{\otimes n}\otimes\bigotimes\limits_{j\in\calJ^c}\phi_{0,j}^{\otimes n}\right)},\label{eq:No_Comm_Jamming_Asymptotic}
\end{align}
\end{subequations}where the expectation is taken with respect to the probability measures used to generate the jamming states. 
The key difference between the innocent state $\phi_{0,j}$ and the jamming state $\sigma_j$, for $j\in\calJ$, is that $\phi_{0,j}$ is fixed and known to all terminals, including the warden, whereas $\sigma_j$ is a randomly generated density operator whose realization is known only to Transmitter~$j$. 
In this case, in Definitions~\ref{defi:Encoder_Decoder}, \ref{defi:Code}, and \ref{defi:Asymp_CC}, the states $\rho_0$ and $\rho_0^{\otimes n}$ are given by \eqref{eq:No_Comm_Jamming}, and all the results in Section~\ref{sec:Main_Results} remain valid.

\section{\texorpdfstring{\ac{MAC}}{MAC} with a Special Message Hierarchy}
\label{sec:Special_Message_Hierarchy}
Motivated by \cite{GunduzSimeone10}, we first consider \acp{MAC} with a specific message hierarchy. In Section~\ref{sec:General_Message_Hierarchy}, we show that this hierarchy can represent any \ac{MAC} with general message sets by introducing virtual transmitters.
\begin{definition}[Special Message Hierarchy]
\label{defi:Message_Hierarchy}
A \ac{MAC} with $T$ transmitters and $S$ messages, where Transmitter~$t$ has access to message indices $\calI_t$, for $t\in[T]$, is said to have a special message hierarchy if, for any $t_1,t_2\in[T]$ and $t_1\ne t_2$, the intersection set  $\calI_{t_1}\cap \calI_{t_2}$ is either an empty set or equal to $\calI_{t_3}$ for some $t_3\in[T]$.
\end{definition}This special message hierarchy imposes a specific structure on the message subsets $\calI_t$, for $t\in[T]$.  Now we begin by providing some essential definitions. 
We denote a graph $\Gamma$ by $\Gamma = (\calV,\calL)$, where $\calV$ is the vertex set and $\calL$ is the edge set, $\calL\subseteq\calV\times\calV$. $\Gamma$ is said to be a directed graph if it has two maps, $\text{init}:\calL\to\calV$ and $\text{ter}:\calL\to\calV$, assigning an initial vertex $\text{init}(\ell)$ and a terminal vertex $\text{ter}(\ell)$ to each edge $\ell\in\calL$, where $\text{init}(\ell)\ne\text{ter}(\ell)$. 
Then the edge $\ell$ is defined as being directed from $\text{init}(\ell)$ to $\text{ter}(\ell)$, and is represented by the ordered pair $\ell = \left(\text{init}(\ell),\text{ter}(\ell)\right)$. A sub-graph $P = (\calV',\calL')$, where
$\calV'= \{v_0,\dots, v_k\}\subseteq\calV$ and $\calL'= \{\ell_0,\dots,\ell_{k-1}\}\subseteq\calL$, is called a directed path on $\Gamma$ if each edge $\ell_i$ is directed from $v_i$ to $v_{i+1}$ for all $i<k$. 
A directed graph is referred to as a \textit{rooted directed graph} if there exists a directed path between a designated vertex, known as the root, and every other vertex in the graph. A directed cycle is a directed path that begins and ends at the same vertex. A directed graph is said to be \textit{acyclic} if it does not contain any directed cycles.

\begin{definition}[Message Graph]
    \label{defi:Message_Graph}
    A rooted, directed, acyclic graph $\Gamma=(\calV,\calL)$ is called a message graph if, for any edge $\ell\in\calL$, there is no directed path from  $\text{init}(\ell)$ to $\text{ter}(\ell)$ other than the edge $\ell$ itself.
\end{definition}
It is important to note that a message graph is not necessarily a (directed) tree, as multiple directed paths between two vertices that are not connected by the same edge may exist. Nonetheless, we can introduce terminology analogous to that used for trees, which will prove useful in the subsequent discussion. The \textit{parents} of a node $v_i$ are the nodes that are directly connected to $v_i$ and lie on the path from $v_i$ to the root. A child of a vertex $v_i$ is a vertex for which $v_i$ is the parent. The set $\bbD_i$ of \textit{descendants} of a vertex $v_i$ includes all of its children, the children of its children, and so on. The set $\bbA_i$ of \textit{ancestors} of a vertex $v_i$ includes all of its parents, the parents of its parents, and so on.  
A \textit{leaf} is a vertex with no children, meaning its descendant set is empty. 
\begin{figure*}[t!]
    \centering
    \hspace{-5mm}\begin{subfigure}[t]{0.5\textwidth}
        \centering
        \includegraphics[height=2.6in]{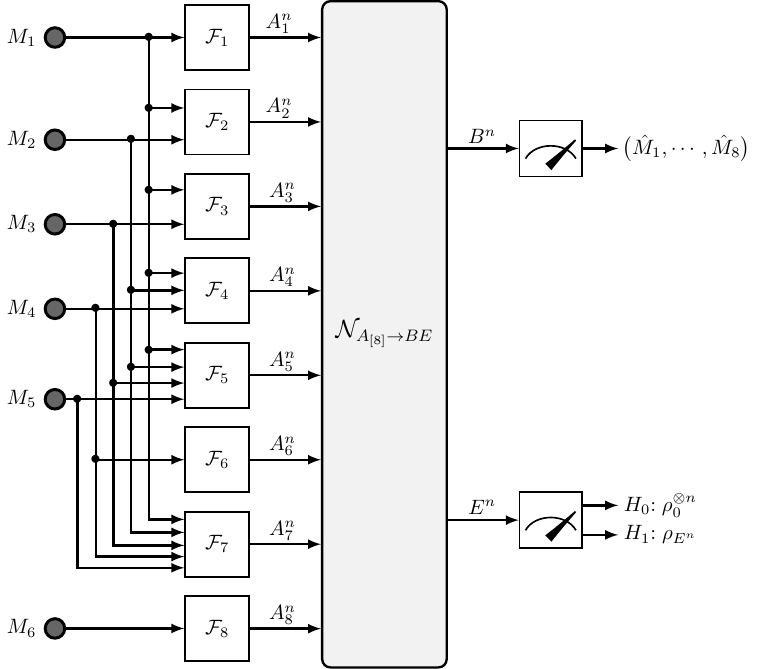}
        \caption{A \ac{MAC} with a special message hierarchy}
        \label{fig:4M6T}
    \end{subfigure}%
    ~ 
    \hspace{-5mm}\begin{subfigure}[t]{0.5\textwidth}
        \centering
        \includegraphics[height=2.6in]{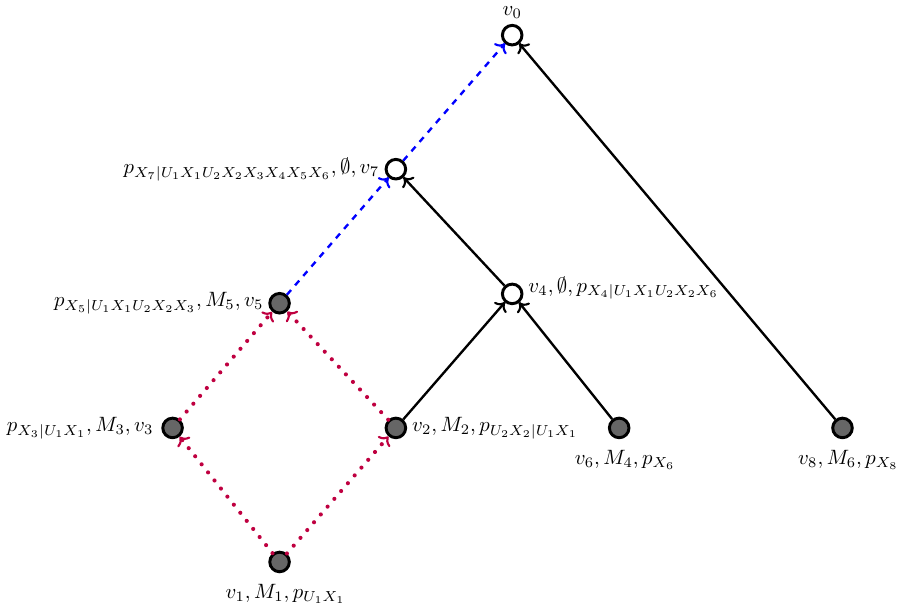}
        \caption{The associated message graph $\Gamma$}
        \label{fig:4M6T_Graph}
    \end{subfigure}
    \caption{A \ac{MAC} with a special message hierarchy and its corresponding message graph. In the message graph, each vertex is annotated with its private message and the distribution of the associated classical codebook. For example, $v_4,\emptyset,p_{X_4\lvert U_1X_1U_2X_2X_6}$ denotes that vertex $v_4$ has no private message and its codebook distribution is $p_{X_4\lvert U_1X_1U_2X_2X_6}$, while $v_3,M_3,p_{X_3\lvert U_1X_1}$ denotes that vertex $v_3$ has $M_3$ as its private message with codebook distribution $p_{X_3\lvert U_1X_1}$. In the associated message graph, the vertices with private messages, i.e., $\bbT=\{v_1, v_2, v_3, v_5,v_6,v_8\}$, are illustrated by filled vertices. Consider the subset $\{v_5\} \subset \bbT$ of vertices with private messages. The proper ancestral sub-graph associated with this subset, $\{v_5, v_7, v_0\}$, is illustrated with blue dashed edges, while the proper descendant sub-graph, $\{v_5, v_2, v_3, v_1\}$, is represented with purple dotted edges.}
    \label{fig:4M6T_Example}
    \vspace{-0.7cm}
\end{figure*}
\begin{definition}[Associated \texorpdfstring{\ac{MAC}}{MAC} Message Graph]
\label{defi:Graph_Representation}
Given a \ac{MAC} with a message structure as defined in \eqref{eq:Messages_at_TXt}, the associated graph $\Gamma=(\calV,\calL)$ is defined as follows. We have $\calV=\{v_0,v_1,\dots,v_T\}$, where the vertex $v_t$ corresponds to Transmitter~$t$ in the \ac{MAC}, for $t\in[T]$. We add a directed edge from $v_{t_1}$ to $v_{t_2}$ if $\calI_{t_1}\subset\calI_{t_2}$ and there does not exist any $\calI_{t_3}$, where $t_3\ne t_1,t_2$, such that $\calI_{t_1}\subset\calI_{t_3}\subset\calI_{t_2}$. Finally, we add the edges $(v_t, v_0)$ if $v_t$ has no parent. The root vertex $v_0$ does not correspond to a transmitter in the network, but it is included in the graph to make it a rooted graph. This step may be skipped if there is a transmitter observing all messages.
\end{definition}
By construction, the graph associated with a \ac{MAC} is a message graph. As an example, consider an 8-transmitter \ac{MAC} with 6 messages, as illustrated in Fig.~\ref{fig:4M6T}, and its corresponding message graph shown in Fig.~\ref{fig:4M6T_Graph}, which has the special message hierarchy. The graph consists of a total of 9 vertices, including one for each transmitter and a root vertex, $v_0$. For instance, note that vertex $v_4$ serves as the parent of $v_2$ and $v_6$ because Transmitter~4 has access to the messages available at Transmitter~2, namely $M_1$ and $M_2$, as well as the messages available at Transmitter~6, namely $M_4$. Vertex $v_1$ has two parents since $M_1$ is also known by Transmitter~2 and Transmitter~3. Moreover, examples of descendant sets include $\bbD_4 = \{v_1, v_2, v_6\}$ and $\bbD_5 = \{v_1, v_2, v_3\}$. 
\begin{definition}[Multi-Parent Vertices]
    In a message graph $\Gamma= (\calV,\calL)$, a vertex $v\in\calV$ is referred to as a \ac{MP} if it has more than one parent. The set of all such MP-vertices in $\calV$ is denoted by $\bbM$.
\end{definition}
Note that a message graph is a tree if and only if it contains no MP-vertices. In the example shown in Fig.~\ref{fig:4M6T_Example}, we have $\bbM = \{v_1,v_2\}$, indicating that the graph is not a tree. As we will see later, in characterizing the capacity region, users corresponding to MP-vertices require special treatment. Although the message graph is not a tree, we still refer to vertices without children as \textit{leaves} and denote this set by $\bbL$. In the example shown in Fig.~\ref{fig:4M6T_Example}, we have $\bbL = \br{v_1,v_6,v_8}$.
\begin{definition}[Private Messages]
    The private messages of Transmitter~$t$ are defined as the message indices in $\calI_t$ that are not accessible to any users in the descendant set $\bbD_t$ of user $t$. If $\bbD_t$ is empty, then all messages available to Transmitter~$t$ are considered its private messages. We denote the set of vertices with private messages by $\bbT$.
\end{definition}For example, in the message graph in Fig.~\ref{fig:4M6T_Example}, $\bbT=\br{v_1,v_2,v_3,v_5,v_6,v_8}$. The following lemma from \cite{GunduzSimeone10} plays a key role in our achievability proof; for completeness, we also provide its proof here.
\begin{lemma}[\hspace{-0.01mm}{\cite[Lemma~3.1]{GunduzSimeone10}}]
\label{lemma:One_Message_One_Vertex}
    If the underlying \ac{MAC} follows the special message hierarchy as in Definition~\ref{defi:Message_Hierarchy}, each Transmitter~$t$ (with an associated vertex $v_t \neq v_0$) can have at most one private message, denoted as $M(t)$. Furthermore, every message $M_s$, for $s\in[S]$, in the system serves as a private message for exactly one user, meaning $M_s = M(t)$ for a unique $t \in [T]$.
\end{lemma}
\begin{proof}
    Assume that $M_1$ and $M_2$ belong to $\calI_t$ but are not available to any transmitters in $\bbD_t$. Then, there needs to be at least one transmitter in the \ac{MAC} such that $M_1\in\calI_\ell$ but $M_2\notin\calI_\ell$, as otherwise, one can combine $M_1$ and $M_2$ into a single message. 
    But now $M_1\in\calI_t\cap\calI_\ell$, and therefore, by Definition~\ref{defi:Message_Hierarchy}, we must have $\calI_j=\calI_t\cap\calI_\ell$ for some Transmitter~$j\in[T]$. However, since $M_2\notin\calI_j$, it follows that $j\ne t$ and $v_j\in \bbD_t$. This leads to a contradiction with the initial assumption, and therefore, there can be at most one private message for each transmitter.

    Now let the message $M_s$, for $s\in[S]$, be the private message for two Transmitters $t$ and $j$. Therefore, $M_s\in\calI_t\cap\calI_j$ and by Definition~\ref{defi:Message_Hierarchy}, we must have $\calI_\ell=\calI_t\cap\calI_j$ for some Transmitter $\ell\in[T]$ and $M_s\in\calI_\ell$. Therefore, $v_\ell\in \bbD_t$ and $v_\ell\in \bbD_j$, which contradicts the assumption that $M_s$ is the private message for two Transmitters $t$ and $j$.
\end{proof}Private messages $M(t)$ are also included in Fig.~\ref{fig:4M6T_Example}. We denote the rate of the private messages of a Transmitter~$t$ with $R(t)$. For example, in Fig.~\ref{fig:4M6T_Example}, since $M(4)=M(7)=\emptyset$ we have $R(4)=R(7)=0$ whereas $M(8)=M_6$ and therefore $R(8)=R_6$.
In what follows, we introduce two definitions that are central to the presentation of our main results. The first, which generalizes \cite[Definition~3.5]{GunduzSimeone10}, is used to present the rate constraints arising in the reliability analysis, while the second is used to present the rate constraints arising in the resolvability analysis.
\begin{definition}[Proper Ancestral Sub-Graph]
\label{defi:Proper_Rooted_sub-graph}
    A graph $\Omega=(\calV',\calL')$ is defined as a proper ancestral sub-graph of the message graph $\Gamma=(\calV,\calL)$ if it satisfies the following conditions: 
    \begin{enumerate}[(i)]
        \item $\calV'\subseteq\calV$, $\calL'\subseteq\calL$, and  $v_0\in\calV'$;
        \item if a vertex belongs to the set $\calV'$ then all its ancestors also belong to the set $\calV'$, i.e., if $v_i\in\calV'$, for $i\in[T]$, then $\bbA_i\subseteq\calV'$;
        \item any vertex $v\in\calV'$ without an incoming edge is a vertex with a private message, i.e., if $v\in\calV'$ and there exists no $\ell\in\calL'$ such that ter$(\ell)=v$ then $v\in\bbT$.
    \end{enumerate}
\end{definition}
\begin{figure*}[b!]
    \hrulefill
    \setcounter{equation}{12}
\begin{subequations}\label{eq:Pe_Covertness}
\begin{align}
    \bbP\left\{\hat{M}\ne M\right\}&\le6\mu_{B,[T]}^\alpha2^{\alpha\pr{\sum\limits_{s=1}^SR_s-\ubar{\D}_{1-\alpha}\pr{\rho_{\tilde{X}_{[T]}B}\big\lVert\rho_{\tilde{X}_{[T]}}\otimes\rho_B}}}+6\sum\limits_{\calS\subset[S]}\mu_{B,\calT_\calS^c}^\alpha2^{\alpha\pr{\sum\limits_{s\in\calS}R_s-\ubar{\I}_{1-\alpha}\pr{\tilde{X}_{\calT_\calS};B\big\lvert \tilde{X}_{\calT_\calS^c}}}},\label{eq:Pe_Thm}\\
    \left\lVert\tau_E- \rho_0\right\rVert_1&\le\frac{2}{\sqrt{\alpha}}\left[\pr{\mu_{E,\card{\mathfrak{B}}}}^{\frac{\alpha}{2}}2^{\frac{\alpha}{2}\pr{-\sum\limits_{s=1}^SR_s+\ubar{\D}_{1+\alpha}\pr{\rho_{\tilde{X}_{[T]}E}\big\lVert\rho_{\tilde{X}_{[T]}}\otimes\rho_E}}}\right.\nonumber\\
    &\qquad\left.+\sum_{i=1}^{\card{\mathfrak{B}}}\pr{\mu_{E,i-1}}^{\frac{\alpha}{2}}2^{\frac{\alpha}{2}\pr{-\sum\limits_{s\in\calS_i}R_s+\ubar{\D}_{1+\alpha}\pr{\rho_{\tilde{X}\pr{\calS_i}E}\big\lVert\rho_{\tilde{X}\pr{\calS_i}}\otimes\rho_E}}}\right],\label{eq:Covertnes_Thm}
    \end{align}
    \end{subequations}
    \setcounter{equation}{10}
\end{figure*}
Intuitively, every set of vertices with private messages and all their ancestors corresponds to a proper ancestral sub-graph. Alternatively, since each message is the private message of only one vertex, for every subset $\calS\subset[S]$ of the messages, the set of all vertices that have access to at least one message that belongs to the subset $\calS$ corresponds to a proper ancestral sub-graph. Note that Definition~\ref{defi:Proper_Rooted_sub-graph} differs from the definition of a proper rooted sub-graph in \cite[Definition~3.5]{GunduzSimeone10}. Specifically, a proper rooted sub-graph in \cite[Definition~3.5]{GunduzSimeone10} is defined as a sub-graph that only satisfies the first two conditions of Definition~\ref{defi:Proper_Rooted_sub-graph}. The key motivation for adopting a different definition in this paper is to ensure the elimination of all redundant rate constraints, whereas in \cite{GunduzSimeone10}, only a subset of these redundant constraints is removed. Each rooted sub-graph corresponds to a set of transmitters, and this set is considered proper if, for every message available as a private message to a transmitter in the set, all the transmitters that have access to that private message are also included in the set. For instance, in the message graph $\Gamma$ of Fig.~\ref{fig:4M6T_Example}, the rooted
sub-graph $(\calV',\calL')$ with $\calV'=\{v_0,v_7,v_5,v_4,v_2\}$ and $\calL'=\{(v_7,v_0),(v_5,v_7),(v_4,v_7),(v_2,v_5),(v_2,v_4)\}$ is a proper ancestral sub-graph. Also, in the message graph $\Gamma$ of Fig.~\ref{fig:4M6T_Example}, the rooted sub-graph $(\calV^\star,\calL^\star)$ with $\calV^\star=\{v_0,v_7,v_5\}$ and $\calL'=\{(v_7,v_0),(v_5,v_7)\}$ is a proper ancestral sub-graph, which is illustrated with blue dashed edges in Fig.~\ref{fig:4M6T_Example}. Note that the sub-graph with $(\tilde{\calV},\tilde{\calL})$ with $\tilde{\calV}=\{v_0,v_7,v_5,v_2\}$ and $\tilde{\calL}=\{(v_7,v_0),(v_5,v_7),(v_2,v_5)\}$ is not a proper ancestral sub-graph since $(v_1,v_4)\notin\tilde{\calL}$. Also, the sub-graph with $(\tilde{\calV},\tilde{\calL})$ with $\tilde{\calV}=\{v_0,v_7,v_4\}$ and $\tilde{\calL}=\{(v_7,v_0),(v_4,v_7)\}$ is not a proper ancestral sub-graph since $v_4\notin\bbT=\{v_1,v_2,v_3,v_5,v_6,v_8\}$. Since each proper ancestral sub-graph $\Omega$ is defined by its set of vertices, for convenience, we will use the vertex set of $\Omega$ to denote $\Omega$. In the example shown in Fig.~\ref{fig:4M6T_Example}, there are 23 proper ancestral sub-graphs, which are listed as follows,
\begin{align}
    &\br{v_8,v_0},\br{v_6,v_4,v_7,v_0},\br{v_5,v_7,v_0},\br{v_8,v_5,v_7,v_0},\nonumber\\
    &\{v_8,v_6,v_4,v_7,v_0\},\{v_5,v_6,v_4,v_7,v_0\},\br{v_8,v_6,v_5,v_4,v_7,v_0},\nonumber\\
    &\br{v_3,v_5,v_7,v_0},\br{v_3,v_6,v_5,v_4,v_7,v_0},\br{v_3,v_8,v_5,v_7,v_0},\nonumber\\
    &\br{v_3,v_6,v_8,v_5,v_4,v_7,v_0},\br{v_2,v_5,v_4,v_7,v_0},\{v_2,v_3,v_5,v_4,\nonumber\\
    &\,\,\,v_7,v_0\},\br{v_2,v_6,v_5,v_4,v_7,v_0},\br{v_2,v_3,v_6,v_5,v_4,v_7,v_0},\nonumber\\
    &\br{v_2,v_8,v_5,v_4,v_7,v_0},\br{v_2,v_8,v_3,v_5,v_4,v_7,v_0},\{v_2,v_8,v_6,\nonumber\\
    &\,\,\,v_5,v_4,v_7,v_0\},\br{v_2,v_8,v_3,v_6,v_5,v_4,v_7,v_0},\{v_1,v_2,v_3,v_5,\nonumber\\
    &\,\,\,v_4,v_7,v_0\},\br{v_1,v_6,v_2,v_3,v_5,v_4,v_7,v_0},\{v_1,v_8,v_2,v_3,v_5,\nonumber\\
    &\,\,\,v_4,v_7,v_0\},\br{v_1,v_6,v_8,v_2,v_3,v_5,v_4,v_7,v_0}.\label{eq:Example_PASG}
\end{align}
\begin{definition}[Proper Descendant Sub-Graph]
\label{defi:Proper_Ance_Leaf_Path}
 A graph $\Xi=(\calV',\calL')$ is defined as a proper descendant sub-graph of the message graph $\Gamma=(\calV,\calL)$ if it satisfies the following conditions:
 \begin{enumerate}[(i)]
     \item $\calV'\subseteq\calV$ and $\calL'\subseteq\calL$ and the set $\calV'$ includes at least one leaf, i.e., $\exists v\in\calV'$ such that $v\in\bbL$;
     \item if a vertex belongs to the set $\calV'$ then all its descendant also belong to the set $\calV'$, i.e., if $v_i\in\calV'$, for $i\in[T]$, then $\bbD_i\subseteq\calV'$;
     \item if all the descendants with private messages of a vertex without a private message belong to the set $\calV'$, then this vertex must also belong to $\calV'$, i.e., if $v_i \in \calV - \bbT$, for $i \in [T]$, with $\bbD_i \triangleq \{\bbT', \bar{\bbT}'\}$, where $\bbT' \subseteq \bbT$ represents the vertices with private messages and $\bar{\bbT}'$ is the complement of $\bbT'$ \ac{wrt} $\bbD_i$, and if $\bbT' \subseteq \calV'$, then we must have $v_i \in \calV'$.
 \end{enumerate}
\end{definition}
Intuitively, every set of vertices with private messages, along with all their descendants, together with the vertices without private messages whose descendant sets are entirely contained within the descendants of this set, corresponds to a proper descendant sub-graph.  
Each proper descendant sub-graph corresponds to a set of transmitters that all have access to a subset of the private messages of the leaves. For example, in the message graph $\Gamma$ of Fig.~\ref{fig:4M6T_Example}, $\Xi\triangleq\{v_1,v_2,v_3\}$ or $\Xi\triangleq\{v_1,v_6\}$ are proper descendant sub-graphs, but $\Xi\triangleq\{v_1,v_3,v_5\}$ is not a proper descendant sub-graph since $v_2\in\bbD_5$ and $v_2\notin\Xi$, condition (ii) of Definition~\ref{defi:Proper_Ance_Leaf_Path} is violated. 
Also, in the message graph $\Gamma$ of Fig.~\ref{fig:4M6T_Example}, $\Xi^\star=\{v_5,v_3,v_2,v_1\}$ is a proper descendant sub-graph which is illustrated with purple dotted edges in Fig.~\ref{fig:4M6T_Example}, but $\Xi'=\{v_5,v_3,v_2,v_1,v_6,v_7\}$ is not a proper descendant sub-graph since all the descendants with private messages of $v_4$ belong to $\Xi'$ but $v_4\notin\Xi'$, condition (iii) of Definition~\ref{defi:Proper_Ance_Leaf_Path} is violated. 
In the message graph $\Gamma$ of Fig.~\ref{fig:4M6T_Example}, there are 23 proper descendant sub-graphs: 
\begin{align}
    &\br{v_1}, \br{v_1,v_3}, \br{v_1,v_2}, \br{v_1,v_2,v_3}, \br{v_1,v_2,v_3,v_5}, \nonumber\\
    &\br{v_1,v_6}, \br{v_1,v_6,v_2,v_4}, \br{v_1,v_6,v_3},\br{v_1,v_6,v_3,v_2,v_4},\nonumber\\
    & \br{v_1,v_6,v_3,v_2,v_4,v_5,v_7}, \br{v_1,v_8}, \br{v_1,v_8,v_2}, \br{v_1,v_8,v_3},\nonumber\\ &\br{v_1,v_8,v_2,v_3},\br{v_1,v_6,v_2,v_3,v_5}, \br{v_1,v_6,v_8}, \{v_1,v_6,v_8,\nonumber\\
    &\,\,\,v_2,v_4\}, \br{v_1,v_6,v_8,v_3}, \br{v_1,v_6,v_8,v_2,v_4,v_3},\{v_1,v_6,v_8,\nonumber\\
    &\,\,\,v_3,v_5,v_2,v_4,v_7\}, \br{v_6},\br{v_6,v_8},\br{v_8}.\label{eq:Example_PDSG}
\end{align}

\section{Main Results}
\label{sec:Main_Results}
In this section, we first establish a one-shot achievable rate region for \acp{MAC} with general message sets and subsequently extend this result to the asymptotic setting. We then show that the resulting rate region is optimal for classical-quantum channels.
\subsection{One-Shot Results}
\label{sec:One_Shot_Results}
\setcounter{equation}{13}
\begin{theorem}[One-shot Achievable Rate Region]
\label{thm:Achievable_One_Shot}
Given a quantum \ac{MAC} $\calN_{A_{[T]}\to BE}$, with a special message hierarchy as in Definition~\ref{defi:Message_Hierarchy}, and a classical-quantum state $\rho_{\tilde{X}_{[T]}A_{[T]}}=\sum\limits_{\tilde{x}_{[S]}}\prod\limits_{t=1}^Tp_{\tilde{X}_t\lvert X_{D_t}U_{D_t}}(\tilde{x}_t\lvert x_{D_t},u_{D_t})\den{\tilde{x}_1}{\tilde{x}_1}_{\tilde{X}_1}\otimes\cdots\otimes\den{\tilde{x}_T}{\tilde{x}_T}_{\tilde{X}_T}\otimes\theta_{A_1}^{x_1}\otimes\cdots\otimes\theta_{A_T}^{x_T}$, where for any multi-parent transmitter, i.e., $t\in\bbM$,  $\tilde{x}_t=(x_t,u_t)$ and otherwise $\tilde{x}_t=x_t$, such that $\rho_E=\rho_0$, with $\rho_{BE}=\tra_{\tilde{X}_{[T]}}\sbr{\bbI\otimes\calN_{A_{[T]}\to BE}\pr{\rho_{\tilde{X}_{[T]}A_{[T]}}}}$, there exists a $\pr{2^{R_1},2^{R_2},\cdots,2^{R_S},1}$ code such that \eqref{eq:Pe_Covertness}, given at the bottom of the previous page holds, where $\alpha\in\left(0,\frac{1}{2}\right)$; $\calT_\calS$, for $\calS\subset[S]$, denotes the set of channel input indices for which each channel input has access to at least one message $M_s$, for some $s \in \calS$; 
   $\mathfrak{B} \triangleq \pr{\calS_1, \calS_2, \dots, \calS_{\card{\mathfrak{B}}}}$ denotes an ordered collection of all non-empty strict subsets of $[S]$, arranged such that for any $1 \le i < j \le \card{\mathfrak{B}}$, we have $\card{\calS_i} \le \card{\calS_j}$; $\tilde{X}\pr{\calS_i}$, for $i\in[\card{\mathfrak{B}}]$, denotes the set of channel inputs for which each channel input has only access to a subset of messages $M_{\calS'}$, for some $\calS' \subseteq \calS_i$, $\rho_B$ and $\rho_E$ are the marginals of $\rho_{BE}$, and $\mu_{B,[T]}$, $\mu_{B,\calT_\calS^c}$, $\mu_E$, and $\mu_{E,i}$ are defined in~\eqref{eq:gs}.
\end{theorem}
The proof of Theorem~\ref{thm:Achievable_One_Shot}, outlined below and detailed in Appendix~\ref{proof:thm:Achievable_One_Shot}, relies on simultaneous pinching \cite{QIT_Hayashi}, superposition coding, and channel resolvability. Our codebook construction proceeds recursively from the leaves of the message graph toward the root.

For a transmitter corresponding to a multi-parent leaf, we first generate an auxiliary codebook and then superimpose a codeword on each auxiliary codeword. On the other hand, if a leaf has only a single parent, we directly generate its random codebook. For the parents of the leaves, if the corresponding vertex carries a private message, we construct its codebook via superposition coding over all combinations of its descendants’ codewords. If the vertex has no private message, we instead generate a single codeword through superposition coding for each such combination. This procedure is repeated until the root is reached. For example, in the \ac{MAC} illustrated in Fig.~\ref{fig:4M6T_Example}, the distribution of the codebook corresponding to each vertex is explicitly depicted within the message graph.

Our decoding quantum measurements are constructed using pinching maps defined with respect to the spectral decomposition of the receiver states \cite{QIT_Hayashi}, while the covertness analysis relies on a new channel resolvability lemma that builds on the properties of pinching maps, the sandwiched R\'{e}nyi relative entropy, and the purified distance. 

To clarify the notation in Theorem~\ref{thm:Achievable_One_Shot}, consider the \ac{MAC} illustrated in Fig.~\ref{fig:4M6T_Example}, where $S=6$ and $T=8$. For this specific \ac{MAC}, the number of non-empty exact subsets of $[6]$ is $2^6-2=62$. Consequently, the second term on the \ac{RHS} of \eqref{eq:Pe_Thm} and that of \eqref{eq:Covertnes_Thm} each consist of a sum of 62 terms. As an example, take the subset $\calS=\br{1,3,4}$. In this case, $\calT_\calS=\br{1,2,3,4,5,6,7}$, since each channel input in this set has access to at least one message $M_s$ with $s \in \calS$. Moreover, for $\calS_i=\br{1,3,4}$, we have $\tilde{X}(\calS_i)=(U_1,X_1,X_3,X_6)$, because Transmitters 1, 3, and 6 only have access to a subset of the messages in $\calS$, and among them, only Transmitter~1 corresponds to a multi-parent vertex.
\begin{corollary}[One-Shot Achievable Rate Region for \acp{MAC} with General Message Sets]
    By removing the covertness constraint $\rho_E=\rho_0$--and thus eliminating \eqref{eq:Covertnes_Thm}--the achievable rate region described in Theorem~\ref{thm:Achievable_One_Shot} reduces to a one-shot achievable rate region for quantum \acp{MAC} with general message sets.
\end{corollary}

\subsection{Asymptotic Results}
\label{sec:Asymp_Results}
In this section, we extend the one-shot achievable rate region established in Theorem~\ref{sec:One_Shot_Results} to the asymptotic setting, in which the transmitters are permitted to access the channel over multiple uses.
\begin{theorem}
\label{thm:Achievable_Asymp}
Let $\mathfrak{O}$ and $\mathfrak{X}$ denote the set of all proper ancestral sub-graphs and all proper descendant sub-graphs of the associated message graph $\Gamma$, respectively. Define the rate region $\calR_{\textup{CC-qMAC}}$ as follows,
\begin{subequations}
\begin{align}%
&\calR_{\textup{CC-qMAC}}\triangleq\bigcup_{\rho_{\tilde{X}_{[T]}A_{[T]}}\in\calG} \nonumber\\
&\left.\begin{cases}(R,R_2,\cdots,R_S):\\
  \sum\limits_{v\in\Omega}R(v)< I\left(X_\Omega;B\lvert X_{\Omega^c},U_{\Omega^c\cap\bbM}\right),\,\,\Omega\in\mathfrak{O}
\end{cases}\hspace{-3mm}\right\},\label{eq:inRnone}
\end{align}where
\begin{align}
  &\calG \triangleq \nonumber\\
  &\left.\begin{cases}\rho_{\tilde{X}_{[T]}A_{[T]}}:\\
  \rho_{\tilde{X}_{[T]}A_{[T]}}=\sum\limits_{\tilde{x}_{[S]}}\prod\limits_{t=1}^Tp_{\tilde{X}_t\lvert X_{D_t}U_{D_t}}(\tilde{x}_t\lvert x_{D_t},u_{D_t})\\
\times\den{\tilde{x}_1}{\tilde{x}_1}_{\tilde{X}_1}\otimes\dots\otimes\den{\tilde{x}_T}{\tilde{x}_T}_{\tilde{X}_T}\otimes\theta_{A_1}^{x_1}\otimes\dots\otimes\theta_{A_T}^{x_T},\\
\sum\limits_{v\in\Xi}R(v)> I\pr{X_\Xi;E},\,\Xi\in\mathfrak{X},\\
\rho_{BE}=\tra_{\tilde{X}_{[T]}}\sbr{\bbI\otimes\calN_{A_{[T]}\to BE}\pr{\rho_{\tilde{X}_{[T]}A_{[T]}}}},\\
\rho_E=\rho_0,\\
\end{cases}\hspace{-3mm}\right\},\label{eq:thm_S}
\end{align}
\end{subequations}$\tilde{x}_t=(x_t,u_t)$ if $t\in\bbM$ and $\tilde{x}_t=x_t$, otherwise,  $X_\Omega\triangleq\br{X_t:v_t\in\calV'}$ for a sub-graph $\Omega = (\calV',\calL')$, and $R(v)$ denotes the rate of the private message of the vertex $v$. 
An inner bound on the covert capacity region of a quantum \ac{MAC} $\calN_{A_{[T]}\to BE}$ with a special message hierarchy as defined in Definition~\ref{defi:Message_Hierarchy}, is
    \begin{align}
        \calC_{\textup{CC-qMAC}}\supseteq\calR_{\textup{CC-qMAC}}.\nonumber
    \end{align}
\end{theorem}Theorem~\ref{thm:Achievable_Asymp} is proved in Appendix~\ref{proof:thm:Achievable_Asymp}. Recall from Lemma~\ref{lemma:One_Message_One_Vertex} that for each $v$, $R(v) = R_s$ for some $s \in [S]$. To clarify the notation in Theorem~\ref{thm:Achievable_Asymp}, consider the \ac{MAC} illustrated in Fig.~\ref{fig:4M6T_Example}. Its associated message graph has 23 proper ancestral sub-graphs, listed in \eqref{eq:Example_PASG}, and 23 proper descendant sub-graphs, listed in \eqref{eq:Example_PDSG}. The rate constraint in \eqref{eq:inRnone} includes one rate constraint for each proper ancestral sub-graph, and the rate constraint in \eqref{eq:thm_S} includes one rate constraint for each proper descendant sub-graph. 
As an example, consider the proper ancestral sub-graphs $\br{v_3,v_6,v_5,v_4,v_7,v_0}$ and $\br{v_3,v_8,v_5,v_7,v_0}$. Their rate constraints are,
\begin{align}
    &R(v_3)+R(v_6)+R(v_5)+R(v_4)+R(v_7)\nonumber\\
    &\qquad=R_3+R_4+R_5+0+0\nonumber\\
    &\qquad<I(X_3,X_4,X_5,X_6,X_7;B\lvert U_1,X_1,U_2,X_2,X_8),\nonumber\\
    &R(v_3)+R(v_8)+R(v_5)+R(v_7)\nonumber\\
    &\qquad=R_3+R_6+R_5+0\nonumber\\
    &\qquad<I(X_3,X_5,X_7,X_8;B\lvert U_1,X_1,U_2,X_2,X_4,X_6),\nonumber
\end{align}respectively. Also, for the proper descendant sub-graphs $\br{v_1,v_8,v_3}$ and $\br{v_1,v_8,v_2,v_3}$, the corresponding rate constraints are
\begin{align}
    R(v_1)+R(v_8)+R(v_3)&=R_1+R_3+R_6\nonumber\\
    &>I(X_1,X_3,X_8;E),\nonumber\\ R(v_1)+R(v_8)+R(v_2)+R(v_3)&=R_1+R_2+R_3+R_6\nonumber\\
    &>I(X_1,X_2,X_3,X_8;E),\nonumber
\end{align}respectively. 
\begin{corollary}[Classical Communication Over a Quantum \ac{MAC} with General Message Sets]
    By removing the covertness constraint $\rho_E=\rho_0$--and thus eliminating the rate constraints in \eqref{eq:thm_S}--the achievable rate region described in Theorem~\ref{thm:Achievable_Asymp} reduces to the achievable rate region for classical communication over quantum \acp{MAC} with general message sets, generalizing a result presented in \cite[Theorem~3.2]{GunduzSimeone10} for classical channels.
\end{corollary}

\begin{remark}[An Equivalent Characterization of the Achievable Rate Region in Theorem~\ref{thm:Achievable_Asymp}]
    \label{rem:Symmetry}
    If the rate constraints in \eqref{eq:inRnone} and \eqref{eq:thm_S} are symmetric--in the sense that for every proper ancestral sub-graph $\Omega \in \mathfrak{O}$ there exists a proper descendant sub-graph $\Xi \in \mathfrak{X}$, and for every proper descendant sub-graph $\Xi \in \mathfrak{X}$ there exists a proper ancestral sub-graph $\Omega \in \mathfrak{O}$--such that $\sum\limits_{v\in\Xi}R(v)=\sum\limits_{v\in\Omega}R(v)$, then the rate constraint $\sum\limits_{v\in\Xi}R(v)> I\pr{X_\Xi;E}$ in \eqref{eq:thm_S} can be rewritten as $I\left(X_\Omega;B\lvert X_{\Omega^c},U_{\Omega^c\cap\bbM}\right)> I\pr{X_\Xi;E}$, for all $\Xi \in \mathfrak{X}$ and $\Omega \in \mathfrak{O}$. For example, see Corollary~\ref{cor:qMAC_Helper}.
\end{remark}

\begin{remark}[Interpretation]
The rate constraints in \eqref{eq:inRnone} correspond to the communication rates for a \ac{MAC} with a general message set \cite{Han_MAC79,GunduzSimeone10,Romero17}. The existence of a secret shared key between the legitimate terminals plays a crucial role in covert communication. For instance, in point-to-point channels, the terminals must share a secret key of rate on the order of $\sqrt{n}$ bits in order to transmit on the order of $\sqrt{n}$ covert bits over $n$ channel uses, unless the legitimate receiver’s channel is strictly better than the warden’s channel \cite{LPD_by_Resolvability}. In the latter case, no secret key is required.

In our setting, as discussed in Remark~\ref{rem:Symmetry}, the rate constraints in \eqref{eq:thm_S}, when combined with those in \eqref{eq:inRnone}, can be interpreted as mutual-information constraints (see~Corollary~\ref{cor:qMAC_Helper}). Consequently, this combination can be viewed as conditions on the channel quality required to achieve keyless positive covert communication rates. Note that, by employing achievability techniques similar to those used in the proof of Theorem~\ref{thm:Achievable_Asymp}, one can show that if a secret shared key of sufficient rate is available, either between each transmitter and the receiver individually or jointly among all legitimate terminals, the rate constraints in \eqref{eq:thm_S} can be removed.
\end{remark}
\begin{remark}[Feasibility of Achieving a Positive Covert Rate]
The feasibility of achieving a positive covert rate under the conditions of Theorems~\ref{thm:Achievable_One_Shot} and~\ref{thm:Achievable_Asymp} is determined primarily by whether the covertness constraint $\rho_E=\rho_0$ can be satisfied. The state $\rho_0$, induced at the warden's output in the no-communication mode, depends only on the channel and the innocent input states $\phi_{0,t}$, $t\in[T]$. Consequently, the feasibility of positive-rate covert communication depends critically on whether cooperation among the transmitters can induce the same state $\rho_0$ at the warden while simultaneously conveying information to the legitimate receiver.

From a geometric perspective, the covertness constraint $\rho_E=\rho_0$ is satisfied if and only if $\rho_0$ belongs to the convex hull of the states that can be induced at the warden's output by non-innocent signaling. In a conventional \ac{MAC}, where each transmitter has access only to its own private message, this condition is generally not satisfied. However, when messages are assigned to multiple transmitters, cooperation may enlarge the set of output states that can be induced at the warden, potentially causing $\rho_0$ to lie within the achievable set.

Since precise channel knowledge is essential for controlling the state induced at the warden's output, unknown phase noise, dephasing, or fading may alter the induced state in a manner that prevents the equality $\rho_E=\rho_0$ from being maintained. In such scenarios, the positive covert rates guaranteed by the main results of this paper may disappear, consistent with known impossibility results for several channel models \cite[Theorem~1]{Bash_15}.

More generally, the feasibility of positive-rate covert communication depends jointly on the channel structure and the message assignment among the transmitters. For example, if some messages are available only at a single transmitter, then cooperation cannot be used to conceal those messages, and achieving positive covert rates for them may be impossible. Conversely, when the legitimate receiver's channel and the warden's channel exhibit sufficient asymmetry with respect to the channel inputs, transmitter cooperation may enable the covertness constraint to be satisfied and positive covert rates to be achieved, as illustrated by the examples in Section~\ref{sec:Main_Results}.
\end{remark}

\begin{remark}[Effect of an External Secret Shared Key]
The availability of an external secret shared key among the legitimate terminals may enlarge the achievable covert rate region in Theorem~\ref{thm:Achievable_Asymp} and, in some cases, enable positive covert rates even when the corresponding keyless system cannot. This is because a shared secret key can help satisfy the constraints
\begin{align}
\sum_{v\in\Xi}R(v) > I(X_{\Xi};E), \qquad \Xi\in\mathfrak{X},\label{eq:Consts_Key}
\end{align}
which appear in Theorem~\ref{thm:Achievable_Asymp}.

Since our framework involves multiple transmitters and general message sets, several secret-key-sharing models can be considered. One natural and arguably strongest scenario is the availability of a common secret key shared among all transmitters and the receiver. Another interesting scenario is to associate a separate secret key $K_s$, for $s\in[S]$, with each message source $M_s$, where every transmitter that has access to message $M_s$ also has access to the corresponding secret key $K_s$.

The amount of secret key required generally depends on the relative quality of the legitimate receiver's channel and the warden's channel. For example, when the legitimate receiver has a sufficiently strong advantage over the warden, positive covert rates may be achievable without any pre-shared secret key. In contrast, when the warden's channel is comparatively strong, additional secret-key resources may be required to satisfy \eqref{eq:Consts_Key}.

It is important to emphasize, however, that the availability of a secret key does not, in general, relax the covertness constraint $\rho_E=\rho_0$ \cite{LPD_on_AWGN,LPD_by_Resolvability,LPD_over_DMC}. A secret key can assist in meeting the resolvability requirements, but does not alter the fundamental requirement that the state induced at the warden during communication must be indistinguishable from that induced in the absence of communication.
\end{remark}
\begin{remark}[Deterministic vs.\ Stochastic Encoders]
The achievability schemes in Theorems~\ref{thm:Achievable_One_Shot} and \ref{thm:Achievable_Asymp} are developed for deterministic encoders. However, they can be readily extended to stochastic encoders via channel prefixing~\cite{CsiszarKorner}. Specifically, an auxiliary random variable is introduced at each transmitter, and the corresponding channel input is generated according to a conditional distribution induced by the channel prefixing operation~\cite{CsiszarKorner}. In Theorem~\ref{thm:Achievable_Asymp_cq}, we show that our results are optimal for classical-quantum channels when the transmitters are restricted to deterministic encoding.
\end{remark}
\begin{remark}[Detemistic Encoder vs Stochastic Encoder]
    The achievability schemes in Theorems~\ref{thm:Achievable_One_Shot} and \ref{thm:Achievable_Asymp} are based on deterministic encoding but they can be extended to stochastic encoders via channel prefixing~\cite{CsiszarKorner}. In this case, an auxiliary random variable is assigned to each transmitter, and then the channel input is generated by means of channel prefixing~\cite{CsiszarKorner}. In Theorem~\ref{thm:Achievable_Asymp_cq} we show that our results are optimal for classical-quantum channels when the transmitters are restricted to only employ deterministic encoding.
\end{remark}
\subsection{Capacity of Classical-Quantum Channels}
In this section, we show that our achievability scheme is optimal for covert communication over the classical-quantum \acp{MAC} with general message sets. For $t\in[T]$, let $\calX_t$ be a finite alphabet corresponding to the input system of Transmitter~t, and $\calH_b$ and $\calH_e$ be finite-dimensional Hilbert spaces corresponding to the output systems of the legitimate receiver and the warden. Let $\calN_{X_{[T]}\to BE}$ denote a classical-quantum channel, defined as a linear map that assigns classical inputs $X_{[T]}\in\calX_1\times\cdots\times\calX_T$ to a quantum state $\rho_{BE}\in\calD\left(\calH_b\otimes\calH_e\right)$.
\begin{theorem}
\label{thm:Achievable_Asymp_cq}
Let $\mathfrak{O}$ and $\mathfrak{X}$ denote the set of all proper ancestral sub-graphs and all proper descendant sub-graphs of the associated message graph $\Gamma$, respectively. The covert capacity of the classical-quantum \ac{MAC} $\calN_{X_{[T]}\to BE}$ with a special message hierarchy, as defined in Definition~\ref{defi:Message_Hierarchy}, when the transmitters are restricted to deterministic encoding~is
\begin{subequations}\label{eq:CJ_cq}
\begin{align}%
&\calC_{\textup{cq-MAC}}=\bigcup_{\rho_{\tilde{X}_{[T]}}\in\calG}
\nonumber\\
&\left.\begin{cases}(R,R_2,\cdots,R_S):\\
  \sum\limits_{v\in\Omega}R(v)< I\left(X_\Omega;B\lvert X_{\Omega^c},U_{\Omega^c\cap\bbM}\right),\,\,\Omega\in\mathfrak{O}
\end{cases}\hspace{-3mm}\right\},\label{eq:inRnone_cq}
\end{align}where
\begin{align}
  \calG \triangleq \left.\begin{cases}\rho_{\tilde{X}_{[T]}}:\\
  \rho_{\tilde{X}_{[T]}}=\sum\limits_{\tilde{x}_{[S]}}\prod\limits_{t=1}^Tp_{\tilde{X}_t\lvert X_{D_t}U_{D_t}}(\tilde{x}_t\lvert x_{D_t},u_{D_t})\\
\times\den{\tilde{x}_1}{\tilde{x}_1}_{\tilde{X}_1}\otimes\dots\otimes\den{\tilde{x}_T}{\tilde{x}_T}_{\tilde{X}_T},\\
\sum\limits_{v\in\Xi}R(v)\ge I\pr{X_\Xi;E},\Xi\in\mathfrak{X},\\
\rho_{BE}=\tra_{\tilde{X}_{[T]}}\sbr{\bbI\otimes\calN_{A_{[T]}\to BE}\pr{\rho_{\tilde{X}_{[T]}A_{[T]}}}},\\
\abs{\calU_i}\le S+\prod_{t\in\bbA_i}\abs{\calX_t},i\in\bbM,\\
\rho_E=\rho_0,\\
\end{cases}\hspace{-2mm}\right\},\label{eq:thm_S_cq}
\end{align} and $\tilde{x}_t=(x_t,u_t)$ if $t\in\bbM$ and $\tilde{x}_t=x_t$, otherwise. 
\end{subequations}
\end{theorem}
The achievability proof of Theorem~\ref{thm:Achievable_Asymp_cq} follows directly from that of Theorem~\ref{thm:Achievable_Asymp}, and the converse proof is provided in Appendix~\ref{proof:thm:Achievable_Asymp_cq}. The proof of the cardinality bounds follows the approach in \cite{Han_MAC79} and is omitted for brevity.

\section{\texorpdfstring{\ac{MAC}}{MAC} with a General Message Hierarchy}
\label{sec:General_Message_Hierarchy}
Now, we show that the characterization of the special message hierarchy provided in Section~\ref{sec:Special_Message_Hierarchy} can be applied to describe any \ac{MAC} with general message sets. 
Given a \ac{MAC} with any message structure, consider all possible pairs of sets $\calI_i$ and $\calI_j$, where $i\neq j$. If $\calI_i \cap \calI_j$ is neither empty nor equal to the message set of any existing transmitters, create a ``virtual transmitter'' that has access to the messages in $\calI_i \cap \calI_j$ but has no channel input. After considering all pairs of transmitters, apply the same procedure to the new \ac{MAC}, which includes the virtual transmitters, and repeat the process until no further virtual transmitters need to be created. Since the cardinality of the message sets assigned to virtual transmitters decreases with each iteration, this process will eventually terminate after a finite number of steps. Ultimately, this process results in a \ac{MAC} that adheres to the special message hierarchy. 
Note that, although the virtual transmitters do not have channel inputs, each virtual transmitter corresponds to an MP-vertex, and thus an auxiliary \ac{RV} is assigned to each of them.

As an example, consider the \ac{MAC} in Fig.~\ref{fig:4M6T} with the first transmitter removed. The resulting \ac{MAC} no longer satisfies the special message hierarchy property because the message sets of the second and third transmitters intersect in $\br{M_1}$, which is neither empty nor identical to the message set of any other transmitter. 
We then introduce a ``virtual" first transmitter, which has no channel input but has access to $M_1$. With this addition, the resulting \ac{MAC} satisfies the special message hierarchy property, and its message graph coincides with that in Fig.~\ref{fig:4M6T_Graph}. 

After introducing all virtual users as described above, the number of auxiliary variables can often be reduced. Specifically, for any virtual user without a private message, its auxiliary \ac{RV} can be set equal to that of its descendants, thereby eliminating the need for an additional auxiliary variable. For example, consider a 4 transmitter \ac{MAC} with $M_{\calI_1}=\br{M_1,M_2,M_3}$, $M_{\calI_2} = \br{M_2,M_3,M_4}$, $M_{\calI_3} = \br{M_2,M_5}$, $M_{\calI_4} = \br{M_3,M_6}$. Following the above algorithm, we introduce the following virtual users: $M_{\calI_5}=\br{M_2,M_3}$, $M_{\calI_6} = \br{M_2}$, $M_{\calI_7} = \br{M_3}$. Note that we assign auxiliary \acp{RV}, say $U_6$ and $U_7$, to the two virtual transmitters that are leaves of the message graph. For the virtual transmitter without a private message, i.e, $v_5$, we do not introduce an additional auxiliary \ac{RV}; instead, we set its auxiliary variable to $(U_6, U_7)$.

Using the above algorithm, Theorems~\ref{thm:Achievable_One_Shot}, \ref{thm:Achievable_Asymp}, and \ref{thm:Achievable_Asymp_cq} can be applied to any \ac{MAC} with a general message set, with the corresponding channel inputs of the virtual users taken to be empty.

\section{Examples: \texorpdfstring{\ac{MAC}}{MAC} with a Helper}
\label{sec:Examples}
In this section, we present three examples that demonstrate the feasibility of achieving a positive covert communication rate over channels with multiple transmitters. As illustrated in Fig.~\ref{fig:MAC_Helper}, consider a quantum \ac{MAC} with three transmitters and two messages, where the first and second transmitters aim to covertly transmit messages $M_1$ and $M_2$ to the receiver, respectively. The third transmitter, which has access to both $M_1$ and $M_2$, facilitates covert communication between the transmitters and the receiver. We begin by applying Theorem~\ref{thm:Achievable_Asymp} to derive an achievable rate region for the quantum \ac{MAC} under consideration. As a direct application of our results, we present a finite-dimensional quantum \ac{MAC} with three transmitters and two messages, and derive an achievable rate region for this channel. We next consider the classical Gaussian instance of this \ac{MAC} and derive its {\em capacity region}. We then extend the analysis to the quantum regime by studying the single-mode bosonic counterpart, for which we obtain an achievable rate region. 
\begin{figure*}[t!]
    \centering
    \hspace{-5mm}\begin{subfigure}[t]{0.5\textwidth}
        \centering
        \includegraphics[height=5.0cm]{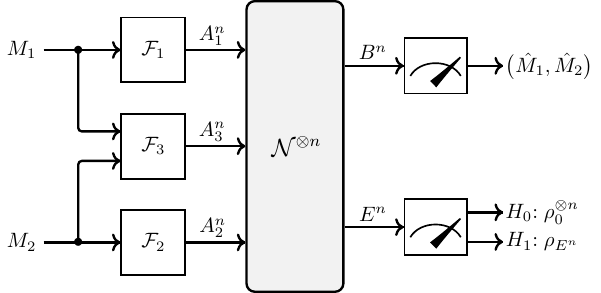}
        \caption{Covert Communication Over a  \ac{MAC} with a Helper}
        \label{fig:2M3T}
    \end{subfigure}%
    ~ 
    \hspace{-5mm}\begin{subfigure}[t]{0.7\textwidth}
        \centering
        \includegraphics[height=1.5in]{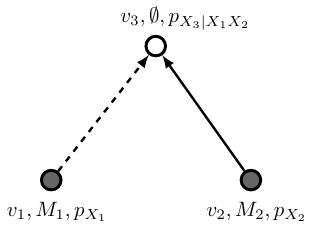}
        \caption{The associated message graph $\Gamma$}
        \label{fig:2M3T_Graph}
    \end{subfigure}
    \caption{A \ac{MAC} with a helper and its corresponding message graph.}
    \label{fig:MAC_Helper}
    \vspace{-0.7cm}
\end{figure*}
The following corollary establishes an achievable rate region for the quantum \ac{MAC} described above and depicted in Fig.~\ref{fig:MAC_Helper}.
\begin{corollary}
\label{cor:qMAC_Helper}
Define the rate region $\calR_{\textup{CC-qMAC}}$ as follows,
\begin{subequations}
\begin{align}%
&\calR_{\textup{CC-qMAC}}\triangleq\nonumber\\
&\bigcup_{\rho_{X_{[3]}A_{[3]}}\in\calG} %
\left\{ \begin{array}{rl}
  (R_1,R_2) \,:\;
    R_1&< I(X_1,X_3;B\lvert X_2),\\
    R_2&< I(X_2,X_3;B\lvert X_1),\\
    R_1+R_2&< I(X_1,X_2,X_3;B),
	\end{array}
\right\},\label{eq:MAC_Helper_Reliability}
\end{align}where
\begin{align}
  \calG \triangleq\left.\begin{cases}\rho_{X_{[3]}A_{[3]}}:\rho_{X_{[3]}A_{[3]}}=\sum\limits_{x_1}\sum\limits_{x_2}\sum\limits_{x_3}p_{X_1}(x_1)\\
  \times p_{X_2}(x_2)p_{X_3\lvert X_1X_2}(x_3\lvert x_1,x_2)\den{x_1}{x_1}_{X_1}\\
  \otimes\den{x_2}{x_2}_{X_2}\otimes\den{x_3}{x_3}_{X_3}\otimes\theta_{A_1}^{x_1}\otimes\theta_{A_2}^{x_2}\otimes\theta_{A_3}^{x_3},\\
I(X_1,X_3;B\lvert X_2)>I(X_1;E),\\
I(X_2,X_3;B\lvert X_1)>I(X_2;E),\\
I(X_1,X_2,X_3;B)>I(X_1,X_2,X_3;E),\\
\rho_{BE}=\tra_{X_{[3]}}\sbr{\bbI\otimes\calN_{A_{[3]}\to BE}\pr{\rho_{X_{[3]}A_{[3]}}}},\\
\rho_E=\rho_0,\\
\end{cases}\hspace{-2mm}\right\}.\label{eq:thm_qMAC_Helper_S}
\end{align} 
An inner bound for the covert capacity region of the quantum \ac{MAC} $\calN_{A_{[3]}\to BE}$, depicted in Fig.~\ref{fig:MAC_Helper},~is
    \begin{align}
        \calC_{\textup{CC-qMAC}}\supseteq\calR_{\textup{CC-qMAC}}.\nonumber
    \end{align}
    \end{subequations}
\end{corollary}
\begin{proof}
    The proof follows from that of Theorem~\ref{thm:Achievable_Asymp}, noting that the associated message graph has three proper ancestral sub-graphs, $\mathfrak{O}=\br{\br{v_1,v_3}, \br{v_2,v_3}, \br{v_1,v_2,v_3}}$, and three proper descendant sub-graphs, $\mathfrak{X}=\br{\br{v_1}, \br{v_2}, \br{v_1,v_2,v_3}}$. Since $v_3$ has no private message, i.e., $R(v_3)=0$, the rate constraints associated with proper ancestral sub-graphs and those associated with proper descendant sub-graphs exhibit the symmetry described in Remark~\ref{rem:Symmetry}. Consequently, the two sets of constraints can be combined and expressed as the mutual information constraints in \eqref{eq:thm_qMAC_Helper_S}. 
\end{proof}
\begin{corollary}[Classical Communication Over a Quantum \ac{MAC}]
    By setting  $X_3=A_3=\emptyset$ and removing the covertness constraint $\rho_E=\rho_0$--and thus eliminating the mutual information constraints in \eqref{eq:thm_qMAC_Helper_S}--the achievable rate region described in Corollary~\ref{cor:qMAC_Helper} reduces to the achievable rate region for classical communication over quantum \acp{MAC}, presented in \cite[Theorem~9]{QMAC}.
\end{corollary}
\begin{corollary}[Communication Over a classical \ac{MAC} with a Helper]
\label{cor:Classical_MAC_Helper}
    By removing the covertness constraint $\rho_E=\rho_0$--and thus eliminating the mutual information constraints in \eqref{eq:thm_qMAC_Helper_S}--and assuming that the channel is classical, the achievable rate region described in Corollary~\ref{cor:qMAC_Helper} reduces to the capacity region for communication over a classical \ac{MAC} with a helper, presented in \cite[Problem~5.20]{ElGamalKim}, see also \cite{Han_MAC79}.
\end{corollary}
\begin{corollary}[Covert Communication Over a classical \ac{MAC} with one message]
\label{cor:Covert_Classical_MAC_One_Message}
    By setting  $X_2=A_2=\emptyset$ and $M_2=\emptyset$, and assuming the channel is classical, the achievable rate region described in Corollary~\ref{cor:qMAC_Helper} reduces to the capacity region for covert communication over the classical \ac{MAC} with one message, presented in \cite[Theorem~36]{MyDissertation}, when there is no secret-shared key between the transmitter and the receiver. 
\end{corollary}Note that the classical capacity results mentioned in Corollary~\ref{cor:Classical_MAC_Helper} and Corollary~\ref{cor:Covert_Classical_MAC_One_Message} can also be recovered using Theorem~\ref{thm:Achievable_Asymp_cq}.
\subsection{A Finite-Dimensional \texorpdfstring{\ac{MAC}}{MAC} with a Helper}
\label{sec:Example_Finite}
\setcounter{equation}{17}
Let $\calH_{A_1}=\calH_{A_2}=\calH_{A_3}=\calH_E=\bbC^2$ and $\calH_B=\bbC^8$, 
and let the innocent input states be $\phi_{0,1}
=\phi_{0,2}=\phi_{0,3}=\ketbra{0}{0}$. 
\begin{subequations}\label{eq:Finit_Exam_Channel}
Define the following isometry
\begin{align}
    W\triangleq \ket{0}_E\otimes\Pi_+ + \ket{1}_E\otimes\Pi_-,
\end{align}where $\Pi_+
=\frac{\bbI+V}{2}$, $\Pi_-=\frac{\bbI-V}{2}$, and $V = Z\otimes Z\otimes Z$, with $Z$ denoting the Pauli-$Z$ operator. Let the channel be as follows
\begin{align}
    &\calN_{A_{[3]}\to BE}(\rho_{A_{[3]}})=W\rho_{A_{[3]}}W^\dagger\nonumber\\
    &=\ketbra{0}{0}_E\otimes \Pi_+\rho_{A_{[3]}}\Pi_+ +\ketbra{1}{1}_E\otimes\Pi_-\rho_{A_{[3]}}\Pi_- \nonumber\\
    &\qquad+ \ketbra{0}{1}_E\otimes\Pi_+\rho_{A_{[3]}}\Pi_- + \ketbra{1}{0}_E\otimes\Pi_-\rho_{A_{[3]}}\Pi_+.
\end{align} 
\end{subequations}
Therefore, the warden's channel is
\begin{align*}
    \calN_{A_{[3]}\to E}(\rho_{A_{[3]}})&=\tra_B\pr{W\rho_{A_{[3]}}W^\dagger}\nonumber\\
&= \tra(\Pi_+\rho_{A_{[3]}})
\ketbra{0}{0} + \tra(\Pi_-\rho_{A_{[3]}})
\ketbra{1}{1},
\end{align*}where the second equality follows from the cyclicity of the trace operator and since $\Pi_+\Pi_-=\Pi_-\Pi_+=0$, $\Pi_+^2=\Pi_+$, and $\Pi_-^2=\Pi_-$. Observe that the computational-basis states $\ket{000}$, $\ket{011}$, $\ket{101}$, and $\ket{110}$ belong to the (+1)-eigenspace of $V=Z\otimes Z\otimes Z$. Consequently, each of the corresponding density operators induces the output state $\ketbra{0}{0}$ at the warden, whereas the remaining computational-basis states induce the output state $\ketbra{1}{1}$.  
Also, the legitimate receiver's channel is 
\begin{align*}
    \calN_{A_{[3]}\to B}(\rho_{A_{[3]}})&=\tra_E\pr{W\rho_{A_{[3]}}W^\dagger}\nonumber\\
&= \Pi_+\rho_{A_{[3]}}\Pi_+ + \Pi_-\rho_{A_{[3]}}\Pi_-.
\end{align*}
\begin{corollary}
\label{cor:Example_Finite}
An achievable rate region for the quantum \ac{MAC} $\calN_{A_{[3]}\to BE}$ defined in \eqref{eq:Finit_Exam_Channel} is
\begin{align}
     \calC_{\textup{CC-qMAC}}=\left\{
(R_1,R_2):
\;
R_1<1,\;
R_2<1
\right\}.\nonumber
\end{align}
\end{corollary}
\begin{proof}
For $t\in[3]$, let $\calX_t=\{0,1\}$, $\theta_{A_t}^{x_t}=\ketbra{x_t}{x_t}$. The helper chooses $X_3=X_1\oplus X_2$, hence the set of transmitted states is $\br{\ketbra{000}{000},\ketbra{011}{011},\ketbra{101}{101},\ketbra{110}{110}}$. 
With this choice, every transmitted codeword belongs to the $+1$-eigenspace of $V=Z\otimes Z\otimes Z$.
Hence, the warden always observes the state $\rho_E=\ketbra{0}{0}$. 
Since the innocent input state is $\ketbra{000}{000}$, the output of the warden in the no-communication mode is also $\rho_0=\ketbra{0}{0}$. Therefore, the covertness constraint is satisfied.

Furthermore, if $\ket{x_1x_2x_3}\in\br{\ket{000},\ket{011},\ket{101},\ket{110}}$ we have $\Pi_+\ket{x_1x_2x_3}=\ket{x_1x_2x_3}$ and $\Pi_-\ket{x_1x_2x_3}=0$ and therefore $\calN_{A_{[3]}\to B}(\rho_{A_{[3]}})=\rho_{A_{[3]}}$. Consequently,
\begin{subequations}\label{eq:Exam_Finite_MIs}
\begin{align}
&I(X_1;E)=I(X_2;E)=I(X_1,X_2,X_3;E)=0,\\
&I(X_1,X_3;B|X_2)=H(X_1|X_2)=1,\\
&I(X_2,X_3;B|X_1)=H(X_2|X_1)=1,\\
&I(X_1,X_2,X_3;B)=H(X_1,X_2)=2,
\end{align}
\end{subequations}
where the last equalities follow from choosing $X_1$ and $X_2$ independently according to the Bernoulli-$\frac{1}{2}$ distribution. 
Substituting \eqref{eq:Exam_Finite_MIs} into Corollary~\ref{cor:qMAC_Helper} completes the proof.
\end{proof}

\subsection{Classical Gaussian \texorpdfstring{\ac{MAC}}{MAC} with a Helper}
\label{sec:Example_Gaussian}
Let the channel in the problem depicted in Fig.~\ref{fig:MAC_Helper} be a classical Gaussian channel, where the channel outputs are given by
\begin{subequations}\label{eq:Gaussian_MAC_Helper}
\begin{align}
    B&=X_1+X_2+hX_3+N_B,\,\, E=X_1+X_2+X_3+N_E,\label{eq:Gaussian_Channel_Law}
\end{align}and $h\in\bbR$ is a constant known by all the terminals, $N_B\sim\calN(0,\sigma_B^2)$ and $N_E\sim\calN(0,\sigma_E^2)$ are independent white Gaussian noises, and both are independent of the channel inputs $(X_1,X_2,X_3)$. We assume that the channel inputs are subject to the following power constraints
\begin{align}
    \frac{1}{n}\sum_{i=1}^nx_{1,i}^2&\le P_1,\quad\frac{1}{n}\sum_{i=1}^nx_{2,i}^2\le P_2,\quad
    \frac{1}{n}\sum_{i=1}^nx_{3,i}^2\le P_3.\label{eq:Power_Const_3}
\end{align}
\end{subequations}
Also, let the innocent symbols be $x_{0,1}=x_{0,2}=x_{0,3}=0$; therefore, the distribution induced at the output of the warden in the no-communication mode is $q_0^{\otimes n}$, where $q_0=\calN(0,\sigma_E^2)$. We assume that the covertness constraint is as defined in Section~\ref{sec:Classical_Definition}, and define the capacity as the closure of all achievable rates, denoted by $\calC_{\textup{CC-GMAC}}$, which is characterized in the following theorem.
\begin{theorem}
\label{thm:Capacity_Gaussian}
    The covert capacity of the three-user \ac{MAC} described in \eqref{eq:Gaussian_MAC_Helper} is 
\begin{align}
&\calC_{\textup{CC-GMAC}}\triangleq
  \left\{ \begin{array}{rl}
  &\hspace{-4mm}(R_1,R_2):\\
&\hspace{-4mm}R_1<\frac{1}{2}\log\pr{1+\frac{(h-1)^2\tilde{P}_1}{\sigma_B^2}}\\
&\hspace{-4mm}R_2<\frac{1}{2}\log\pr{1+\frac{(h-1)^2\tilde{P}_2}{\sigma_B^2}}\\
&\hspace{-4mm}R_1+R_2<\frac{1}{2}\log\pr{1+\frac{(h-1)^2\tilde{P}_3}{\sigma_B^2}}\\
	\end{array}\nonumber
\right\},
\end{align}where $\tilde{P}_1\triangleq\min\br{P_1,P_3}$, $\tilde{P}_2\triangleq\min\br{P_2,P_3}$, and $\tilde{P}_3\triangleq\min\br{P_1+P_2,P_3}$.
\end{theorem}
\begin{remark}
    When the third transmitter in \eqref{eq:Gaussian_Channel_Law} is removed, the problem reduces to covert communication over a two-user \ac{MAC} channel, as studied in \cite{MAC_LPD}. In that setting, the authors show that the covert capacity obeys the square root law and therefore vanishes as the block length grows. This behavior is also reflected in Theorem~\ref{thm:Capacity_Gaussian}: setting $P_3 = 0$ results in a covert capacity of zero. 
\end{remark}
The capacity region in Theorem~\ref{thm:Capacity_Gaussian} is plotted in Fig.~\ref{fig:Plot_Gaussian} for $\sigma_B^2 = 1$, $h = 3$, and $P_1 = P_2 = 5$, with six different values of $P_3\in\br{0,0.5,1,3,6,12}$. As shown in the figure, the capacity region enlarges as $P_3$ increases. Moreover, for this example, the capacity region remains unchanged for $P_3 \ge P_1 + P_2 = 10$, since in this case, both transmitters can already communicate with the legitimate receiver using their maximum power. 
\begin{figure}[t!]
\centering
\includegraphics[width=8.5cm]{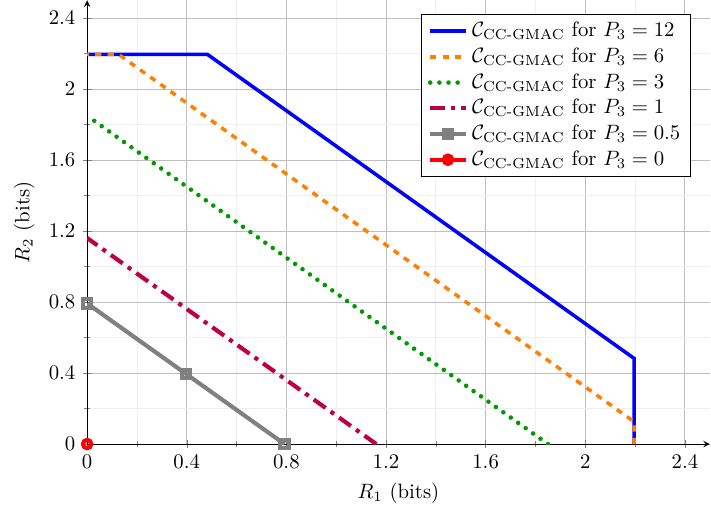}
\caption{The capacity region of the Gaussian \ac{MAC} with a helper shown for $\sigma_B^2 = 1$, $h = 3$, and $P_1 = P_2 = 5$, for $P_3\in\br{0,0.5,1,3,6,12}$.
}
\label{fig:Plot_Gaussian}
\vspace{-0.7cm}
\end{figure}
\begin{proof}
Similar to the arguments used in \cite[Section~8.6]{Cover_Book}, \cite[Section~VI.B]{Covert_With_State}, and \cite[Section~V.A]{Action_Covert}, one can show that the achievability proof of Corollary~\ref{cor:qMAC_Helper} holds when the channel is a classical \ac{AWGN} channel, as described in \eqref{eq:Gaussian_Channel_Law}, and the channel inputs are restricted to satisfy the power constraints in \eqref{eq:Power_Const_3}.

Here, we set the channel input of the Transmitter~$t$, for $t \in [3]$, to follow a zero-mean Gaussian distribution with variance $\tilde{P}_t \leq P_t$, chosen to satisfy the covertness constraint in \eqref{eq:thm_qMAC_Helper_S}. Specifically, because the distribution induced in the no-communication mode is $\calN(0,\sigma_E^2)$, it follows that $E \sim \calN(0,\sigma_E^2)$. Hence, the third transmitter should select its channel input as $X_3 = -(X_1 + X_2)$, thereby canceling the effect of the first and second transmitters' inputs at the warden. We now proceed to determine the optimal choices of $\tilde{P}_t$, for $t \in [3]$, that maximize the achievable rate region under the covertness constraint. As seen in Fig.~\ref{fig:Gaussian_Cases}, given a fixed $P_3$, there are five possible cases for $P_1$ and $P_2$. Therefore, power allocations for $X_1$ and $X_2$ are then chosen based on the relationship between $P_1$, $P_2$, and $P_3$ as follows:
\begin{enumerate}[(i).]
    \item If $P_1 + P_2 \le P_3$, the first and second transmitters use their full power budgets, i.e., $P_1$ and $P_2$, respectively. The third transmitter, having sufficient power, is able to cancel $X_1 + X_2$.

    \item If $P_1 + P_2 > P_3$ and both $P_1 \ge P_3$ and $P_2 \ge P_3$, the first and second transmitters transmit with powers $\lambda P_3$ and $(1-\lambda)P_3$, respectively, for some $\lambda \in [0,1]$. Their combined power does not exceed $P_3$, allowing cancellation by the third transmitter.

    \item If $P_1 + P_2 > P_3$, $P_1 < P_3$, and $P_2 \ge P_3$, then the first transmitter uses power $\lambda P_1$ and the second uses $P_3 - \lambda P_1$, for $\lambda \in [0,1]$, ensuring the total remains within $P_3$.

    \item If $P_1 + P_2 > P_3$, $P_1 \ge P_3$, and $P_2 < P_3$, then the second transmitter uses $\lambda P_2$ and the first uses $P_3 - \lambda P_2$, for $\lambda \in [0,1]$.

    \item If $P_1 + P_2 > P_3$, $P_1 < P_3$, and $P_2 < P_3$, the first transmitter uses power $\lambda(P_3 - P_2) + (1 - \lambda)P_1$, and the second uses $(1 - \lambda)(P_3 - P_1) + \lambda P_2$, for $\lambda \in [0,1]$.
\end{enumerate}
In each case, the power allocations are chosen such that the sum of the signals from the first and second transmitters can be exactly canceled by the third transmitter, satisfying the covertness constraint. Since the logarithm function is monotonically increasing in the input powers, it suffices to select the maximum allowable power for each transmitter when determining the boundary of the achievable region. This leads to the choices $\tilde{P}_1\triangleq\min\br{P_1,P_3}$, $\tilde{P}_2\triangleq\min\br{P_2,P_3}$, and $\tilde{P}_3\triangleq\min\br{P_1+P_2,P_3}$. Finally, by setting $X_1 \sim \calN(0, \tilde{P}_1)$, $X_2 \sim \calN(0, \tilde{P}_2)$, and $X_3 \sim \calN(0, \tilde{P}_3)$ in  \eqref{eq:MAC_Helper_Reliability}, we obtain the rate region $\calC_{\textup{CC-GMAC}}$.

We now prove the converse part of Theorem~\ref{thm:Capacity_Gaussian}. Let $\tilde{P}_1\triangleq\bbE[X_1^2]\le P_1$, $\tilde{P}_2\triangleq\bbE[X_2^2]\le P_2$, and $\tilde{P}_3\triangleq\bbE[X_3^2]\le P_3$. We have,
    \begin{subequations}\label{eq:Gaussian_1}
    \begin{align}
        nR_1&\mathop\le\limits^{(a)}nI(X_1,X_3;B\lvert X_2)+n\delta\nonumber\\
        &\mathop=\limits^{(b)}n\sbr{\dent(B\lvert X_2)-\dent(B\lvert X_{[3]})}+n\delta\nonumber\\
        &=n\sbr{\dent(X_1+X_2+hX_3+N_B\lvert X_2)-\dent(N_B)}+n\delta\nonumber\\
        &=n\sbr{\dent(X_1+hX_3+N_B\lvert X_2)-\dent(N_B)}+n\delta\nonumber\\
        &\mathop\le\limits^{(c)}n\left[\frac{1}{2}\log(2\pi e)^2\pr{\tilde{P}_2\pr{\tilde{P}_1+2h\bbE[X_1X_3]+h^2\tilde{P}_3+\sigma_B^2}}\right.\nonumber\\
        &\qquad-h^2\bbE^2[X_2X_3]-\frac{1}{2}\log2\pi e\pr{\tilde{P}_2}\nonumber\\
        &\left.\qquad-\frac{1}{2}\log2\pi e\sigma_B^2\right]+n\delta,\label{eq:Gaussian_R1_1}
    \end{align}where $(a)$ is proved through steps analogous to those used to prove \eqref{eq:Upper_R_cq}; $(b)$ follows since $\dent(\cdot)$ denotes the differential entropy; and $(c)$ follows from the maximum differential entropy lemma \cite{ElGamalKim}. Similarly, we have,
    \begin{align}
        &nR_2\le n\left[\frac{1}{2}\log(2\pi e)^2\pr{\tilde{P}_1\pr{\tilde{P}_2+2h\bbE[X_2X_3]+h^2\tilde{P}_3+\sigma_B^2}}\right.\nonumber\\
        &\left.\qquad-h^2\bbE^2[X_1X_3]-\frac{1}{2}\log2\pi e\pr{\tilde{P}_1}-\frac{1}{2}\log2\pi e\sigma_B^2\right]+n\delta,\label{eq:Gaussian_R2_1}\\
        &n(R_1+R_2)\le n\left[\frac{1}{2}\log2\pi e\left(\tilde{P}_1+\tilde{P}_2+h^2\tilde{P}_3+2h\bbE[X_1X_3]\right.\right.\nonumber\\
        &\left.+2h\bbE[X_2X_3]+\sigma_B^2\Big)-\frac{1}{2}\log2\pi e\sigma_B^2\right]+n\delta.\label{eq:Gaussian_R1R2_1}
    \end{align}
    \end{subequations}
    Similar to the proof in \eqref{eq:Covertness_C_Degraded} one can show that $\bbD\pr{P_{Z^n}\lVert q_0^\ot}$ leads to $p_Z=q_0$. The covertness constraint $p_Z=q_0$, where $q_0\sim\calN(0,\sigma_E^2)$, is satisfied if $\tilde{P}_1+\tilde{P}_2+\tilde{P}_3+2\bbE[X_1X_3]+2\bbE[X_2X_3]+\sigma_E^2=\sigma_E^2$, which is equivalent to
    \begin{align}
        &\bbE[X_1X_3]+\bbE[X_2X_3]=-\frac{\tilde{P}_1+\tilde{P}_2+\tilde{P}_3}{2}\nonumber\\
        &\Leftrightarrow\bbE[YX_3]=-\frac{\tilde{P}_1+\tilde{P}_2+\tilde{P}_3}{2},\nonumber
    \end{align}where $Y\triangleq X_1+X_2$, and since $X_1$ and $X_2$ are independent and normally distributed, $Y$ is also normally distributed with variance $\tilde{P}_1+\tilde{P}_2$. Therefore, the correlation coefficient between $Y$ and $X_3$ is equal to
    \begin{align}
        \rho_{YX_3}&=\frac{\bbE[YX_3]}{\sqrt{(\tilde{P}_1+\tilde{P}_2)\tilde{P}_3}}\nonumber\\
        &=-\frac{\tilde{P}_1+\tilde{P}_2+\tilde{P}_3}{2\sqrt{(\tilde{P}_1+\tilde{P}_2)\tilde{P}_3}}.\nonumber
    \end{align}Note that the correlation coefficient satisfies $-1\le\rho_{YX_3}\le1$, and since $\tilde{P}_i\ge0$, $i\in[3]$, we have $\rho_{YX_3}\le0$. Therefore, we must have $-1\le\rho_{YX_3}$, which is equivalent to $\tilde{P}_1+\tilde{P}_2+\tilde{P}_3\le2\sqrt{(\tilde{P}_1+\tilde{P}_2)\tilde{P}_3}\Leftrightarrow\pr{\sqrt{\tilde{P}_1+\tilde{P}_2}-\sqrt{\tilde{P}_3}}^2\le0$. This inequality is satisfied only when $\tilde{P}_1+\tilde{P}_2=\tilde{P}_3$, which means that we have $\rho_{YX_3}=-1$ and therefore $X_3=-Y=-(X_1+X_2)$. Therefore, substituting $X_3=-(X_1+X_2)$ and $\tilde{P}_3=\tilde{P}_1+\tilde{P}_2$ in \eqref{eq:Gaussian_R1_1} leads to
    \begin{subequations}
        \begin{align}
            R_1&\le \frac{1}{2}\log(2\pi e)^2\pr{\tilde{P}_2\pr{\tilde{P}_1-2h\tilde{P}_1+h^2\tilde{P}_3+\sigma_B^2}-h^2\tilde{P}_2^2}\nonumber\\
            &\quad-\frac{1}{2}\log2\pi e\pr{\tilde{P}_2}-\frac{1}{2}\log2\pi e\sigma_B^2+\delta\nonumber\\
            &=\frac{1}{2}\log2\pi e\pr{\tilde{P}_1-2h\tilde{P}_1+h^2\tilde{P}_3+\sigma_B^2-h^2\tilde{P}_2}\nonumber\\
            &\quad-\frac{1}{2}\log2\pi e\sigma_B^2+\delta\nonumber\\
            &=\frac{1}{2}\log\pr{1+\frac{(1-2h)\tilde{P}_1-h^2\tilde{P}_2+h^2\tilde{P}_3}{\sigma_B^2}}+\delta\nonumber\\
            &=\frac{1}{2}\log\pr{1+\frac{(1-2h)\tilde{P}_1+h^2\tilde{P}_1}{\sigma_B^2}}+\delta\nonumber\\
            &=\frac{1}{2}\log\pr{1+\frac{(1-h)^2\tilde{P}_1}{\sigma_B^2}}+\delta.\nonumber
        \end{align}
    \end{subequations}
    \setcounter{equation}{21}
Similarly, substituting $X_3=-(X_1+X_2)$ and $\tilde{P}_3=\tilde{P}_1+\tilde{P}_2$ in \eqref{eq:Gaussian_R2_1} and \eqref{eq:Gaussian_R1R2_1} leads to 
    \begin{align}
&\bigcup_{\substack{\tilde{P}_1\le P_1,\\ \tilde{P}_2\le P_2,\\\tilde{P}_3\le P_3,\\ \tilde{P}_1+\tilde{P}_2=\tilde{P}_3}}
  \left\{ \begin{array}{rl}
  &\hspace{-4mm}(R_1,R_2):\\
&\hspace{-4mm}R_1<\frac{1}{2}\log\pr{1+\frac{(1-h)^2\tilde{P}_1}{\sigma_B^2}}\\
&\hspace{-4mm}R_2<\frac{1}{2}\log\pr{1+\frac{(1-h)^2\tilde{P}_2}{\sigma_B^2}}\\
&\hspace{-4mm}R_1+R_2<\frac{1}{2}\log\pr{1+\frac{(1-h)^2\pr{\tilde{P}_1+\tilde{P}_2}}{\sigma_B^2}}\\
	\end{array}\label{eq:Outer_Gaussian_Union}
\right\}.
\end{align}
\begin{figure}[t!]
\centering
\includegraphics[width=8.0cm]{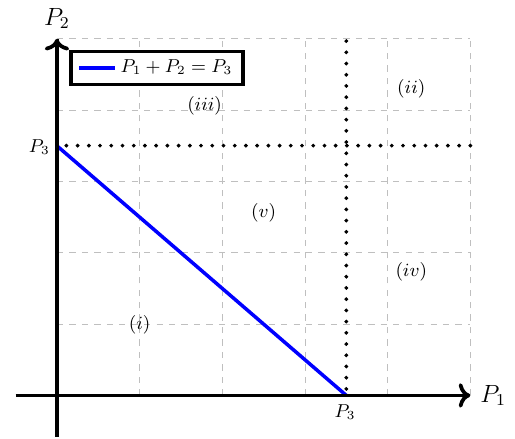}
\caption{Different possible cases for $P_1$ and $P_2$ when $P_3$ is fixed.
}
\label{fig:Gaussian_Cases}
\vspace{-0.7cm}
\end{figure}
The union in \eqref{eq:Outer_Gaussian_Union} is taken over all $\pr{\tilde{P}_1,\tilde{P}_2,\tilde{P}_3}$ such that $0\le\tilde{P}_1\le P_1,0\le\tilde{P}_2\le P_2,0\le\tilde{P}_3\le P_3$, and $\tilde{P}_1+\tilde{P}_2=\tilde{P}_3$. Eliminating $\tilde{P}_3$ gives the feasible power set 
\begin{align}
    \calZ\triangleq\br{\!\pr{\tilde{P}_1,\tilde{P}_2}\!: 0\le\tilde{P}_1\le P_1,0\le\tilde{P}_2\le P_2,\tilde{P}_1\!+\!\tilde{P}_2\le P_3\!}\!.\nonumber
\end{align}Therefore, the \ac{MAC} region in \eqref{eq:Outer_Gaussian_Union} is
\begin{align}
    \bigcup_{\pr{\tilde{P}_1,\tilde{P}_2}\in\calZ}\left\{ \begin{array}{rl}
  &\hspace{-4mm}(R_1,R_2):\\
&\hspace{-4mm}R_1<\frac{1}{2}\log\pr{1+\frac{(1-h)^2\tilde{P}_1}{\sigma_B^2}}\\
&\hspace{-4mm}R_2<\frac{1}{2}\log\pr{1+\frac{(1-h)^2\tilde{P}_2}{\sigma_B^2}}\\
&\hspace{-4mm}R_1+R_2<\frac{1}{2}\log\pr{1+\frac{(1-h)^2\pr{\tilde{P}_1+\tilde{P}_2}}{\sigma_B^2}}\\
	\end{array}\right\}.\label{eq:Outer_Gaussian_Union_2}
\end{align}By monotonicity of the logarithm functions in $\pr{\tilde{P}_1,\tilde{P}_2}$, the region in \eqref{eq:Outer_Gaussian_Union_2} is upper bounded by
\begin{align}
\left\{ \begin{array}{rl}
&\hspace{-4mm}(R_1,R_2):\\
&\hspace{-4mm}R_1<\frac{1}{2}\log\pr{1+\frac{(1-h)^2\min\br{P_1,P_3}}{\sigma_B^2}}\nonumber\\ &\hspace{-4mm}R_2<\frac{1}{2}\log\pr{1+\frac{(1-h)^2\min\br{P_2,P_3}}{\sigma_B^2}}\nonumber\\ &\hspace{-4mm}R_1+R_2<\frac{1}{2}\log\pr{1+\frac{(1-h)^2\min\br{P_1+P_2,P_3}}{\sigma_B^2}}
    \end{array}\right\}.\nonumber
\end{align} 
\end{proof}
\vspace{-1cm}
\subsection{Bosonic \texorpdfstring{\ac{MAC}}{MAC} with a Helper}
\label{sec:Bosonic}
We now consider the quantum counterpart of the classical Gaussian case of the \ac{MAC} in Fig.~\ref{fig:MAC_Helper}, which is the single-mode bosonic \ac{MAC} with a helper.  
We first briefly introduce notions related to the bosonic channels. For a comprehensive description of continuous-variable bosonic systems, we refer the readers to \cite{RevModPhys12}. 
Operators acting on quantum states are denoted using a hat notation, e.g., $\hat{a}$, $\hat{b}$, and $\hat{c}$. The single-mode Hilbert space is spanned by the Fock basis $\br{\ket{n}}_{n=0}^\infty$, where each $\ket{n}$ is an eigenstate of the number operator $\hat{n} = \hat{a}^\dagger \hat{a}$. Here, $\hat{a}$ denotes the bosonic field annihilation operator. In particular, $\ket{0}$ represents the vacuum state of the field. The creation operator $\hat{a}^\dagger$ raises the excitation level according to $\hat{a}^\dagger \ket{n} = \sqrt{n+1} \ket{n+1}$ for $n \ge 0$, while the annihilation operator $\hat{a}$ lowers it via $\hat{a} \ket{n} = \sqrt{n} \ket{n-1}$. 
A coherent state $\ket{\alpha}$, where $\alpha \in \bbC$, corresponds to an oscillation of the electromagnetic field, and it is obtained by applying the displacement operator $D(\alpha)$ to the vacuum state, i.e., $\ket{\alpha} = D(\alpha) \ket{0}$, where $D(\alpha) = \exp\left(\alpha \hat{a}^\dagger - \alpha^* \hat{a}\right)$.
This operation resembles the action of the creation operator but in a more physically meaningful, phase-space-shifting manner. 
The average density operator for a thermal state with average photon number $N\ge0$, defined as
\begin{align}
\tau(N)&\triangleq\int\limits_\bbC d^2\alpha\frac{e^{-\frac{\abs{\alpha}^2}{N}}}{\pi N}\den{\alpha}{\alpha}\nonumber\\
&=\frac{1}{N+1}\sum_{n=0}^\infty\pr{\frac{N}{N+1}}^n\den{n}{n}.\nonumber
\end{align}
\begin{figure}[t!]
\centering
\includegraphics[width=8.8cm]{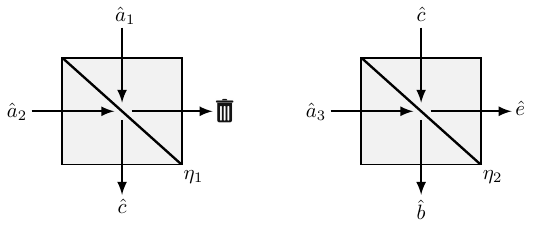}
\caption{The beam splitter relation of the single-mode bosonic \ac{MAC} with a helper. 
The Transmitter~$t$, for $t \in [2]$, encodes message $M_t$ into a coherent state and transmits it over the channel. With knowledge of the coherent states transmitted by the first two transmitters, the third transmitter encodes the message pair $(M_1, M_2)$ into a coherent state and transmits it over the channel. The communication channel $\calN_{A_{[3]} \to BE}$ consists of two beam splitters: the first, on the left, with transmissivity $\eta_1=\frac{1}{2}$, combines the states transmitted by the first two transmitters; the second, on the right, with transmissivity $\eta_2=\frac{1}{3}$, mixes the output of the first beam splitter with the state transmitted by the third transmitter.
}
\label{fig:Beam_Splitter}
\vspace{-0.7cm}
\end{figure}
A single-mode two-user bosonic \ac{MAC} is modeled by a beam splitter \cite{MAC_Bosonic_05}. Here we consider a bosonic \ac{MAC} with a helper, where the channel inputs are three electromagnetic field modes with annihilation operators $\hat{a}_1$, $\hat{a}_2$, and $\hat{a}_3$, and the outputs observed by the legitimate receiver and the warden are modes with annihilation operators $\hat{b}$ and $\hat{e}$, respectively. This is modeled by two beam splitters, as illustrated in Fig.~\ref{fig:Beam_Splitter}. As an illustrative example, we set the transmissivities of the beam splitters as 
$\eta_1 = \tfrac{1}{2}$ and $\eta_2 = \tfrac{1}{3}$ that depend on the length of the optical fiber and its absorption length \cite{Eisert07}. The input-output relation in the Heisenberg picture \cite{Holevo01} is given by:
\begin{subequations}\label{eq:Bosonic_Channel_Law}
\begin{align}
    \hat{c}&\triangleq\sqrt{\eta_1}\hat{a}_1+\sqrt{1-\eta_1}\hat{a}_2\nonumber\\
    &=\sqrt{\frac{1}{2}}\hat{a}_1+\sqrt{\frac{1}{2}}\hat{a}_2,\label{eq:Bosonic_intermediate}\\
    \hat{b}&\triangleq\sqrt{\eta_2}\hat{c}+\sqrt{1-\eta_2}\hat{a}_3\nonumber\\
    &=\sqrt{\eta_2\eta_1}\hat{a}_1+\sqrt{\eta_2(1-\eta_1)}\hat{a}_2+\sqrt{1-\eta_2}\hat{a}_3\nonumber\\
    &=\sqrt{\frac{1}{6}}\hat{a}_1+\sqrt{\frac{1}{6}}\hat{a}_2+\sqrt{\frac{2}{3}}\hat{a}_3,\label{eq:Bosonic_receiver}\\
    \hat{e}&\triangleq\sqrt{1-\eta_2}\hat{c}-\sqrt{\eta_2}\hat{a}_3\nonumber\\
    &=\sqrt{(1-\eta_2)\eta_1}\hat{a}_1+\sqrt{(1-\eta_2)(1-\eta_1)}\hat{a}_2-\sqrt{\eta_2}\hat{a}_3\nonumber\\
    &=\sqrt{\frac{1}{3}}\hat{a}_1+\sqrt{\frac{1}{3}}\hat{a}_2-\sqrt{\frac{1}{3}}\hat{a}_3,\label{eq:Bosonic_warden}
\end{align}
We assume that the transmitters use coherent state protocols with input constraints. The input state of the Transmitter~$t$, for $t\in[3]$, is a coherent state $\ket{X_t}$, where $X_t\in\bbC$, such that the codewords satisfy,
\begin{align}
    \frac{1}{n}\sum_{i=1}^n\abs{x_{1,i}}^2&\le N_1,\,\,\frac{1}{n}\sum_{i=1}^n\abs{x_{2,i}}^2\le N_2,\,\,
    \frac{1}{n}\sum_{i=1}^n\abs{x_{3,i}}^2\le N_3.\label{eq:Power_Const_Bosonic}
\end{align}The corresponding received states at the legitimate receiver and the warden are the coherent states $\ket{\sqrt{\frac{1}{6}}X_1+\sqrt{\frac{1}{6}}X_2+\sqrt{\frac{4}{6}}X_3}$ and $\ket{\sqrt{\frac{1}{3}}\pr{X_1+X_2-X_3}}$, respectively.
\end{subequations}
Also, let the innocent state for the transmitters be $\ket{x_{t,0}}=\ket{0}$, for $t\in[3]$, therefore, the state induced at the output of the warden in the no communication mode is $\rho_0=\den{0}{0}$. We assume that the covertness constraint is as defined in \eqref{eq:CC_CSK_Metric_Asymp}, and the capacity is defined as the closure of all achievable rates and is denoted by the $\calC_{\textup{CC-BMAC}}$.
\begin{theorem}
\label{eq:thm_Bosonicn}
    The covert capacity $\calC_{\textup{CC-BMAC}}$ of the three-user bosonic \ac{MAC} described in \eqref{eq:Bosonic_Channel_Law} is lower bounded by 
\begin{align}
&\calC_{\textup{CC-BMAC}}\supseteq
  \left\{ \begin{array}{rl}
  &\hspace{-4mm}(R_1,R_2):\\
&\hspace{-4mm}R_1<g\pr{\frac{3}{2}\tilde{N}_1}\\
&\hspace{-4mm}R_2<g\pr{\frac{3}{2}\tilde{N}_2}\\
&\hspace{-4mm}R_1+R_2<g\pr{\frac{3}{2}\tilde{N}_3}\\
	\end{array}\nonumber
\right\},
\end{align}where $g(N)$ is the Shannon entropy
of the Bose-Einstein probability distribution, given by 
\begin{align}
    g(N)\triangleq\begin{cases}
(N+1)\log(N+1)+N\log(N),&\text{if}\Squad N>0,\\
0,&\text{if}\Squad N=0,\\
\end{cases}\nonumber
\end{align}and $\tilde{N}_1\triangleq\min\br{N_1,N_3}$, $\tilde{N}_2\triangleq\min\br{N_2,N_3}$, and $\tilde{N}_3\triangleq\min\br{N_1+N_2,N_3}$.
\end{theorem}

\begin{proof}
Using the discretization and limiting argument developed by Guha~et al.~\cite{Guha07}, we extend the finite-dimensional result in Corollary~\ref{cor:qMAC_Helper} to continuous-variable bosonic systems with infinite-dimensional Hilbert spaces. To prove Theorem~\ref{eq:thm_Bosonicn}, we set the random variables $X_t$, for $t \in [3]$, as circularly symmetric Gaussian distributed with zero mean and variance $\tilde{N}_t \le N_t$.  
To satisfy the covertness constraint in \eqref{eq:thm_qMAC_Helper_S}, the third transmitter chooses $X_3=X_1+X_2$, canceling the effects of the states transmitted by the first and second transmitters at the warden. Therefore, the state received by the legitimate receiver is a coherent state given by $\ket{\sqrt{\frac{3}{2}}(X_1 + X_2)}$. For a given $N_3$, the average photon number allocations for $\ket{X_1}$ and $\ket{X_2}$ are chosen based on the relationship between $N_1$, $N_2$, and $N_3$, which is the same as that in the achievability proof of Theorem~\ref{thm:Capacity_Gaussian} and shown in Fig.~\ref{fig:Gaussian_Cases}. 
This leads to the choices $\tilde{N}_1\triangleq\min\br{N_1,N_3}$, $\tilde{N}_2\triangleq\min\br{N_2,N_3}$, and $\tilde{N}_3\triangleq\min\br{N_1+N_2,N_3}$. Using these average photon numbers, we compute the mutual information terms in \eqref{eq:MAC_Helper_Reliability}, which leads to the achievable rate region presented in Theorem~\ref{eq:thm_Bosonicn}.

Note that, since the state induced at the warden's output in the no-communication mode is the vacuum state, which is not full rank (equivalently, its minimum eigenvalue is zero), step~$(d)$ in the single-letterization of the covertness constraint~\eqref{eq:Covertness_C_Degraded} no longer holds. Consequently, the converse proof for the Gaussian channel, presented in Section~\ref{sec:Example_Gaussian}, cannot be directly extended to the bosonic channel considered in this section.
\end{proof}

\section{Conclusions}
\label{sec:Conclusion}
This paper investigates covert communication at positive rates over \acp{MAC} with general message sets. We establish both one-shot and asymptotic achievable rate regions, and demonstrate that our achievability scheme is tight in the classical-quantum setting. To illustrate the applicability of our framework, we present several examples. First, we consider a finite-dimensional \ac{MAC} with a helper and derive an achievable covert rate region. We then study a classical Gaussian \ac{MAC} with a helper and characterize its covert capacity region. Finally, we extend the analysis to a single-mode bosonic \ac{MAC} with a helper and establish the achievability of positive covert rates in this setting.

\begin{appendices}

\section{Proof of Theorem~\ref{thm:Achievable_One_Shot}}
\label{proof:thm:Achievable_One_Shot} 
We show that for each $\delta_1,\dots,\delta_S,\epsilon>0$, there exists a $\pr{2^{R_1-\delta_1},\dots,2^{R_S-\delta_S},1,\epsilon}$ code for the quantum channel $\calN_{A_{[T]}\to BE}$ that satisfies both the reliability constraint and the covertness constraint.

\subsection{Random Codebook Generation}
\label{sec:Codebook_Cons} 
We start generating the code from the leaves of the corresponding message graph. 
For each leaf $v_t\in\bbL$ with private message $M(t)$, let $C_{U,t}\triangleq\br{U_t(m(t))}_{m(t)\in\calM(t)}$, where $t\in[T]$ and  $\calM(t)\triangleq\brk{2^{R(t)}}$, be a random codebook generated \ac{iid} according to $p_{U_t}$. Also, for each $u_t\in C_{U,t}$, let  $C_{X,t}\triangleq\br{X_t(m(t))}$ be a random sequence generated \ac{iid} according to $p_{X_t\lvert U_t}(\cdot\lvert u_t)$. 
A realization of $C_{U,t}$ and $C_{X,t}$ are denoted by $\calC_{U,t}\triangleq\br{u_t(m(t))}_{m(t)\in\calM(t)}$ and $\calC_{X,t}\triangleq\br{x_t(m(t))}$, respectively. Henceforth, for all $t \in [T]$, we set $U_t = \emptyset$ if $v_t \notin \bbM$, and define $\tilde{X}_t \triangleq (U_t, X_t)$.

For each parent of the leaves, we generate a separate codebook. For instance, for a parent $v_t$, for each combination of the sequences of its descendants codebooks, we generate a codebook $C_{U,t}\triangleq\br{U_t(m(t))}_{m(t)\in\calM(t)}$ and $C_{X,t}\triangleq\br{X_t(m(t))}$, where $\calM(t)\triangleq\brk{2^{R(t)}}$, and the codewords are generated \ac{iid} according to $p_{U_t\lvert U_{\bbD_t}X_{\bbD_t}}$ and $p_{X_t\lvert U_tU_{\bbD_t}X_{\bbD_t}}$, respectively. We proceed in the same manner until we reach $v_0$. Since $v_0$ does not correspond to any transmitter in the system, it has no channel input. Also, note that if a node $v_t\notin\bbM$, we have $C_{U,t}=\emptyset$ and $C_{X,t}\triangleq\br{X_t(m(t))}_{m(t)\in\calM(t)}$ is generated \ac{iid} according to $p_{X_t\lvert U_{\bbD_t}X_{\bbD_t}}$. Now let $C\triangleq\br{C_{U,t},C_{X,t}}_{t\in[T]}$ be the set of all the generated random codebooks and $\calC\triangleq\br{\calC_{U,t},\calC_{X,t}}_{t\in[T]}$ be a realization of~$C$.

\subsection{Encoding}To send the message $m_{\calI_t}$, the Transmitter~$t$ selects $u_t(m_{\calI_t})$ and then $x_t(m_{\calI_t})$, and prepares $\rho_{A_1}=\theta^{x_{t}}$ and transmits it over the channel. This encoding scheme induces the following state at the output of the warden:
\begin{align}
    &\tau_{E\lvert\calC}\triangleq\frac{1}{2^{\sum_{s=1}^SR_s}}\sum_{m_{[S]}}\rho_{E\lvert\tilde{x}_1\pr{m_{\calI_1}},\tilde{x}_2\pr{m_{\calI_2}},\dots,\tilde{x}_T\pr{m_{\calI_T}}}.\label{eq:Joint_Dist}
\end{align}
\subsection{Pinching}
\label{sec:Pinching}
Our decoding scheme is based on the simultaneous pinching method \cite{QIT_Hayashi}. Consider states $\sigma\in\calD(\calH)$ and $\rho\in\calD(\calH)$, and let $\sigma=\sum_{i}\lambda_i\den{z_i}{z_i}$, where $\lambda_j\ne\lambda_k$ for $j\ne k$, be the spectral decomposition of the state $\sigma$. The pinching operation of the state $\rho$ \ac{wrt} the spectral decomposition of the state $\sigma$ is defined as $\Delta_\sigma(\rho)\triangleq\sum_i\den{z_i}{z_i}\rho\den{z_i}{z_i}=\sum_i\bra{z_i}\rho\ket{z_i}\den{z_i}{z_i}$. Note that the state $\sigma$ and the state $\Delta_\sigma(\rho)$ commute. Now consider the following classical-quantum states
\begin{subequations}
\begin{align}
    &\rho_{U_{[T]}X_{[T]}B}\triangleq\sum_{u_{[T]},x_{[T]}}\prod_{t=1}^Tp_{U_tX_t\lvert U_{\bbD_t}X_{\bbD_t}}(u_t,x_t\lvert u_{\bbD_t},x_{\bbD_t})\nonumber\\
    &\times\den{u_1}{u_1}_{U_1}\otimes\den{x_1}{x_1}_{X_1}\otimes\cdots\otimes\den{u_T}{u_T}_{U_T}\nonumber\\
    &\otimes\den{x_T}{x_T}_{X_T}\otimes\rho_{B\lvert u_{[T]},x_{[T]}},\label{eq:general_joint_State}\\
    &\rho_{\pr{U_\calT,X_\calT}-\pr{U_{\calT^c},X_{\calT^c}}-B}\triangleq\nonumber\\
    &\sum_{u_{[T]},x_{[T]}}\prod_{t=1}^Tp_{U_tX_t\lvert U_{\bbD_t}X_{\bbD_t}}(u_t,x_t\lvert u_{\bbD_t},x_{\bbD_t})\den{u_{\bar{v}_1}}{u_{\bar{v}_1}}_{U_{\bar{v}_1}}\nonumber\\
    &\otimes\den{x_{\bar{v}_1}}{x_{\bar{v}_1}}_{X_{\bar{v}_1}}\otimes\cdots\otimes\den{u_{\bar{v}_{|\calT^c|}}}{u_{\bar{v}_{|\calT^c|}}}_{U_{\bar{v}_{|\calT^c|}}}\nonumber\\
    &\otimes\den{x_{\bar{v}_{|\calT^c|}}}{x_{\bar{v}_{|\calT^c|}}}_{X_{\bar{v}_{|\calT^c|}}}\otimes\rho_{U_\calT X_\calT\lvert u_{\calT^c},x_{\calT^c}}\otimes\rho_{B\lvert u_{\calT^c},x_{\calT^c}},\label{eq:PDSG_joint_State2}
\end{align}where, $\calT\subset[T]$ is a set of vertex indices and $\calT^c$ is the complement of $\calT$ \ac{wrt} $[T]$; $(u_{v_i},x_{v_i})_{i\in\abs{\calT^c}}$ are the auxiliary \ac{RV} and the channel input of the transmitter associated with the vertex $v_i$; and $\rho_{U_{\calT}X_{\calT}\lvert u_{\calT^c},x_{\calT^c}}$ and $\rho_{B\lvert u_{\calT^c},x_{\calT^c}}$ represent the appropriate marginals of the state $\rho_{U_{[T]}X_{[T]}B}$ given in \eqref{eq:general_joint_State}. 
\end{subequations}
Let $\Delta_B$ denote the pinching operation \ac{wrt} the spectral decomposition of the state $\rho_B$. 
Now, we define $\Delta_\calT$, for $\calT\subset[T]$, as 
\begin{align}
\Delta_\calT(\rho)&\triangleq\sum_{\tilde{x}_\calT}\den{\tilde{x}_\calT}{\tilde{x}_\calT}\otimes\Delta_{B\lvert\tilde{x}_\calT}\pr{\bra{\tilde{x}_\calT}\rho\ket{\tilde{x}_\calT}},\nonumber
\end{align}
where $\ket{\tilde{x}_\calT}\triangleq\bigotimes\limits_{j\in\calT}\ket{\tilde{x}_j}$,  $\ket{\tilde{x}_j}=\ket{u_j}\otimes\ket{x_j}$, and $\Delta_{B\lvert\tilde{x}_\calT}$ is the pinching map \ac{wrt} the spectral decomposition of the state $\Delta_B\pr{\rho_{B|\tilde{x}_\calT}}$. 
Now, we define the following constants,
\begin{subequations}\label{eq:gs}
    \begin{align}
        \mu_B&\triangleq\text{number of distinct eigenvalues of } \rho_B,\label{eq:g2}\\
        \mu_{B,\calT}&\triangleq\text{maximum number of distinct eigenvalues of } \nonumber\\
        &\qquad\left\{\Delta_B\pr{\rho_{B\lvert \tilde{x}_\calT}}\right\}_{\tilde{x}_\calT},\label{eq:g3}
    \end{align}where $\calT\subset[T]$ and the maximization in \eqref{eq:g3} is over $\tilde{x}_\calT$.
    
    Also, let $\Delta_E$ denote the pinching operation \ac{wrt} the spectral decomposition of the state $\rho_E$, and $\mathfrak{B} \triangleq \pr{\calS_1, \calS_2, \dots, \calS_{\card{\mathfrak{B}}}}$ be an ordered collection of all non-empty strict subsets of $[S]$, arranged such that for any $1 \le i < j \le \card{\mathfrak{B}}$, we have $\card{\calS_i} \le \card{\calS_j}$. Define, for each subset $\calS \subset [S]$, the set $\tilde{X}(\calS)$ as the collection of channel inputs that depend only on the messages whose indices belong to $\calS$. Also, define $\Delta_{E\mid\tilde{X}\pr{\calS_{i-1}}}\Big(\rho_{E\big\lvert\tilde{X}(\calS_i)}\Big)$ as the pinching map \ac{wrt} the spectral decomposition of the operator $\Delta_{E\mid\tilde{X}\pr{\calS_{i-2}}}\Big(\rho_{E\big\lvert\tilde{X}(\calS_{i-1})}\Big)$, with the convention $\tilde{X}(\calS_0)=\emptyset$ and therefore $\Delta_{E\mid\emptyset}=\Delta_E$. Also, for $i\in[\card{\mathfrak{B}}]$, define 
    \begin{align}
    \mu_E&\triangleq\text{number of distinct eigenvalues of } \rho_E,\label{eq:g1}\\
        \mu_{E,i}&\triangleq \text{maximum number of the distinct eigenvalues of the} \nonumber\\
        &\quad\text{operator}\,\,\br{\Delta_{E\mid\tilde{x}(\calS_{i-2})}\pr{\rho_{E\mid\tilde{x}(\calS_{i-1})}}}_{\pr{\tilde{x}(\calS_1),\dots,\tilde{x}(\calS_{i-1})}}.
    \end{align}
    This hierarchical pinching construction ensures that all resulting operators remain block-diagonal in the eigenbasis of $\rho_E$, with successive refinements within each block. Consequently, all such operators commute.
\end{subequations}

\subsection{Decoding and Error Probability Analysis}
For any two Hermitian matrices $A$ and $B$, we define the projection $\{A \succeq B\}$ as $\sum_{i:\lambda_i \geq 0} P_i$, where $A - B$ has the spectral decomposition $\sum_i \lambda_i P_i$, with $\lambda_i$ being the eigenvalues and $P_i$ being the projection onto the eigenspace corresponding to $\lambda_i$. Now, we define the following projection operators for each non-empty strict subset $\calS\subset[S]$ of the message indices, 
\begin{subequations}\label{eq:Decoding_Projections}
\begin{align}
    &\Pi_{U_{[T]}X_{[T]}B}^{([S])}\triangleq\br{\Delta_B\pr{\rho_{U_{[T]}X_{[T]}B}}\succeq2^{\sum_{s=1}^SR_s}\rho_{U_{[T]}X_{[T]}}\otimes\rho_B},\label{eq:Decoding_Projection_1}\\
    &\Pi_{U_{[T]}X_{[T]}B}^{(\calS)}\triangleq\Bigg\{\Delta_{\calT_\calS^c}\pr{\rho_{U_{[T]}X_{[T]}B}}\succeq\nonumber\\
    &\left.\qquad2^{\sum_{s\in\calS}R_s}\Delta_B\pr{\rho_{\pr{U_{\calT_\calS},X_{\calT_\calS}}-\pr{U_{\calT_\calS^c},X_{\calT_\calS^c}}-B}}\right\},\label{eq:Decoding_Projection_2}
\end{align}
\end{subequations}where, for each non-empty strict subset $\calS \subset [S]$, the set $\calT_\calS$ consists of the vertex indices that have access to at least one message $m_s$ with $s \in \calS$. The complement of this set is denoted by $\calT_\calS^c$. For example, for the \ac{MAC} depicted in Fig.~\ref{fig:4M6T_Example}, consider $\calS=\br{4,5}$ as a set of message indices. Since each transmitter in the set $\br{4,5,6,7}$ has access to at least one message in $\calS$, we have $\calT_\calS=\br{4,5,6,7}$ and $\calT_\calS^c=\br{1,2,3,8}$. 
Also, define $\Pi_{U_{[T]}X_{[T]}B} \triangleq \Pi_{U_{[T]}X_{[T]}B}^{([S])}\prod\limits_{\calS\in\mathfrak{S}}\Pi_{U_{[T]}X_{[T]}B}^{(\calS)}$, where $\prod\limits_{\calS\in\mathfrak{S}}\Pi_{U_{[T]}X_{[T]}B}^{(\calS)}\triangleq\Pi_{U_{[T]}X_{[T]}B}^{(\calS_1)}\Pi_{U_{[T]}X_{[T]}B}^{(\calS_2)}\dots\Pi_{U_{[T]}X_{[T]}B}^{(\calS_{\abs{\mathfrak{S}}})}$ and $\mathfrak{S}\triangleq\br{\calS_1,\calS_2,\cdots,\calS_{\abs{\mathfrak{S}}}}$ is the set of all the non-empty strict subsets of $[S]$. 
Note that the projection $\Pi_{U_{[T]}X_{[T]}B}^{([S])}$ and all projections in the collection $\br{\Pi_{U_{[T]}X_{[T]}B}^{(\calS)}}_{\calS\subset[S]}$ mutually commute with one another. 
Moreover, we define the following operator:
\begin{align}
    &\Upsilon\pr{m_{[S]}}\triangleq\tra_{U_{[T]}X_{[T]}}\Big[\Pi_{U_{[T]}X_{[T]}B}\big(\den{U_1(m_{\calI_1})}{U_1(m_{\calI_1})}\nonumber\\
    &\quad\otimes\den{X_1(m_{\calI_1})}{X_1(m_{\calI_1})}\otimes\den{U_2(m_{\calI_2})}{U_2(m_{\calI_2})}\nonumber\\
    &\quad\otimes\den{X_2(m_{\calI_2})}{X_2(m_{\calI_2})}\otimes\cdots\otimes\den{U_T(m_{\calI_T})}{U_T(m_{\calI_T})}\nonumber\\
    &\quad\otimes\den{X_T(m_{\calI_T})}{X_T(m_{\calI_T})}\otimes\bbI_ B\big)\Big].\label{eq:Gamma_Operator}
\end{align}To obtain a set of \ac{POVM} operators, we normalize \eqref{eq:Gamma_Operator} as follows:
\begin{align}
    &\Lambda\pr{m_{[S]}}\triangleq\nonumber\\
    &\pr{\sum_{m'_{[S]}}\Upsilon\pr{m'_{[S]}}}^{-1}\Upsilon\pr{m_{[S]}}\pr{\sum_{m'_{[S]}}\Upsilon\pr{m'_{[S]}}}^{-1}.\label{eq:POVMs}
\end{align}The receiver decodes the messages by applying the \ac{POVM} operators specified in \eqref{eq:POVMs}. The following lemma plays a crucial role in the error analysis.
\begin{lemma}
    \label{lemma:error_analysis}
    For $\alpha\in\pr{0:\frac{1}{2}}$ and for each non-empty strict subset $\calS\subset[S]$, we have
    \begin{subequations}\label{eq:Lemma_Error}
    \begin{align}
    &\tra\left[\left(\bbI-\Pi_{U_{[T]}X_{[T]}B}^{([S])}\right)\rho_{U_{[T]}X_{[T]}B}\right]\le\nonumber\\
    &\mu_{B,[T]}^\alpha2^{\alpha\sum_{s=1}^SR_s}2^{-\alpha\ubar{\D}_{1-\alpha}\pr{\rho_{U_{[T]}X_{[T]}B}\lVert\rho_{U_{[T]}X_{[T]}}\otimes\rho_B}},\label{eq:Lemma_Error_1}\\
    &\tra\left[\left(\bbI-\Pi_{U_{[T]}X_{[T]}B}^{(\calS)}\right)\rho_{U_{[T]}X_{[T]}B}\right]\le\nonumber\\
    &\mu_{B,\calT_\calS^c}^\alpha2^{\alpha\sum_{s\in\calS}R_s}2^{-\alpha\ubar{\I}_{1-\alpha}\pr{U_{\calT_\calS},X_{\calT_\calS};B\lvert U_{\calT_\calS^c},X_{\calT_\calS^c}}},\label{eq:Lemma_Error_2}\\
    &\tra\left[\Pi_{U_{[T]}X_{[T]}B}^{([S])}\left(\rho_{U_{[T]}X_{[T]}}\otimes\rho_B\right)\right]\le\nonumber\\
    &\mu_{B,[T]}^\alpha2^{-(1-\alpha)\sum_{s=1}^SR_s}2^{-\alpha\ubar{\D}_{1-\alpha}\pr{\rho_{U_{[T]}X_{[T]}B}\lVert\rho_{U_{[T]}X_{[T]}}\otimes\rho_B}},\label{eq:Lemma_Error_3}\\
    &\tra\left[\Pi_{U_{[T]}X_{[T]}B}^{(\calS)}\left(\rho_{\pr{U_{\calT_\calS},X_{\calT_\calS}}-\pr{U_{\calT_\calS^c},X_{\calT_\calS^c}}-B}\right)\right]\le\nonumber\\
    &\mu_{B,\calT_\calS^c}^\alpha2^{-(1-\alpha)\sum_{s\in\calS}R_s}2^{-\alpha\ubar{\I}_{1-\alpha}\pr{U_{\calT_\calS},X_{\calT_\calS};B\lvert U_{\calT_\calS^c},X_{\calT_\calS^c}}},\label{eq:Lemma_Error_4}
    \end{align}where $U_t = \emptyset$ if $v_t\notin\bbM$, for $t \in [T]$, and for every non-empty strict subset $\calS \subset [S]$, we define the set $\calT_\calS$ as the set of transmitter indices that have access to at least one message $m_s$ with $s \in \calS$. The set $\calT_\calS^c$ denotes the complement of $\calT_\calS$ \ac{wrt} $[T]$.
\end{subequations}
\end{lemma}
The proof of Lemma~\ref{lemma:error_analysis} is provided in Appendix~\ref{proof:lemma_error_analysis}.

\textit{Error Analysis:} 
To bound the probability of error averaged over the random choice of the codebook, it is sufficient to bound $\bbE_C\bbP\br{\hat{M}_{[S]}\ne\mathbf{1}_{[S]}\big\lvert M_{[S]}=\mathbf{1}_{[S]}}$, where $\mathbf{1}_{[S]}$ is the all-one vector of length $S$, leveraging the symmetry of the codebook construction. We have,
\begin{align}
    &\bbE_C\bbP\br{\hat{M}_{[S]}\ne\mathbf{1}_{[S]}\big\lvert M_{[S]}=\mathbf{1}_{[S]}}\nonumber\\
    &=\bbE_C\sbr{\tra\sbr{\pr{\sum_{m'_{[S]}\ne\mathbf{1}_{[S]}}\!\!\Lambda\pr{m'_{[S]}}\!}\rho_{B\lvert \tilde{X}_1(\mathbf{1}_{\calI_1}),\cdots,\tilde{X}_T(\mathbf{1}_{\calI_T})}}}\nonumber\\
    &=\sum_{m'_{[S]}\ne\mathbf{1}_{[S]}}\bbE_C\sbr{\tra\sbr{\Lambda\pr{m'_{[S]}}\rho_{B\lvert \tilde{X}_1(\mathbf{1}_{\calI_1}),\cdots,\tilde{X}_T(\mathbf{1}_{\calI_T})}}}\nonumber\\
    &\le\bbE_C\sbr{\tra\sbr{\pr{\bbI-\Lambda\pr{\mathbf{1}_{[S]}}}\rho_{B\lvert \tilde{X}_1(\mathbf{1}_{\calI_1}),\tilde{X}_2(\mathbf{1}_{\calI_2}),\cdots,\tilde{X}_T(\mathbf{1}_{\calI_T})}}}\nonumber\\
    &\le2\bbE_C\sbr{\tra\sbr{\pr{\bbI-\Upsilon\pr{\mathbf{1}_{[S]}}}\rho_{B\lvert \tilde{X}_1(\mathbf{1}_{\calI_1}),\tilde{X}_2(\mathbf{1}_{\calI_2}),\cdots,\tilde{X}_T(\mathbf{1}_{\calI_T})}}}\nonumber\\
    &+4\!\!\sum_{m'_{[S]}\ne\mathbf{1}_{[S]}}\hspace{-2mm}\bbE_C\sbr{\tra\sbr{\Upsilon\pr{m'_{[S]}}\rho_{B\lvert \tilde{X}_1(\mathbf{1}_{\calI_1}),\cdots,\tilde{X}_T(\mathbf{1}_{\calI_T})}}},\label{eq:Prob_Err_Analysis}
\end{align}where the last inequality follows from the Hayashi-Nagaoka inequality \cite{Hayashi03}. Now we bound the first term on the \ac{RHS} of \eqref{eq:Prob_Err_Analysis} as follows,
\begin{align}
    &2\bbE_C\sbr{\tra\sbr{\pr{\bbI-\Upsilon\pr{\mathbf{1}_{[S]}}}\rho_{B\lvert X_1(\mathbf{1}_{\calI_1}),X_2(\mathbf{1}_{\calI_2}),\cdots,X_T(\mathbf{1}_{\calI_T})}}}\nonumber\\
    &\mathop=\limits^{(a)}2\bbE_C\left[\tra\left[\left(\bbI-\tra_{U_{[T]}X_{[T]}}\left[\Pi_{U_{[T]}X_{[T]}B}\left(\den{U_1(\mathbf{1}_{\calI_1})}{U_1(\mathbf{1}_{\calI_1})}
    \right.\right.\right.\right.\right.\nonumber\\
    &\quad\otimes\den{X_1(\mathbf{1}_{\calI_1})}{X_1(\mathbf{1}_{\calI_1})}\otimes\den{U_2(\mathbf{1}_{\calI_2})}{U_2(\mathbf{1}_{\calI_2})}\nonumber\\
    &\quad\otimes\den{X_2(\mathbf{1}_{\calI_2})}{X_2(\mathbf{1}_{\calI_2})}\otimes\cdots\otimes\den{U_T(\mathbf{1}_{\calI_T})}{U_T(\mathbf{1}_{\calI_T})}\nonumber\\
    &\hspace{3.3mm}\left.\otimes\den{X_T(\mathbf{1}_{\calI_T})}{X_T(\mathbf{1}_{\calI_T})}\otimes\bbI_B\right)\Big]\Big)\nonumber\\
    &\left.\left.\hspace{3.0mm}\times\rho_{B\lvert \tilde{X}_1(\mathbf{1}_{\calI_1}),\tilde{X}_2(\mathbf{1}_{\calI_2}),\cdots,\tilde{X}_T(\mathbf{1}_{\calI_T})}\right]\right]\nonumber\\
    &\mathop=\limits^{(b)}2\tra\left[\left(\bbI-\Pi_{U_{[T]}X_{[T]}B}\right)\rho_{\tilde{X}_1(\mathbf{1}_{\calI_1}),\tilde{X}_2(\mathbf{1}_{\calI_2}),\cdots,\tilde{X}_T(\mathbf{1}_{\calI_T}),B}\right]\nonumber\\
    &=2\tra\left[\left(\bbI-\Pi_{U_{[T]}X_{[T]}B}\right)\rho_{U_{[T]}X_{[T]}B}\right]\nonumber\\
    &\mathop\le\limits^{(c)}2\tra\left[\left(\bbI-\Pi_{U_{[T]}X_{[T]}B}^{([S])}\right)\rho_{U_{[T]}X_{[T]}B}\right]\nonumber\\
    &\quad+\sum_{\calS\subset[S]}2\tra\left[\left(\bbI-\Pi_{U_{[T]}X_{[T]}B}^{(\calS)}\right)\rho_{U_{[T]}X_{[T]}B}\right]\nonumber\\
    &\mathop\le\limits^{(d)}2\times \mu_{B,[T]}^\alpha2^{\alpha\sum_{s=1}^SR_s}2^{-\alpha\ubar{\D}_{1-\alpha}\pr{\rho_{U_{[T]}X_{[T]}B}\lVert\rho_{U_{[T]}X_{[T]}}\otimes\rho_B}}\nonumber\\
    &\quad+2\sum_{\calS\subset[S]}\mu_{B,\calT_\calS^c}^\alpha2^{\alpha\sum_{s\in\calS}R_s}2^{-\alpha\ubar{\I}_{1-\alpha}\pr{U_{\calT_\calS},X_{\calT_\calS};B\lvert U_{\calT_\calS^c},X_{\calT_\calS^c}}},\label{eq:Error_Analysis_1}
\end{align}where
\begin{itemize}
    \item[$(a)$] follows from the definition of $\Upsilon$ in \eqref{eq:Gamma_Operator};
    \item[$(b)$] follows from the linearity of the trace operation and expectation, as well as by taking the expectation \ac{wrt} the random codebook $C$;
    \item[$(c)$] follows from the definition of $\Pi_{U_{[T]}X_{[T]}B}$ and since all the projections $\Pi_{U_{[T]}X_{[T]}B}^{(\calS)}$, for $\calS\in\mathfrak{S}$, commute with one another, and for each non-empty strict subset $\calS\in\mathfrak{S}$, we have
    \begin{subequations}\label{eq:Projection_Inequalities}
    \begin{align}
    &0\preceq\pr{\bbI-\Pi_{U_{[T]}X_{[T]}B}^{([S])}}^2=\bbI-\Pi_{U_{[T]}X_{[T]}B}^{([S])},\label{eq:Projection_Inequality_1}\\
    &0\preceq\pr{\bbI-\Pi_{U_{[T]}X_{[T]}B}^{(\calS)}}^2=\bbI-\Pi_{U_{[T]}X_{[T]}B}^{(\calS)},\label{eq:Projection_Inequality_S}\\
    &0\!\preceq\!\pr{\!\bbI-\Pi_{U_{[T]}X_{[T]}B}^{([S])}-\!\!\prod\limits_{\calS\in\mathfrak{S}}\!\Pi_{U_{[T]}X_{[T]}B}^{(\calS)}+\Pi_{U_{[T]}X_{[T]}B}\!}^2\nonumber\\
    &=\bbI-\Pi_{U_{[T]}X_{[T]}B}^{([S])}-\prod\limits_{\calS\in\mathfrak{S}}\Pi_{U_{[T]}X_{[T]}B}^{(\calS)}+\Pi_{U_{[T]}X_{[T]}B}\nonumber\\
    &\Rightarrow\bbI-\Pi_{U_{[T]}X_{[T]}B}\preceq\bbI-\Pi_{U_{[T]}X_{[T]}B}^{([S])}\nonumber\\
    &\qquad\qquad\qquad\qquad\quad+\bbI-\prod\limits_{\calS\in\mathfrak{S}}\Pi_{U_{[T]}X_{[T]}B}^{(\calS)},\label{eq:Projection_Inequality_[S]}\\
    &0\preceq\left(\bbI-\Pi_{U_{[T]}X_{[T]}B}^{(\calS_1)}-\prod\limits_{\substack{\calS\in\mathfrak{S}\\\calS\ne\calS_1}}\Pi_{U_{[T]}X_{[T]}B}^{(\calS)}\right.\nonumber\\
    &\qquad\left.+\prod\limits_{\calS\in\mathfrak{S}}\Pi_{U_{[T]}X_{[T]}B}^{(\calS)}\right)^2\nonumber\\
    &=\bbI-\Pi_{U_{[T]}X_{[T]}B}^{(\calS_1)}-\!\!\prod\limits_{\substack{\calS\in\mathfrak{S}\\\calS\ne\calS_1}}\Pi_{U_{[T]}X_{[T]}B}^{(\calS)}+\!\!\prod\limits_{\calS\in\mathfrak{S}}\Pi_{U_{[T]}X_{[T]}B}^{(\calS)}\\
    &\Rightarrow\bbI-\prod\limits_{\calS\in\mathfrak{S}}\Pi_{U_{[T]}X_{[T]}B}^{(\calS)}\preceq\bbI-\Pi_{U_{[T]}X_{[T]}B}^{(\calS_1)}\nonumber\\
    &\qquad\qquad\qquad\qquad\qquad+\bbI-\prod\limits_{\substack{\calS\in\mathfrak{S}\\\calS\ne\calS_1}}\Pi_{U_{[T]}X_{[T]}B}^{(\calS)},\label{eq:Projection_Inequality_S1}\\
    &\Rightarrow\bbI-\Pi_{U_{[T]}X_{[T]}B}\preceq\bbI-\Pi_{U_{[T]}X_{[T]}B}^{([S])}+\bbI-\Pi_{U_{[T]}X_{[T]}B}^{(\calS_1)}\nonumber\\
    &\qquad\qquad\qquad\qquad\quad+\bbI-\prod\limits_{\substack{\calS\in\mathfrak{S}\\\calS\ne\calS_1}}\Pi_{U_{[T]}X_{[T]}B}^{(\calS)},\label{eq:Projection_Inequality}
    \end{align}where the last inequality follows by substituting \eqref{eq:Projection_Inequality_S1} to the \ac{RHS} of \eqref{eq:Projection_Inequality_[S]}; continuing this process to bound $\bbI-\prod\limits_{\substack{\calS\in\mathfrak{S}\\\calS\ne\calS_1}}\Pi_{U_{[T]}X_{[T]}B}^{(\calS)}$, we obtain 
    \begin{align}
        \bbI-\Pi_{U_{[T]}X_{[T]}B}\preceq&\,\,\bbI-\Pi_{U_{[T]}X_{[T]}B}^{([S])}\nonumber\\
        &+\sum_{\calS\in\mathfrak{S}}\pr{\bbI-\Pi_{U_{[T]}X_{[T]}B}^{(\calS)}};\nonumber
    \end{align}
    \end{subequations}
    \item[$(d)$] follows from Lemma~\ref{lemma:error_analysis}.
\end{itemize}
Now, we bound the second term on the \ac{RHS} of \eqref{eq:Prob_Err_Analysis}. Specifically, we have
\begin{align}
    &4\!\sum_{m'_{[S]}\ne\mathbf{1}_{[S]}}\!\!\bbE_C\sbr{\tra\sbr{\Upsilon\pr{m'_{[S]}}\rho_{B\lvert \tilde{X}_1(\mathbf{1}_{\calI_1}),\tilde{X}_2(\mathbf{1}_{\calI_2}),\cdots,\tilde{X}_T(\mathbf{1}_{\calI_T})}}}\nonumber\\
    &\mathop=\limits^{(a)}4\sum_{\calS\subset[S]}\sum_{m'_\calS\ne\mathbf{1}_\calS}\bbE_C\left[\tra\left[\tra_{U_{[T]}X_{[T]}}\Big[\Pi_{U_{[T]}X_{[T]}B}\right.\right.\nonumber\\
    &\quad\times\big(\den{U_1(m''_{\calI_1})}{U_1(m''_{\calI_1})}\otimes\den{X_1(m''_{\calI_1})}{X_1(m''_{\calI_1})}
    \nonumber\\
    &\quad\otimes\cdots\otimes\den{U_T(m''_{\calI_T})}{U_T(m''_{\calI_T})}\otimes\den{X_T(m''_{\calI_T})}{X_T(m''_{\calI_T})}\nonumber\\
    &\quad\left.\left.\otimes\,\,\bbI_B\big)\Big]\rho_{B\lvert \tilde{X}_1(\mathbf{1}_{\calI_1}),\tilde{X}_2(\mathbf{1}_{\calI_2}),\cdots,\tilde{X}_T(\mathbf{1}_{\calI_T})}\right]\right]\nonumber\\
    &\mathop=\limits^{(b)}4\sum_{\calS\subset[S]}\sum_{m'_\calS\ne\mathbf{1}_\calS}\tra\left[\bbE_C\left[\Pi_{U_{[T]}X_{[T]}B}\big(\den{U_1(m''_{\calI_1})}{U_1(m''_{\calI_1})}\right.\right.\nonumber\\
    &\quad\otimes\den{X_1(m''_{\calI_1})}{X_1(m''_{\calI_1})}\otimes\cdots\otimes\den{U_T(m''_{\calI_T})}{U_T(m''_{\calI_T})}\nonumber\\
    &\quad\otimes\den{X_T(m''_{\calI_T})}{X_T(m''_{\calI_T})}\otimes\bbI_B\big)\nonumber\\
    &\left.\left.\hspace{2.8mm}\times\rho_{B\lvert \tilde{X}_1(\mathbf{1}_{\calI_1}),\tilde{X}_2(\mathbf{1}_{\calI_2}),\cdots,\tilde{X}_T(\mathbf{1}_{\calI_T})}\right]\right]\nonumber\\
    &\mathop\le\limits^{(c)}4\sum_{\calS\subset[S]}2^{\sum_{s\in\calS}R_s}\tra\left[\Pi_{\tilde{X}_{[T]}B}\pr{\rho_{\tilde{X}_{\calT_\calS}-\tilde{X}_{\calT_\calS^c}-B}}\right],\label{eq:bounding_second_1}
\end{align}where
\begin{itemize}
    \item[$(a)$] follows from the definition of $\Upsilon$ in \eqref{eq:Gamma_Operator}, and by defining  $m''_{\calI_t} = \mathbf{1}_{\calI_t}$ if, for each $s \in \calS$ and $t \in [T]$, we have $s \notin \calI_t$; otherwise, we define $m''_{\calI_t} \triangleq (m'_{\calL_t}, \mathbf{1}_{\calL^c_t})$, where $\calL_t$ is a subset of $\calI_t$ consisting of elements that all belong to $\calS$;
    \item[$(b)$] follows from the linearity of the trace operation and expectation; 
    \item[$(c)$] follows by taking the expectation \ac{wrt} the random codebook $C$, symmetry of codebook construction \ac{wrt} $m'_\calS$, and defining $\calT_\calS$ as the set of all indices $t$ for which $m''_{\calI_t}\ne\mathbf{1}_{\calI_t}$, that is $\calT_\calS$, for $\calS\subset[S]$, denotes the set of channel inputs for which each channel input has access to at least one message $M_s$, for some $s \in \calS$.
\end{itemize}
Now we bound the \ac{RHS} of \eqref{eq:bounding_second_1} for two different cases: (i) when $\calT_\calS^c=\emptyset$, and $(ii)$ when $\calT_\calS^c\ne\emptyset$. First, when $\calT_\calS^c=\emptyset$, we have
\begin{align}
    &4\sum_{\calS\subset[S]}2^{\sum_{s\in\calS}R_s}\tra\left[\Pi_{\tilde{X}_{[T]}B}\pr{\rho_{\tilde{X}_{\calT_\calS}-\tilde{X}_{\calT_\calS^c}-B}}\right]\nonumber\\
    &\mathop=\limits^{(a)}4\times2^{\sum_{s\in\calS}R_s}\tra\left[\Pi_{\tilde{X}_{[T]}B}\pr{\rho_{\tilde{X}_{[T]}}\otimes\rho_B}\right]\nonumber\\
    &\mathop\le\limits^{(b)}4\times2^{\sum_{s\in\calS}R_s}\tra\left[\Pi_{\tilde{X}_{[T]}B}^{([S])}\pr{\rho_{\tilde{X}_{[T]}}\otimes\rho_B}\right]\nonumber\\
    &\mathop\le\limits^{(c)}4\times2^{\sum_{s\in\calS}R_s}\mu_{B,[T]}^\alpha2^{-(1-\alpha)\sum_{s=1}^SR_s}\nonumber\\
    &\quad\times2^{-\alpha\ubar{\D}_{1-\alpha}\pr{\rho_{U_{[T]}X_{[T]}B}\lVert\rho_{U_{[T]}X_{[T]}}\otimes\rho_B}}\nonumber\\
    &=4\mu_{B,[T]}^\alpha2^{\alpha\pr{\sum_{s=1}^SR_s-\ubar{\D}_{1-\alpha}\pr{\rho_{U_{[T]}X_{[T]}B}\lVert\rho_{U_{[T]}X_{[T]}}\otimes\rho_B}}},\label{eq:bounding_second_1_Case1}
\end{align}where
\begin{itemize}
    \item[$(a)$] follows since $\calT_\calS^c=\emptyset$;
    \item[$(b)$] follows from \eqref{eq:Projection_Inequality_1};
    \item[$(c)$] follows from Lemma~\ref{lemma:error_analysis}. 
\end{itemize}
\begin{figure*}[b!]
\hrulefill
\setcounter{equation}{40}
\begin{align}
    &\tra\left[\left(\bbI-\Pi_{\tilde{X}_{[T]}B}^{([S])}\right)\rho_{\tilde{X}_{[T]}B}\right]\nonumber\\
    &\quad\mathop=\limits^{(a)}\tra\left[\pr{\Delta_B\left(\bbI-\Pi_{\tilde{X}_{[T]}B}^{([S])}\right)}\rho_{\tilde{X}_{[T]}B}\right]\nonumber\\
    &\quad\mathop=\limits^{(b)}\tra\left[\left(\bbI-\Pi_{\tilde{X}_{[T]}B}^{([S])}\right)\Delta_B\pr{\rho_{\tilde{X}_{[T]}B}}\right]\nonumber\\
    &\quad=\tra\left[\left(\bbI-\Pi_{\tilde{X}_{[T]}B}^{([S])}\right)\pr{\Delta_B\pr{\rho_{\tilde{X}_{[T]}B}}}^{1-\alpha}\pr{\Delta_B\pr{\rho_{\tilde{X}_{[T]}B}}}^\alpha\right]\nonumber\\
    &\quad\mathop\le\limits^{(c)}2^{\alpha\sum_{s=1}^SR_s}\tra\left[\left(\bbI-\Pi_{\tilde{X}_{[T]}B}^{([S])}\right)\pr{\Delta_B\pr{\rho_{\tilde{X}_{[T]}B}}}^{1-\alpha}\pr{\rho_{\tilde{X}_{[T]}}\otimes\rho_B}^\alpha\right]\nonumber\\
    &\quad\mathop\le\limits^{(d)}2^{\alpha\sum_{s=1}^SR_s}\tra\left[\pr{\Delta_B\pr{\rho_{\tilde{X}_{[T]}B}}}^{1-\alpha}\pr{\rho_{\tilde{X}_{[T]}}\otimes\rho_B}^\alpha\right]\nonumber\\
    &\quad\mathop=\limits^{(e)}2^{\alpha\sum_{s=1}^SR_s}\tra\left[\pr{\pr{\rho_{\tilde{X}_{[T]}}\otimes\rho_B}^{\frac{\alpha}{2(1-\alpha)}}\Delta_B\pr{\rho_{\tilde{X}_{[T]}B}}\pr{\rho_{\tilde{X}_{[T]}}\otimes\rho_B}^{\frac{\alpha}{2(1-\alpha)}}}^{1-\alpha}\right]\nonumber\\
    &\quad=2^{\alpha\sum_{s=1}^SR_s}\tra\left[\pr{\pr{\rho_{\tilde{X}_{[T]}}\otimes\rho_B}^{\frac{\alpha}{2(1-\alpha)}}\Delta_B\pr{\rho_{\tilde{X}_{[T]}B}}\pr{\rho_{\tilde{X}_{[T]}}\otimes\rho_B}^{\frac{\alpha}{2(1-\alpha)}}}\right.\nonumber\\
    &\qquad\times\left.\pr{\pr{\rho_{\tilde{X}_{[T]}}\otimes\rho_B}^{\frac{\alpha}{2(1-\alpha)}}\Delta_B\pr{\rho_{\tilde{X}_{[T]}B}}\pr{\rho_{\tilde{X}_{[T]}}\otimes\rho_B}^{\frac{\alpha}{2(1-\alpha)}}}^{-\alpha}\right]\nonumber\\
    &\quad\mathop=\limits^{(f)}2^{\alpha\sum_{s=1}^SR_s}\tra\left[\pr{\pr{\rho_{\tilde{X}_{[T]}}\otimes\rho_B}^{\frac{\alpha}{2(1-\alpha)}}\rho_{\tilde{X}_{[T]}B}\pr{\rho_{\tilde{X}_{[T]}}\otimes\rho_B}^{\frac{\alpha}{2(1-\alpha)}}}\right.\nonumber\\
    &\qquad\times\left.\pr{\pr{\rho_{\tilde{X}_{[T]}}\otimes\rho_B}^{\frac{\alpha}{2(1-\alpha)}}\Delta_B\pr{\rho_{\tilde{X}_{[T]}B}}\pr{\rho_{\tilde{X}_{[T]}}\otimes\rho_B}^{\frac{\alpha}{2(1-\alpha)}}}^{-\alpha}\right]\nonumber\\
    &\quad\mathop\le\limits^{(g)}\mu_{B,[T]}^\alpha2^{\alpha\sum_{s=1}^SR_s}\tra\left[\pr{\pr{\rho_{\tilde{X}_{[T]}}\otimes\rho_B}^{\frac{\alpha}{2(1-\alpha)}}\rho_{\tilde{X}_{[T]}B}\pr{\rho_{\tilde{X}_{[T]}}\otimes\rho_B}^{\frac{\alpha}{2(1-\alpha)}}}\right.\nonumber\\
    &\qquad\times\left.\pr{\pr{\rho_{\tilde{X}_{[T]}}\otimes\rho_B}^{\frac{\alpha}{2(1-\alpha)}}\rho_{\tilde{X}_{[T]}B}\pr{\rho_{\tilde{X}_{[T]}}\otimes\rho_B}^{\frac{\alpha}{2(1-\alpha)}}}^{-\alpha}\right]\nonumber\\
    &\quad=\mu_{B,[T]}^\alpha2^{\alpha\sum_{s=1}^SR_s}\tra\left[\pr{\pr{\rho_{\tilde{X}_{[T]}}\otimes\rho_B}^{\frac{\alpha}{2(1-\alpha)}}\rho_{\tilde{X}_{[T]}B}\pr{\rho_{\tilde{X}_{[T]}}\otimes\rho_B}^{\frac{\alpha}{2(1-\alpha)}}}^{1-\alpha}\right]\nonumber\\
    &\quad=\mu_{B,[T]}^\alpha2^{\alpha\sum_{s=1}^SR_s}2^{-\alpha\ubar{\D}_{1-\alpha}\pr{\rho_{U_{[T]}X_{[T]}B}\lVert\rho_{U_{[T]}X_{[T]}}\otimes\rho_B}},\label{eq:First_Term_Lemma}
\end{align}
\setcounter{equation}{36}
\end{figure*}
When $\calT_\calS^c\ne\emptyset$, we have
\begin{align}
    &4\sum_{\calS\subset[S]}2^{\sum_{s\in\calS}R_s}\tra\left[\Pi_{\tilde{X}_{[T]}B}\pr{\rho_{\tilde{X}_{\calT_\calS}-\tilde{X}_{\calT_\calS^c}-B}}\right]\nonumber\\
    &\mathop\le\limits^{(a)}4\sum_{\calS\subset[S]}2^{\sum_{s\in\calS}R_s}\tra\left[\Pi_{\tilde{X}_{[T]}B}^{(\calS)}\pr{\rho_{\tilde{X}_{\calT_\calS}-\tilde{X}_{\calT_\calS^c}-B}}\right]\nonumber\\
    &\mathop\le\limits^{(b)}4\sum_{\calS\subset[S]}2^{\sum_{s\in\calS}R_s} \mu_{B,\calT_\calS^c}^\alpha2^{-(1-\alpha)\sum_{s\in\calS}R_s}\nonumber\\
    &\quad\times2^{-\alpha\ubar{\I}_{1-\alpha}\pr{U_{\calT_\calS},X_{\calT_\calS};B\lvert U_{\calT_\calS^c},X_{\calT_\calS^c}}}\nonumber\\
    &=4\sum_{\calS\subset[S]} \mu_{B,\calT_\calS^c}^\alpha2^{\alpha\pr{\sum_{s\in\calS}R_s-\ubar{\I}_{1-\alpha}\pr{U_{\calT_\calS},X_{\calT_\calS};B\lvert U_{\calT_\calS^c},X_{\calT_\calS^c}}}},\label{eq:bounding_second_1_Case2}
\end{align}where
\begin{itemize}
    \item[$(a)$] follows from \eqref{eq:Projection_Inequality_S};
    \item[$(b)$] follows from Lemma~\ref{lemma:error_analysis}.
\end{itemize}
Therefore, the \ac{RHS} of \eqref{eq:bounding_second_1} is less than the sum of the \ac{RHS} of \eqref{eq:bounding_second_1_Case1} and \eqref{eq:bounding_second_1_Case2}. Now, substituting \eqref{eq:Error_Analysis_1} and the sum of the \ac{RHS} of \eqref{eq:bounding_second_1_Case1} and \eqref{eq:bounding_second_1_Case2} in \eqref{eq:Prob_Err_Analysis} completes the proof of \eqref{eq:Pe_Thm}.

\subsection{Covertness Analysis}
\label{sec:Covertness}
 Let,
    \begin{align}
        \rho_E&\triangleq\sum_{\tilde{x}_{[T]}\in\tilde{\calX}_{[T]}}\prod_{t=1}^Tp_{\tilde{X}_t\lvert\tilde{X}_{\bbD_t}}(\tilde{x}_t\lvert \tilde{x}_{\bbD_t})\rho_{E\lvert\tilde{x}_{[T]}},\label{eq:Lemma_rho}
    \end{align}\sloppy where $\tilde{x}_t=(u_t,x_t)$ and $\tilde{\calX}_t=(\calU_t,\calX_t)$, if $t\in\bbM$, and $\tilde{x}_t=x_t$ and $\tilde{\calX}_t=\calX_t$ otherwise, and $\left\{\prod_{t=1}^Tp_{\tilde{X}_t\lvert \tilde{X}_{\bbD_t}}(\tilde{x}_t\lvert \tilde{x}_{\bbD_t}),\rho_{\tilde{x}_{[T]}}\right\}_{\tilde{x}_{[T]}\in\tilde{\calX}_{[T]}}$ is a given ensemble. To prove that our code design is also covert, we first bound $\bbE_C\Pu\left(\tau_{E\lvert C},\rho_E\right)$, where $\tau_{E\lvert C}$ is the state induced at the output of the warden by our code design, which is defined in \eqref{eq:Joint_Dist}. Then, we choose the distributions $\pr{p_{\tilde{X}_t\lvert \tilde{X}_{\bbD_t}}(\tilde{x}_t\lvert \tilde{x}_{\bbD_t})}_{t\in[T]}$, and the mappings $\theta_{A_1}^{\tilde{x}_1},\theta_{A_2}^{\tilde{x}_2},\cdots,\theta_{A_T}^{\tilde{x}_T}$ such that $\rho_E=\rho_0$. 

\begin{figure*}[b!]
\hrulefill
\setcounter{equation}{43}
\begin{align}
    &\tra\left[\left(\bbI-\Pi_{\tilde{X}_{[T]}B}^{(\calS)}\right)\rho_{\tilde{X}_{[T]}B}\right]\nonumber\\
    &\quad\mathop=\limits^{(a)}\tra\left[\pr{\Delta_{\calT_\calS^c}\left(\bbI-\Pi_{\tilde{X}_{[T]}B}^{(\calS)}\right)}\rho_{\tilde{X}_{[T]}B}\right]\nonumber\\
    &\quad\mathop=\limits^{(b)}\tra\left[\left(\bbI-\Pi_{\tilde{X}_{[T]}B}^{(\calS)}\right)\Delta_{\calT_\calS^c}\pr{\rho_{\tilde{X}_{[T]}B}}\right]\nonumber\\
    &\quad=\tra\left[\left(\bbI-\Pi_{\tilde{X}_{[T]}B}^{(\calS)}\right)\pr{\Delta_{\calT_\calS^c}\pr{\rho_{\tilde{X}_{[T]}B}}}^{1-\alpha}\pr{\Delta_{\calT_\calS^c}\pr{\rho_{\tilde{X}_{[T]}B}}}^\alpha\right]\nonumber\\
    &\quad\mathop\le\limits^{(c)}2^{\alpha\sum_{s\in\calS}R_s}\tra\left[\left(\bbI-\Pi_{\tilde{X}_{[T]}B}^{(\calS)}\right)\pr{\Delta_{\calT_\calS^c}\pr{\rho_{\tilde{X}_{[T]}B}}}^{1-\alpha}\pr{\Delta_B\pr{\rho_{\pr{U_{\calT_\calS},X_{\calT_\calS}}-\pr{U_{\calT_\calS^c},X_{\calT_\calS^c}}-B}}}^\alpha\right]\nonumber\\
    &\quad\mathop\le\limits^{(d)}2^{\alpha\sum_{s\in\calS}R_s}\tra\left[\pr{\Delta_{\calT_\calS^c}\pr{\rho_{\tilde{X}_{[T]}B}}}^{1-\alpha}\pr{\Delta_B\pr{\rho_{\pr{U_{\calT_\calS},X_{\calT_\calS}}-\pr{U_{\calT_\calS^c},X_{\calT_\calS^c}}-B}}}^\alpha\right]\nonumber\\
    &\quad\mathop=\limits^{(e)}2^{\alpha\sum_{s\in\calS}R_s}\tra\left[\pr{\Delta_{\calT_\calS^c}\pr{\rho_{\tilde{X}_{[T]}B}}}^{1-\alpha}\pr{\sigma_{\tilde{X}_{[T]}B}}^\alpha\right]\nonumber\\
    &\quad\mathop\le\limits^{(f)}\mu_{B,\calT_\calS^c}^\alpha2^{\alpha\sum_{s\in\calS}R_s}2^{-\alpha\ubar{\D}_{1-\alpha}\pr{\rho_{U_{[T]}X_{[T]}B}\lVert\sigma_{U_{[T]}X_{[T]}B}}}\nonumber\\
    &\quad\le \mu_{B,\calT_\calS^c}^\alpha2^{\alpha\sum_{s\in\calS}R_s}2^{-\alpha\min\limits_{\sigma_{U_{[T]}X_{[T]}B}}\ubar{\D}_{1-\alpha}\pr{\rho_{U_{[T]}X_{[T]}B}\lVert\sigma_{U_{[T]}X_{[T]}B}}}\nonumber\\
    &\quad=\mu_{B,\calT_\calS^c}^\alpha2^{\alpha\sum_{s\in\calS}R_s}2^{-\alpha \ubar{\I}_{1-\alpha}\pr{U_{\calT_\calS},X_{\calT_\calS};B\lvert U_{\calT_\calS^c},X_{\calT_\calS^c}}},\label{eq:Second_Term_Lemma}
\end{align}
\setcounter{equation}{38}
\end{figure*}
\begin{lemma}
    \label{lemma:Resolvability}
    Let $\rho_E$ be a classical-quantum state as defined in \eqref{eq:Lemma_rho}, and let $\mathfrak{B} \triangleq \pr{\calS_1, \calS_2, \dots, \calS_{\card{\mathfrak{B}}}}$ be an ordered collection of all non-empty strict subsets of $[S]$, arranged such that for any $1 \le i < j \le \card{\mathfrak{B}}$, we have $\card{\calS_i} \le \card{\calS_j}$. Define, for each subset $\calS \subset [S]$, the set $\tilde{X}(\calS)$ as the collection of channel inputs that depend only on the messages whose indices belong to $\calS$.
    Also, let $C$ be a random codebook as defined in Appendix~\ref{sec:Codebook_Cons}. Then, 
\begin{align}
    &\bbE_C\left[\ubar{\D}_{1+\alpha}\big(\tau_{E\lvert C}\lVert\rho_E\big)\right]\nonumber\\
    &\le\frac{1}{\alpha\ln2}\left(\mu_{E,\card{\mathfrak{B}}}^\alpha2^{\alpha\pr{-\sum\limits_{s=1}^SR_s+\ubar{\D}_{1+\alpha}\pr{\rho_{\tilde{X}_{[T]}E}\big\lVert\rho_{\tilde{X}_{[T]}}\otimes\rho_E}}}\right.\nonumber\\
    &\left.\quad+\sum_{i=1}^{\card{\mathfrak{B}}}\mu_{E,i-1}^\alpha2^{\alpha\pr{-\sum\limits_{s\in\calS_i}R_s+\ubar{\D}_{1+\alpha}\pr{\rho_{\tilde{X}\pr{\calS_i}E}\big\lVert\rho_{\tilde{X}\pr{\calS_i}}\otimes\rho_E}}}\right),\nonumber
\end{align}
where $\tau_{E\lvert C}$ is defined in~\eqref{eq:Joint_Dist}.
\end{lemma}
Lemma~\ref{lemma:Resolvability} is proved in Appendix~\ref{proof:lemma_Resolvability}.
\begin{figure*}[b!]
\hrulefill
\setcounter{equation}{44}
\begin{align}
    &\tra\left[\Pi_{U_{[T]}X_{[T]}B}^{(\calS)}\left(\rho_{\pr{U_{\calT_\calS},X_{\calT_\calS}}-\pr{U_{\calT_\calS^c},X_{\calT_\calS^c}}-B}\right)\right]\nonumber\\
    &\quad\mathop=\limits^{(a)}\tra\left[\pr{\Delta_B\pr{\Pi_{U_{[T]}X_{[T]}B}^{(\calS)}}}\left(\rho_{\pr{U_{\calT_\calS},X_{\calT_\calS}}-\pr{U_{\calT_\calS^c},X_{\calT_\calS^c}}-B}\right)\right]\nonumber\\
    &\quad\mathop=\limits^{(b)}\tra\left[\Pi_{U_{[T]}X_{[T]}B}^{(\calS)}\Delta_B\pr{\rho_{\pr{U_{\calT_\calS},X_{\calT_\calS}}-\pr{U_{\calT_\calS^c},X_{\calT_\calS^c}}-B}}\right]\nonumber\\
    &\quad=\tra\left[\Pi_{U_{[T]}X_{[T]}B}^{(\calS)}\pr{\Delta_B\pr{\rho_{\pr{U_{\calT_\calS},X_{\calT_\calS}}-\pr{U_{\calT_\calS^c},X_{\calT_\calS^c}}-B}}}^{1-\alpha}\pr{\Delta_B\pr{\rho_{\pr{U_{\calT_\calS},X_{\calT_\calS}}-\pr{U_{\calT_\calS^c},X_{\calT_\calS^c}}-B}}}^\alpha\right]\nonumber\\
    &\quad\mathop\le\limits^{(c)}2^{-(1-\alpha)\sum_{s\in\calS}R_s}\tra\left[\Pi_{U_{[T]}X_{[T]}B}^{(\calS)}\pr{\Delta_{\calT_\calS^c}\pr{\rho_{U_{[T]}X_{[T]}B}}}^{1-\alpha}\pr{\Delta_B\pr{\rho_{\pr{U_{\calT_\calS},X_{\calT_\calS}}-\pr{U_{\calT_\calS^c},X_{\calT_\calS^c}}-B}}}^\alpha\right]\nonumber\\
    &\quad\mathop\le\limits^{(d)}2^{-(1-\alpha)\sum_{s\in\calS}R_s}\tra\left[\pr{\Delta_{\calT_\calS^c}\pr{\rho_{U_{[T]}X_{[T]}B}}}^{1-\alpha}\pr{\Delta_B\pr{\rho_{\pr{U_{\calT_\calS},X_{\calT_\calS}}-\pr{U_{\calT_\calS^c},X_{\calT_\calS^c}}-B}}}^\alpha\right]\nonumber\\
    &\quad\mathop\le\limits^{(e)}\mu_{B,\calT_\calS^c}^\alpha2^{-(1-\alpha)\sum_{s\in\calS}R_s}2^{-\alpha \ubar{\I}_{1-\alpha}\pr{U_{\calT_\calS},X_{\calT_\calS};B\lvert U_{\calT_\calS^c},X_{\calT_\calS^c}}},\label{eq:Fourth_Term_Lemma}
\end{align}
\setcounter{equation}{38}
\end{figure*}

We have,
\begin{align}
    &\bbE_C\Pu\left(\tau_{E\lvert C},\rho_E\right)^2\nonumber\\
    &\mathop\le\limits^{(a)}\bbE_C\left[1-2^{-\ubar{\D}_{1+\alpha}\left(\tau_{E\lvert C}\big\lVert\rho_E\right)}\right]\nonumber\\
    &\mathop\le\limits^{(b)}\ln2\times\bbE_C\left[\ubar{\D}_{1+\alpha}\left(\tau_{E\lvert C}\Big\lVert\rho_E\right)\right]\nonumber\\
    &\mathop\le\limits^{(c)}\frac{1}{\alpha}\left(\mu_{E,\card{\mathfrak{B}}}^\alpha2^{\alpha\pr{-\sum\limits_{s=1}^SR_s+\ubar{\D}_{1+\alpha}\pr{\rho_{\tilde{X}_{[T]}E}\big\lVert\rho_{\tilde{X}_{[T]}}\otimes\rho_E}}}\right.\nonumber\\
    &\left.\quad+\sum_{i=1}^{\card{\mathfrak{B}}}\mu_{E,i-1}^\alpha2^{\alpha\pr{-\sum\limits_{s\in\calS_i}R_s+\ubar{\D}_{1+\alpha}\pr{\rho_{\tilde{X}\pr{\calS_i}E}\big\lVert\rho_{\tilde{X}\pr{\calS_i}}\otimes\rho_E}}}\right),\label{eq:Final_Covertness}
\end{align}where
\begin{itemize}
    \item[$(a)$] follows since for two quantum states $\rho\in\calD(\calH)$ and $\sigma\in\calD(\calH)$ we have \cite[Corollary~4.3]{Tomamichel16},
    \begin{align}
        F^2(\rho,\sigma)=2^{-\ubar{\D}_{\frac{1}{2}}\left(\rho\lVert\sigma\right)}\ge2^{-\ubar{\D}_{1+\alpha}\left(\rho\lVert\sigma\right)};\nonumber
    \end{align}
    \item[$(b)$] follows since 
    \begin{align}
        1-2^{-\frac{x}{\ln2}}=1-e^{-x}\le x;\nonumber
    \end{align}
    \item[$(c)$] follows from Lemma~\ref{lemma:Resolvability}.
\end{itemize}Now for two arbitrary states $\rho\in\calD(\calH)$ and $\sigma\in\calD(\calH)$ we have \cite[Theorem~9.3.1]{Wilde_Book}
\begin{align}
    \lVert\rho-\sigma\rVert_1\le2\Pu(\rho,\sigma).\label{eq:Tr_Distance_Purified_Dist}
\end{align}Therefore, considering \eqref{eq:Tr_Distance_Purified_Dist} and the fact that for $a,b\in\bbR^+$ we have $\sqrt{a+b}\le\sqrt{a}+\sqrt{b}$, \eqref{eq:Final_Covertness} completes the proof of \eqref{eq:Covertnes_Thm}.

To make the proof steps more transparent, the special cases with four messages and six transmitters, as well as with two messages and three transmitters, are worked out in detail in \cite{4M6Tx} and \cite{2M3Tx}, respectively.
\section{Proof of Lemma~\ref{lemma:error_analysis}}
\label{proof:lemma_error_analysis}
To prove \eqref{eq:Lemma_Error_1}, assuming $\tilde{X}_t = (X_t, U_t)$ if $v_t \in \bbM$, and $\tilde{X}_t = X_t$ otherwise, we have \eqref{eq:First_Term_Lemma} provided at the bottom of the previous page, where
\begin{itemize}
    \item[$(a)$] follows since $\Pi^{([S])}_{\tilde{X}_{[T]}B}$ commutes with $\rho_B$ and therefore $\Delta_B\left(\bbI-\Pi_{\tilde{X}_{[T]}B}^{([S])}\right)=\bbI-\Pi_{\tilde{X}_{[T]}B}^{([S])}$;
    \item[$(b)$] follows since for three arbitrary states $\sigma\in\calD(\calH)$, with spectral decomposition $\sigma=\sum_i\lambda_i\den{x_i}{x_i}$, $\rho_1\in\calD(\calH)$, and $\rho_2\in\calD(\calH)$ we have
    \setcounter{equation}{41}
    \begin{align}
    \tra\left[\Delta_\sigma(\rho_1)\rho_2\right]&=\sum_i\bra{x_i}\rho_1\ket{x_i}\tra\left[\den{x_i}{x_i}\rho_2\right]\nonumber\\
    &=\sum_i\bra{x_i}\rho_1\ket{x_i}\tra\left[\bra{x_i}\rho_2\ket{x_i}\right]\nonumber\\
    &=\sum_i\tra\left[\bra{x_i}\rho_1\ket{x_i}\right]\bra{x_i}\rho_2\ket{x_i}\nonumber\\
    &=\sum_i\tra\left[\rho_1\den{x_i}{x_i}\right]\bra{x_i}\rho_2\ket{x_i}\nonumber\\
    &=\sum_i\tra\left[\rho_1\bra{x_i}\rho_2\ket{x_i}\den{x_i}{x_i}\right]\nonumber\\
    &=\tra[\rho_1\Delta_\sigma(\rho_2)];\label{eq:Pinching_Switching}
    \end{align}
    \item[$(c)$] follows from the definition of $\Pi_{\tilde{X}_{[T]}B}^{([S])}$ in \eqref{eq:Decoding_Projection_1} and since $f(x)=x^\beta$, for $\beta\in(0\,,1]$, is a matrix monotone function \cite[Section~1.5]{QIT_Hayashi};
    \item[$(d)$] follows from \eqref{eq:Projection_Inequality_1} and since for $A,B,C\in\calP(\calH)$ and $C\preceq B$ we have $\tra(CA)\le\tra(BA)$, since $\tra((B-C)A)\ge0$ \cite[Lemma~B.5.2]{RennerDissertation};
    \item[$(e)$] follows since $\Delta_B\pr{\rho_{\tilde{X}_{[T]}B}}$ and $\rho_{\tilde{X}_{[T]}}\otimes\rho_B$ commute;
    \item[$(f)$] follows since for two arbitrary states $\sigma\in\calD(\calH)$, with spectral decomposition $\sigma=\sum_i\lambda_i\den{x_i}{x_i}$, and $\rho\in\calD(\calH)$ we have $\tra\left[\Delta_\sigma(\rho)\sigma\right]=\tra\left[\rho\sigma\right]$ since
        \begin{align}
        \tra[\rho\sigma]&=\sum_i\lambda_i\tra\left[\rho\den{x_i}{x_i}\right]\nonumber\\
        &=\sum_i\lambda_i\tra\left[\bra{x_i}\rho\ket{x_i}\right]\nonumber\\
        &=\tra\left[\Delta_\sigma(\rho)\sigma\right];\nonumber
    \end{align}
    \item[$(g)$] follows from the following properties:  
    \begin{enumerate}[i.]
        \item For two arbitrary states $\sigma,\rho\in\calD(\calH)$, we have
        \begin{align}
            \rho\preceq \mu\Delta_\sigma(\rho),\label{eq:Pinching_Prop0}
        \end{align}where $\mu$ denotes the number of distinct eigenvalues of $\sigma$.
        \item For $A, B, C \in \calP(\calH)$ with $C \preceq B$, it holds that $\operatorname{Tr}(CA) \leq \operatorname{Tr}(BA)$ \cite[Lemma~B.5.2]{RennerDissertation}.  
        \item For any Hermitian matrix $A$ and $\rho \in \calP(\calH)$, the matrix $A\rho A$ remains non-negative \cite[Lemma~B.5.1]{RennerDissertation}. 
        \item The function $f(x) = x^{-\beta}$, for $\beta \in (0,1)$, is a matrix anti-monotone function \cite[Chapter~5]{Bhatia_Book}.
        \end{enumerate}
\end{itemize}
\begin{figure*}[b!]
    \hrulefill
    \setcounter{equation}{46}
\begin{align}
    &\bbE_C\sbr{2^{\alpha\ubar{\D}_{1+\alpha}\big(\tau_{E\lvert C}\lVert\rho_E\big)}}\nonumber\\
    &=\bbE_C\tra\left[\left(\rho_E^{-\frac{\alpha}{2(1+\alpha)}}\tau_{E\lvert C}\rho_E^{-\frac{\alpha}{2(1+\alpha)}}\right)^{1+\alpha}\right]\nonumber\\
    &=\bbE_C\tra\left[\left(\rho_E^{-\frac{\alpha}{2(1+\alpha)}}\frac{1}{2^{\sum_{s=1}^SR_s}}\sum_{m_{[S]}}\rho_{E\lvert\tilde{X}_1\pr{m_{\calI_1}},\tilde{X}_2\pr{m_{\calI_2}},\cdots,\tilde{X}_T\pr{m_{\calI_T}}}\rho_E^{-\frac{\alpha}{2(1+\alpha)}}\right)^{1+\alpha}\right]\nonumber\\
    &=\bbE_C\tra\left[\left(\rho_E^{-\frac{\alpha}{2(1+\alpha)}}\frac{1}{2^{\sum_{s=1}^SR_s}}\sum_{m_{[S]}}\rho_{E\lvert\tilde{X}_1\pr{m_{\calI_1}},\tilde{X}_2\pr{m_{\calI_2}},\cdots,\tilde{X}_T\pr{m_{\calI_T}}}\rho_E^{-\frac{\alpha}{2(1+\alpha)}}\right)\right.\nonumber\\
    &\left.\times\left(\rho_E^{-\frac{\alpha}{2(1+\alpha)}}\frac{1}{2^{\sum_{s=1}^SR_s}}\sum_{m'_{[S]}}\rho_{E\lvert\tilde{X}_1\pr{m'_{\calI_1}},\tilde{X}_2\pr{m'_{\calI_2}},\cdots,\tilde{X}_T\pr{m'_{\calI_T}}}\rho_E^{-\frac{\alpha}{2(1+\alpha)}}\right)^\alpha\right]\nonumber\\
    &\mathop=\limits^{(a)}\frac{1}{2^{\sum_{s=1}^SR_s}}\sum_{m_{[S]}}\bbE_C\tra\left[\left(\rho_E^{-\frac{\alpha}{2(1+\alpha)}}\rho_{E\lvert\tilde{X}_1\pr{m_{\calI_1}},\tilde{X}_2\pr{m_{\calI_2}},\cdots,\tilde{X}_T\pr{m_{\calI_T}}}\rho_E^{-\frac{\alpha}{2(1+\alpha)}}\right)\right.\nonumber\\
    &\left.\times\left(\rho_E^{-\frac{\alpha}{2(1+\alpha)}}\frac{1}{2^{\sum_{s=1}^SR_s}}\sum_{m'_{[S]}}\rho_{E\lvert\tilde{X}_1\pr{m'_{\calI_1}},\tilde{X}_2\pr{m'_{\calI_2}},\cdots,\tilde{X}_T\pr{m'_{\calI_T}}}\rho_E^{-\frac{\alpha}{2(1+\alpha)}}\right)^\alpha\right]\nonumber\\
    &\mathop=\limits^{(b)}\frac{1}{2^{\sum_{s=1}^SR_s}}\sum_{m_{[S]}}\bbE_C\tra\left[\left(\rho_E^{-\frac{\alpha}{2(1+\alpha)}}\rho_{E\lvert\tilde{X}_1\pr{m_{\calI_1}},\tilde{X}_2\pr{m_{\calI_2}},\cdots,\tilde{X}_T\pr{m_{\calI_T}}}\rho_E^{-\frac{\alpha}{2(1+\alpha)}}\right)\right.\nonumber\\
    &\times\left(\rho_E^{-\frac{\alpha}{2(1+\alpha)}}\frac{1}{2^{\sum_{s=1}^SR_s}}\left(\rho_{E\lvert\tilde{X}_1\pr{m_{\calI_1}},\tilde{X}_2\pr{m_{\calI_2}},\cdots,\tilde{X}_T\pr{m_{\calI_T}}}+\sum_{\calS\subset[S]}\sum_{m'_\calS\ne m_\calS}\rho_{E\big\lvert\tilde{X}(\calS^c),\tilde{X}^c(\calS^c)}\right.\right.\nonumber\\
    &\left.\left.\left.+\sum_{\pr{m'_1,m'_2,\cdots,m'_S}\ne\pr{ m_1,m_2,\cdots,m_S}}\rho_{E\lvert\tilde{X}_1\pr{m'_{\calI_1}},\tilde{X}_2\pr{m'_{\calI_2}},\cdots,\tilde{X}_T\pr{m'_{\calI_T}}}\right)\rho_E^{-\frac{\alpha}{2(1+\alpha)}}\right)^\alpha\right]\nonumber\\
    &\mathop\le\limits^{(c)}\frac{1}{2^{\sum_{s=1}^SR_s}}\sum_{m_{[S]}}\bbE_{\pr{m_1,\cdots,m_S}}\tra\left[\left(\rho_E^{-\frac{\alpha}{2(1+\alpha)}}\rho_{E\lvert\tilde{X}_1\pr{m_{\calI_1}},\tilde{X}_2\pr{m_{\calI_2}},\cdots,\tilde{X}_T\pr{m_{\calI_T}}}\rho_E^{-\frac{\alpha}{2(1+\alpha)}}\right)\right.\nonumber\\
    &\times\left(\rho_E^{-\frac{\alpha}{2(1+\alpha)}}\frac{1}{2^{\sum_{s=1}^SR_s}}\left(\rho_{E\lvert\tilde{X}_1\pr{m_{\calI_1}},\tilde{X}_2\pr{m_{\calI_2}},\cdots,\tilde{X}_T\pr{m_{\calI_T}}}+\sum_{\calS\subset[S]}\sum_{m'_\calS\ne m_\calS}\bbE_{\backslash m_\calS}\rho_{E\big\lvert\tilde{X}(\calS^c),\tilde{X}^c(\calS^c)}\right.\right.\nonumber\\
    &\left.\left.\left.+\sum_{\pr{m'_1,m'_2,\cdots,m'_S}\ne\pr{ m_1,m_2,\cdots,m_S}}\bbE_{\backslash\pr{m_1,\cdots,m_{[S]}}}\rho_{E\lvert\tilde{X}_1\pr{m'_{\calI_1}},\tilde{X}_2\pr{m'_{\calI_2}},\cdots,\tilde{X}_T\pr{m'_{\calI_T}}}\right)\rho_E^{-\frac{\alpha}{2(1+\alpha)}}\right)^\alpha\right],\label{eq:Resolv_Analysis_1}
    \end{align}
\end{figure*}
Now we prove \eqref{eq:Lemma_Error_2}, given in \eqref{eq:Second_Term_Lemma} at the bottom of the previous page, where
\begin{itemize}
    \item[$(a)$] follows since $\Pi^{(\calS)}_{\tilde{X}_{[T]}B}$ commutes with $\Delta_B\pr{\rho_{B|\tilde{x}_{\calT_\calS^c}}}$ and therefore 
 $\Delta_{\calT_\calS^c}\left(\bbI-\Pi_{\tilde{X}_{[T]}B}^{(\calS)}\right)=\bbI-\Pi_{\tilde{X}_{[T]}B}^{(\calS)}$;
    \item[$(b)$] follows from \eqref{eq:Pinching_Switching};
    \item[$(c)$] follows from the definition of $\Pi_{\tilde{X}_{[T]}B}^{([S])}$ in \eqref{eq:Decoding_Projection_2} and since $f(x)=x^\beta$, for $\beta\in(0\,,1]$, is a matrix monotone function \cite[Section~1.5]{QIT_Hayashi};
    \item[$(d)$] follows from \eqref{eq:Projection_Inequality_S} and since for $A,B,C\in\calP(\calH)$ and $C\preceq B$ we have $\tra(CA)\le\tra(BA)$, since $\tra((B-C)A)\ge0$ \cite[Lemma~B.5.2]{RennerDissertation};
    \item[$(e)$] follows by defining $\sigma_{\tilde{X}_{[T]}B}\triangleq\Delta_B\pr{\rho_{\pr{U_{\calT_\calS},X_{\calT_\calS}}-\pr{U_{\calT_\calS^c},X_{\calT_\calS^c}}-B}}$;
    \item[$(f)$] follows from similar steps to those that follow step $(d)$ in \eqref{eq:First_Term_Lemma}.
\end{itemize}
\setcounter{equation}{50}
To prove \eqref{eq:Lemma_Error_3} we have,
\begin{align}
    &\tra\left[\Pi_{\tilde{X}_{[T]}B}^{([S])}\left(\rho_{\tilde{X}_{[T]}}\otimes\rho_B\right)\right]\nonumber\\
    &\quad=\tra\left[\Pi_{\tilde{X}_{[T]}B}^{([S])}\left(\rho_{\tilde{X}_{[T]}}\otimes\rho_B\right)^{1-\alpha}\left(\rho_{\tilde{X}_{[T]}}\otimes\rho_B\right)^\alpha\right]\nonumber\\
    &\quad\mathop\le\limits^{(a)}2^{-(1-\alpha)\sum_{s=1}^SR_s}\nonumber\\
    &\quad\quad\times\tra\left[\Pi_{\tilde{X}_{[T]}B}^{([S])}\pr{\Delta_B\pr{\rho_{\tilde{X}_{[T]}B}}}^{1-\alpha}\left(\rho_{\tilde{X}_{[T]}}\otimes\rho_B\right)^\alpha\right]\nonumber\\
    &\quad\mathop\le\limits^{(b)}2^{-(1-\alpha)\sum_{s=1}^SR_s}\nonumber\\
    &\quad\quad\times\tra\left[\pr{\Delta_B\pr{\rho_{\tilde{X}{[T]}B}}}^{1-\alpha}\left(\rho_{\tilde{X}_{[T]}}\otimes\rho_B\right)^\alpha\right]\nonumber\\
    &\quad\mathop\le\limits^{(c)}\mu_{B,[T]}^\alpha2^{-(1-\alpha)\sum_{s=1}^SR_s}\nonumber\\
    &\quad\times2^{-\alpha\ubar{\D}_{1-\alpha}\pr{\rho_{U_{[T]}X_{[T]}B}\lVert\rho_{U_{[T]}X_{[T]}}\otimes\rho_B}},\nonumber
\end{align}where
\begin{itemize}
    \item[$(a)$] follows from the definition of $\Pi_{U_{[T]}X_{[T]}B}^{([S])}$ in \eqref{eq:Decoding_Projection_1} and since $f(x)=x^\beta$, for $\beta\in(0\,,1]$, is a matrix monotone function \cite[Section~1.5]{QIT_Hayashi};
    \item[$(b)$] follows from \eqref{eq:Projection_Inequality_1} and since for $A,B,C\in\calP(\calH)$ and $C \preceq B$ we have $\tra(CA)\le\tra(BA)$, since $\tra((B-C)A)\ge0$ \cite[Lemma~B.5.2]{RennerDissertation};
    \item[$(c)$] follows from similar steps to those that follow the step $(d)$ in \eqref{eq:First_Term_Lemma}.
\end{itemize}
We now prove \eqref{eq:Lemma_Error_4}, given in \eqref{eq:Fourth_Term_Lemma} at the bottom of the previous page,
where
\begin{itemize}
    \item[$(a)$] follows since $\Pi^{(\calS)}_{\tilde{X}_{[T]}B}$ commutes with $\rho_B$ and therefore $\Delta_B\pr{\Pi_{\tilde{X}_{[T]}B}^{(\calS)}}=\Pi_{\tilde{X}_{[T]}B}^{(\calS)}$;
    \item[$(b)$] follows from \eqref{eq:Pinching_Switching};
    \item[$(c)$] follows from \eqref{eq:Decoding_Projection_2} and since $f(x)=x^\beta$, for $\beta\in(0\,,1]$, is a matrix monotone function \cite[Section~1.5]{QIT_Hayashi};
    \item[$(d)$] follows from \eqref{eq:Projection_Inequality_S} and since for $A,B,C\in\calP(\calH)$ and $C \preceq B$ we have $\tra(CA)\le\tra(BA)$, since $\tra((B-C)A)\ge0$ \cite[Lemma~B.5.2]{RennerDissertation};
    \item[$(e)$] follows from similar steps to those that follow the step $(d)$ in \eqref{eq:Second_Term_Lemma}.
\end{itemize}

\begin{figure*}[b!]
    \hrulefill
    \setcounter{equation}{47}
\begin{align}
    &\text{\ac{RHS} of \eqref{eq:Resolv_Analysis_1}}=\frac{1}{2^{\sum_{s=1}^SR_s}}\sum_{m_{[S]}}\bbE_{\pr{m_1,\cdots,m_S}}\tra\left[\left(\rho_E^{-\frac{\alpha}{2(1+\alpha)}}\rho_{E\lvert\tilde{X}_1\pr{m_{\calI_1}},\tilde{X}_2\pr{m_{\calI_2}},\cdots,\tilde{X}_T\pr{m_{\calI_T}}}\rho_E^{-\frac{\alpha}{2(1+\alpha)}}\right)\right.\nonumber\\
    &\times\left(\rho_E^{-\frac{\alpha}{2(1+\alpha)}}\frac{1}{2^{\sum_{s=1}^SR_s}}\left(\rho_{E\lvert\tilde{X}_1\pr{m_{\calI_1}},\tilde{X}_2\pr{m_{\calI_2}},\cdots,\tilde{X}_T\pr{m_{\calI_T}}}\right.\right.\nonumber\\
    &\left.\left.\left.+\sum_{\calS\subset[S]}\sum_{m'_\calS\ne m_\calS}\rho_{E\big\lvert\tilde{X}(\calS^c)}+\sum_{\pr{m'_1,m'_2,\cdots,m'_S}\ne\pr{ m_1,m_2,\cdots,m_S}}\rho_E\right)\rho_E^{-\frac{\alpha}{2(1+\alpha)}}\right)^\alpha\right]\nonumber\\
    &\mathop\le\limits^{(d)}\frac{1}{2^{\sum_{s=1}^SR_s}}\sum_{m_{[S]}}\bbE_{\pr{m_1,\cdots,m_S}}\tra\left[\left(\rho_E^{-\frac{\alpha}{2(1+\alpha)}}\rho_{E\lvert\tilde{X}_1\pr{m_{\calI_1}},\tilde{X}_2\pr{m_{\calI_2}},\cdots,\tilde{X}_T\pr{m_{\calI_T}}}\rho_E^{-\frac{\alpha}{2(1+\alpha)}}\right)\right.\nonumber\\
    &\left.\times\left(\rho_E^{-\frac{\alpha}{2(1+\alpha)}}\frac{1}{2^{\sum_{s=1}^SR_s}}\left(\rho_{E\lvert\tilde{X}_1\pr{m_{\calI_1}},\tilde{X}_2\pr{m_{\calI_2}},\cdots,\tilde{X}_T\pr{m_{\calI_T}}}+\sum_{\calS\subset[S]}2^{\sum_{s\in\calS} R_s}\rho_{E\big\lvert\tilde{X}(\calS^c)}+2^{\sum_{s=1}^SR_s}\rho_E\right)\rho_E^{-\frac{\alpha}{2(1+\alpha)}}\right)^\alpha\right]\nonumber\\
    &\mathop=\limits^{(e)}\frac{1}{2^{\sum_{s=1}^SR_s}}\sum_{m_{[S]}}\bbE_{\pr{m_1,\cdots,m_S}}\tra\left[\left(\rho_E^{-\frac{\alpha}{2(1+\alpha)}}\rho_{E\lvert\tilde{X}_1\pr{m_{\calI_1}},\tilde{X}_2\pr{m_{\calI_2}},\cdots,\tilde{X}_T\pr{m_{\calI_T}}}\rho_E^{-\frac{\alpha}{2(1+\alpha)}}\right)\right.\nonumber\\
    &\times\left.\left(\rho_E^{-\frac{\alpha}{2(1+\alpha)}}\frac{1}{2^{\sum_{s=1}^SR_s}}\left(\rho_{E\lvert\tilde{X}_1\pr{m_{\calI_1}},\tilde{X}_2\pr{m_{\calI_2}},\cdots,\tilde{X}_T\pr{m_{\calI_T}}}+\sum_{\calS\subset[S]}2^{\sum_{s\in\calS^c} R_s}\rho_{E\big\lvert\tilde{X}(\calS)}+2^{\sum_{s=1}^SR_s}\rho_E\right)\rho_E^{-\frac{\alpha}{2(1+\alpha)}}\right)^\alpha\right]\nonumber\\
    &\mathop=\limits^{(f)}\frac{1}{2^{\sum_{s=1}^SR_s}}\sum_{m_{[S]}}\bbE_{\pr{m_1,\cdots,m_S}}\tra\left[\left(\rho_E^{-\frac{\alpha}{2(1+\alpha)}}\rho_{E\lvert\tilde{X}_1\pr{m_{\calI_1}},\tilde{X}_2\pr{m_{\calI_2}},\cdots,\tilde{X}_T\pr{m_{\calI_T}}}\rho_E^{-\frac{\alpha}{2(1+\alpha)}}\right)\right.\nonumber\\
    &\times\left.\left(\rho_E^{-\frac{\alpha}{2(1+\alpha)}}\frac{1}{2^{\sum_{s=1}^SR_s}}\left(\rho_{E\lvert\tilde{X}_1\pr{m_{\calI_1}},\tilde{X}_2\pr{m_{\calI_2}},\cdots,\tilde{X}_T\pr{m_{\calI_T}}}+\sum_{i=1}^{\card{\mathfrak{B}}}2^{\sum_{s\in\calS^c_i} R_s}\rho_{E\big\lvert\tilde{X}(\calS_i)}+2^{\sum_{s=1}^SR_s}\rho_E\right)\rho_E^{-\frac{\alpha}{2(1+\alpha)}}\right)^\alpha\right]\nonumber\\
    &\mathop\le\limits^{(g)}\frac{1}{2^{\sum_{s=1}^SR_s}}\sum_{m_{[S]}}\bbE_{\pr{m_1,\cdots,m_S}}\tra\left[\left(\rho_E^{-\frac{\alpha}{2(1+\alpha)}}\rho_{E\lvert\tilde{X}_1\pr{m_{\calI_1}},\tilde{X}_2\pr{m_{\calI_2}},\cdots,\tilde{X}_T\pr{m_{\calI_T}}}\rho_E^{-\frac{\alpha}{2(1+\alpha)}}\right)\right.\nonumber\\
    &\times\left(\rho_E^{-\frac{\alpha}{2(1+\alpha)}}\frac{1}{2^{\sum_{s=1}^SR_s}}\left(\mu_{E,\card{\mathfrak{B}}}\Delta_{E\mid\tilde{X}_{[T]}}\pr{\rho_{E\lvert\tilde{X}_1\pr{m_{\calI_1}},\tilde{X}_2\pr{m_{\calI_2}},\cdots,\tilde{X}_T\pr{m_{\calI_T}}}}\right.\right.\nonumber\\
    &\left.\left.\left.+\sum_{i=1}^{\card{\mathfrak{B}}}\mu_{E,i-1}2^{\sum_{s\in\calS^c_i} R_s}\Delta_{E\mid\tilde{X}\pr{\calS_{i-1}}}\Big(\rho_{E\big\lvert\tilde{X}(\calS_i)}\Big)+2^{\sum_{s=1}^SR_s}\rho_E\right)\rho_E^{-\frac{\alpha}{2(1+\alpha)}}\right)^\alpha\right]\nonumber\\
    &\mathop\le\limits^{(h)}\frac{1}{2^{\sum_{s=1}^SR_s}}\sum_{m_{[S]}}\bbE_{\pr{m_1,\cdots,m_S}}\tra\left[\left(\rho_E^{-\frac{\alpha}{2(1+\alpha)}}\rho_{E\lvert\tilde{X}_1\pr{m_{\calI_1}},\tilde{X}_2\pr{m_{\calI_2}},\cdots,\tilde{X}_T\pr{m_{\calI_T}}}\rho_E^{-\frac{\alpha}{2(1+\alpha)}}\right)\right.\nonumber\\
    &\times\left(\rho_E^{-\frac{\alpha}{2(1+\alpha)}}\frac{1}{2^{\sum_{s=1}^SR_s}}\left(\mu_{E,\card{\mathfrak{B}}}^\alpha\pr{\Delta_{E\mid\tilde{X}_{[T]}}\pr{\rho_{E\lvert\tilde{X}_1\pr{m_{\calI_1}},\tilde{X}_2\pr{m_{\calI_2}},\cdots,\tilde{X}_T\pr{m_{\calI_T}}}}}^\alpha\right.\right.\nonumber\\
    &\left.\left.\left.+\sum_{i=1}^{\card{\mathfrak{B}}}\mu_{E,i-1}^\alpha2^{\alpha\sum_{s\in\calS^c_i} R_s}\pr{\Delta_{E\mid\tilde{X}\pr{\calS_{i-1}}}\Big(\rho_{E\big\lvert\tilde{X}(\calS_i)}\Big)}^\alpha+2^{\alpha\sum_{s=1}^SR_s}\rho_E^\alpha\right)\rho_E^{-\frac{\alpha^2}{2(1+\alpha)}}\right)\right],\label{eq:Resolv_Analysis_2}
    \end{align}
\end{figure*}

\setcounter{equation}{45}
\section{Proof of Lemma~\ref{lemma:Resolvability}}
\label{proof:lemma_Resolvability}
From the concavity of the $\log$ function and Jensen's inequality, we have
\begin{align}
    \bbE_C\left[\ubar{\D}_{1+\alpha}\big(\tau_{E\lvert C}\lVert\rho_E\big)\right]&\le\frac{1}{\alpha}\log_2\left(\bbE_C\left[2^{\alpha\ubar{\D}_{1+\alpha}\big(\tau_{E\lvert C}\lVert\rho_E\big)}\right]\right).\label{eq:Jensen_Res}
\end{align}
Now, we bound the argument of the $\log$ function in \eqref{eq:Jensen_Res} as provided in \eqref{eq:Resolv_Analysis_1} at the bottom of the page, where
\begin{itemize}
    \item[$(a)$] follows from the linearity of the trace operation and the expectation;
    \item[$(b)$] follows by defining, for each subset $\calS \subset [S]$, the set $\tilde{X}(\calS^c)$ as the collection of channel inputs that are \emph{only} functions of message indices in $\calS^c$, and defining $\tilde{X}^c(\calS^c)$ as the complement of $\tilde{X}(\calS^c)$ with respect to the set of all channel inputs, i.e., $\tilde{X}_{[T]}([S])$. Note that $\tilde{X}^c(\calS^c)$ represents the set of channel inputs for which each input depends on at least one message $M'_s$ with $s \in \calS$;
    \item[$(c)$] follows from the linearity of the expectation and the trace operation and Jensen's inequality.
    \end{itemize}
    \begin{figure*}[b!]
    \hrulefill
    \setcounter{equation}{48}
\begin{align}
    &\text{\ac{RHS} of \eqref{eq:Resolv_Analysis_2}}\mathop=\limits^{(i)}1+\bbE_{\tilde{X}_{[T]}}\tra\left[\frac{\mu_{E,{\card{\mathfrak{B}}}}^\alpha}{2^{\alpha\sum_{s=1}^SR_s}}\rho_{E\lvert\tilde{X}_{[T]}}\pr{\Delta_{E\mid\tilde{X}_{[T]}}\pr{\rho_{E\lvert\tilde{X}_{[T]}}}}^\alpha\rho_E^{-\alpha}\right.\nonumber\\
    &\left.+\sum_{i=1}^{\card{\mathfrak{B}}}\frac{\mu_{E,i-1}^\alpha}{2^{\alpha\sum_{s\in\calS_i}R_s}}\rho_{E\lvert\tilde{X}_{[T]}} \pr{\Delta_{E\mid\tilde{X}\pr{\calS_{i-1}}}\Big(\rho_{E\big\lvert\tilde{X}(\calS_i)}\Big)}^\alpha\rho_E^{-\alpha}\right]\nonumber\\
    &\mathop=\limits^{(j)}1+\bbE_{\tilde{X}_{[T]}}\tra\left[\frac{\mu_{E,\card{\mathfrak{B}}}^\alpha}{2^{\alpha\sum_{s=1}^SR_s}}\rho_{E\lvert\tilde{X}_{[T]}}\Delta_{E\mid\tilde{X}_{[T]}}\pr{\pr{\Delta_{E\mid\tilde{X}_{[T]}}\pr{\rho_{E\lvert\tilde{X}_{[T]}}}}^\alpha\rho_E^{-\alpha}}\right]\nonumber\\
    &+\bbE_{\tilde{X}_{[T]}}\tra\left[\sum_{i=1}^{\card{\mathfrak{B}}}\frac{\mu_{E,i-1}^\alpha}{2^{\alpha\sum_{s\in\calS_i}R_s}}\rho_{E\lvert\tilde{X}_{[T]}} \Delta_{E\mid\tilde{X}\pr{\calS_{i-1}}}\pr{\pr{\Delta_{E\mid\tilde{X}\pr{\calS_{i-1}}}\Big(\rho_{E\big\lvert\tilde{X}(\calS_i)}\Big)}^\alpha\rho_E^{-\alpha}}\right]\nonumber\\
    &=1+\bbE_{\tilde{X}_{[T]}}\tra\left[\frac{\mu_{E,\card{\mathfrak{B}}}^\alpha}{2^{\alpha\sum_{s=1}^SR_s}}\rho_{E\lvert\tilde{X}_{[T]}}\Delta_{E\mid\tilde{X}_{[T]}}\pr{\pr{\Delta_{E\mid\tilde{X}_{[T]}}\pr{\rho_{E\lvert\tilde{X}_{[T]}}}}^\alpha\rho_E^{-\alpha}}\right]\nonumber\\
    &+\bbE_{\tilde{X}(\calS_i)}\tra\left[\sum_{i=1}^{\card{\mathfrak{B}}}\frac{\mu_{E,i-1}^\alpha}{2^{\alpha\sum_{s\in\calS_i}R_s}}\rho_{E\lvert\tilde{X}(\calS_i)} \Delta_{E\mid\tilde{X}\pr{\calS_{i-1}}}\pr{\pr{\Delta_{E\mid\tilde{X}\pr{\calS_{i-1}}}\Big(\rho_{E\big\lvert\tilde{X}(\calS_i)}\Big)}^\alpha\rho_E^{-\alpha}}\right]\nonumber\\
    &\mathop=\limits^{(k)}1+\bbE_{\tilde{X}_{[T]}}\tra\left[\frac{\mu_{E,\card{\mathfrak{B}}}^\alpha}{2^{\alpha\sum_{s=1}^SR_s}}\pr{\Delta_{E\mid\tilde{X}_{[T]}}\pr{\rho_{E\lvert\tilde{X}_{[T]}}}}^{1+\alpha}\rho_E^{-\alpha}\right]\nonumber\\
    &+\bbE_{\tilde{X}(\calS_i)}\tra\left[\sum_{i=1}^{\card{\mathfrak{B}}}\frac{\mu_{E,i-1}^\alpha}{2^{\alpha\sum_{s\in\calS_i}R_s}} \pr{\Delta_{E\mid\tilde{X}\pr{\calS_{i-1}}}\Big(\rho_{E\big\lvert\tilde{X}(\calS_i)}\Big)}^{1+\alpha}\rho_E^{-\alpha}\right]\nonumber\\
    &\mathop=\limits^{(\ell)}1+\tra\left[\frac{\mu_{E,\card{\mathfrak{B}}}^\alpha}{2^{\alpha\sum_{s=1}^SR_s}}\pr{\Delta_{\rho_{\tilde{X}_{[T]}}\otimes\rho_E}\pr{\rho_{\tilde{X}_{[T]}E}}}^{1+\alpha}\pr{\rho_{\tilde{X}_{[T]}}\otimes\rho_E}^{-\alpha}\right]\nonumber\\
    &+\sum_{i=1}^{\card{\mathfrak{B}}}\tra\left[\frac{\mu_{E,i-1}^\alpha}{2^{\alpha\sum_{s\in\calS_i}R_s}}\pr{\Delta_{\rho_{\tilde{X}(\calS_{i-1})}\otimes\rho_E}\pr{\rho_{\tilde{X}(\calS_i)E}}}^{1+\alpha}\pr{\rho_{\tilde{X}(\calS_i)}\otimes\rho_E}^{-\alpha}\right]\nonumber\\
    &\mathop\le\limits^{(m)}1+\frac{\mu_{E,\card{\mathfrak{B}}}^\alpha}{2^{\alpha\sum_{s=1}^SR_s}}2^{\alpha\ubar{\D}_{1-\alpha}\pr{\rho_{\tilde{X}_{[T]}E}\big\lVert\rho_{\tilde{X}_{[T]}}\otimes\rho_E}}+\sum_{i=1}^{\card{\mathfrak{B}}}\frac{\mu_{E,i-1}^\alpha}{2^{\alpha\sum_{s\in\calS_i}R_s}}2^{\alpha\ubar{\D}_{1+\alpha}\pr{\rho_{\tilde{X}\pr{\calS_i}E}\big\lVert\rho_{\tilde{X}\pr{\calS_i}}\otimes\rho_E}},\label{eq:Resolv_Analysis}
\end{align}
\end{figure*}
Now the \ac{RHS} of \eqref{eq:Resolv_Analysis_1} can be bounded further as \eqref{eq:Resolv_Analysis_2}, provided at the bottom of the page, where
\begin{itemize}
    \item[$(d)$] follows since $f(x)=x^\alpha$, for $\alpha\in(0\,,1]$, is a matrix monotone function \cite[Section~1.5]{QIT_Hayashi};
    \item[$(e)$] follows by the change of variables, and since the summation is taken over all non-empty strict subsets $\calS \subset [S]$, the sum remains unchanged under this reindexing;
    \item[$(f)$] follows by defining $\mathfrak{B} \triangleq \pr{\calS_1, \calS_2, \dots, \calS_{\card{\mathfrak{B}}}}$ as an ordered collection of all non-empty strict subsets of $[S]$, arranged such that for any $1 \le i < j \le \card{\mathfrak{B}}$, we have $\card{\calS_i} \le \card{\calS_j}$;
    \item[$(g)$] follows from \eqref{eq:Pinching_Prop0} and defining $\Delta_{E\mid\tilde{X}\pr{\calS_{i-1}}}\Big(\rho_{E\big\lvert\tilde{X}(\calS_i)}\Big)$ as the pinching map \ac{wrt} the spectral decomposition of the operator $\Delta_{E\mid\tilde{X}\pr{\calS_{i-2}}}\Big(\rho_{E\big\lvert\tilde{X}(\calS_{i-1})}\Big)$, with the convention $\tilde{X}(\calS_0)=\emptyset$ and therefore $\Delta_{E\mid\emptyset}=\Delta_E$. Also, for $i\in[\card{\mathfrak{B}}]$, define $\mu_{E,i}$ as the maximum number of the distinct eigenvalues of the operator $\br{\Delta_{E\mid\tilde{x}(\calS_{i-2})}\pr{\rho_{E\mid\tilde{x}(\calS_{i-1})}}}_{\pr{\tilde{x}(\calS_1),\dots,\tilde{x}(\calS_{i-1})}}$. This hierarchical pinching construction ensures that all resulting operators remain block-diagonal in the eigenbasis of $\rho_E$, with successive refinements within each block. Consequently, all such operators commute;
    \item[$(h)$] follows because the terms inside the second parenthesis of the trace operation commute; additionally, it follows from the inequality $(a+b)^\alpha \le a^\alpha + b^\alpha$ for $\alpha < 1$, and from the property that for $A, B, C \in \calP(\calH)$ with $C \preceq B$, we have $\tra(AC) \le \tra(AB)$, which holds since $\tra(A(B - C)) \ge 0$ \cite[Lemma~B.5.2]{RennerDissertation}.
    \end{itemize}
    \begin{figure*}[b!]
\hrulefill
\setcounter{equation}{51}
\begin{subequations}\label{eq:Pe_Covertness_2}
\begin{align}
    \bbP\left\{\hat{M}\ne M\right\}&\le2\pr{\mu_{B,[T]}^{(n)}}^\alpha2^{\alpha\pr{n\sum\limits_{s=1}^SR_s-\ubar{\D}_{1-\alpha}\pr{\rho_{\tilde{X}_{[T]}B}^{\otimes n}\big\lVert\rho_{\tilde{X}_{[T]}}^{\otimes n}\otimes\rho_B^{\otimes n}}}}\!+6\!\!\sum\limits_{\calS\subset[S]}\!\pr{\mu_{B,\calT_\calS^c}^{(n)}}^\alpha2^{\alpha\pr{n\sum\limits_{s\in\calS}R_s-\ubar{\I}_{1-\alpha}\pr{\tilde{X}_{\calT_\calS}^n;B^n\big\lvert \tilde{X}_{\calT_\calS^c}^n}}},\label{eq:Pe_Thm_2}\\
    \left\lVert\tau_{E^n}- \rho_0^{\otimes n}\right\rVert_1&\le\frac{2}{\sqrt{\alpha}}\left[\pr{\mu_{E,\card{\mathfrak{B}}}^{(n)}}^{\frac{\alpha}{2}}2^{\frac{\alpha}{2}\pr{-n\sum\limits_{s=1}^SR_s+\ubar{\D}_{1+\alpha}\pr{\rho_{\tilde{X}_{[T]}E}^{\otimes n}\big\lVert\rho_{\tilde{X}_{[T]}}^{\otimes n}\otimes\rho_E^{\otimes n}}}}\right.\nonumber\\
    &\qquad\left.+\sum_{i=1}^{\card{\mathfrak{B}}}\pr{\mu_{E,i-1}^{(n)}}^{\frac{\alpha}{2}}2^{\frac{\alpha}{2}\pr{-n\sum\limits_{s\in\calS_i}R_s+\ubar{\D}_{1+\alpha}\pr{\rho_{\tilde{X}\pr{\calS_i}E}^{\otimes n}\big\lVert\rho_{\tilde{X}\pr{\calS_i}}^{\otimes n}\otimes\rho_E^{\otimes n}}}}\right],\label{eq:Covertnes_Thm_2}
    \end{align}
\end{subequations}
\setcounter{equation}{49}
\end{figure*}
To bound the \ac{RHS} of \eqref{eq:Resolv_Analysis_2}, we have \eqref{eq:Resolv_Analysis}, given at the bottom of the page, where
\begin{itemize}
    \item[$(i)$] follows from the linearity of the trace operation, the symmetry of the codebook construction \ac{wrt} the messages, and since $\Delta_{E\mid \tilde{X}_{[T]}}\pr{\rho_{E\lvert\tilde{X}_{[T]}}}$, $\br{\Delta_{E\mid\tilde{X}\pr{\calS_{i-1}}}\Big(\rho_{E\big\lvert\tilde{X}(\calS_i)}\Big)}_{i\in[\card{\mathfrak{B}}]}$, and  $\rho_E$ commute;
    \item[$(j)$] follows from the linearity of the trace operation and the expectation and since $\Delta_{E\mid \tilde{X}_{[T]}}\pr{\rho_{E\lvert\tilde{X}_{[T]}}}$ and $\br{\Delta_{E\mid\tilde{X}\pr{\calS_{i-1}}}\Big(\rho_{E\big\lvert\tilde{X}(\calS_i)}\Big)}_{i\in[\card{\mathfrak{B}}]}$ have the same orthonormal basis as $\rho_E$ and therefore
    \begin{align}
        &\pr{\Delta_{E\mid\tilde{X}_{[T]}}\pr{\rho_{E\lvert\tilde{X}_{[T]}}}}^\alpha\rho_E^{-\alpha}\nonumber\\
        &=\Delta_{E\mid\tilde{X}_{[T]}}\pr{\pr{\Delta_{E\mid\tilde{X}_{[T]}}\pr{\rho_{E\lvert\tilde{X}_{[T]}}}}^\alpha\rho_E^{-\alpha}},\nonumber\\
        &\pr{\Delta_{E\mid\tilde{X}\pr{\calS_{i-1}}}\Big(\rho_{E\big\lvert\tilde{X}(\calS_i)}\Big)}^\alpha\rho_E^{-\alpha}\nonumber\\
        &=\Delta_{E\mid\tilde{X}\pr{\calS_{i-1}}}\pr{\pr{\Delta_{E\mid\tilde{X}\pr{\calS_{i-1}}}\Big(\rho_{E\big\lvert\tilde{X}(\calS_i)}\Big)}^\alpha\rho_E^{-\alpha}};\nonumber
    \end{align}
    \item[$(k)$] follows from \eqref{eq:Pinching_Switching};
     \item[$(\ell)$] follows from the definition of pinching operation $\Delta_E$, and since the involved states are classical-quantum states;
    \item[$(m)$] follows from the definition of $\ubar{\D}_{1+\alpha}\big(\cdot\lVert\cdot\big)$ and since for two quantum states $\rho,\sigma\in\calD(\calH)$ and quantum operation $\Gamma(\cdot):\calL(A)\to\calL(B)$, we have \cite{Branum96,Frank13},
    \begin{align}
        \ubar{\D}_{1+\alpha}\big(\Gamma(\rho)\lVert\Gamma(\sigma)\big)\le\ubar{\D}_{1+\alpha}\big(\rho\lVert\sigma\big).\label{eq:Monotonicity_SRenyi}
    \end{align}
\end{itemize}
Now, substituting \eqref{eq:Resolv_Analysis} into \eqref{eq:Jensen_Res} and using the inequality $\log_2(1+x) \le \frac{x}{\ln 2}$ completes the proof of Lemma~\ref{lemma:Resolvability}.

\section{Proof of Theorem~\ref{thm:Achievable_Asymp}}
\label{proof:thm:Achievable_Asymp}
We begin by bounding $\mu_B$, $\mu_{B,\calT}$, $\mu_E$, and $\mu_{E,i}$, for $i\in[\abs{\mathfrak{B}}]$, as defined in \eqref{eq:gs}, under the assumption that the transmitters use the channel $n$ times in an \ac{iid} manner. That is, we consider the scenario in which there are $n$ \ac{iid} copies of the states $\rho_{U_{[T]}X_{[T]}B}$ and $\rho_{\pr{U_\calT,X_\calT}-\pr{U_{\calT^c},X_{\calT^c}}-B}$. Let,
\begin{subequations}\label{eq:gs_n}
    \begin{align}
        \mu_B^{(n)}&\triangleq\text{number of distinct eigenvalues of } \rho_B^{\otimes n},\label{eq:g2_n}\\
        \mu_{B,\calT}^{(n)}&\triangleq\text{maximum number of distinct eigenvalues of } \nonumber\\
        &\qquad\left\{\Delta_{B^n}\pr{\rho_{B^n\lvert \tilde{x}_\calT^n}}\right\}_{\tilde{x}_\calT^n},\label{eq:g3_n}\\
        \mu_E^{(n)}&\triangleq\text{number of distinct eigenvalues of } \rho_E^{\otimes n},\label{eq:g1_n}\\
        \mu_{E,i}^{(n)}&\triangleq \text{maximum number of the distinct eigenvalues of the} \nonumber\\
        &\hspace{-6mm}\text{operator}\,\,\br{\Delta_{E^n\mid\tilde{x}^n(\calS_{i-2})}\pr{\rho_{E^n\mid\tilde{x}^n(\calS_{i-1})}}}_{\pr{\tilde{x}^n(\calS_1),\dots,\tilde{x}^n(\calS_{i-1})}}\!,\label{eq:g4_n}
    \end{align}where $\calT\subset[T]$, $\rho_{B^n\lvert \tilde{x}_\calT^n}\triangleq\rho_{B\lvert \tilde{x}_{\calT,1}}\otimes\rho_{B\lvert\tilde{x}_{\calT,2}}\otimes\cdots\otimes\rho_{B\lvert\tilde{x}_{\calT,n}}$, the maximization in \eqref{eq:g3_n} is over $\tilde{x}_\calT^n$, and the maximization in \eqref{eq:g4_n} is over $\pr{\tilde{x}^n(\calS_1),\dots,\tilde{x}^n(\calS_{i-1})}$.
\end{subequations}
The following lemmas are essential in our proof.
\begin{lemma}
\label{eq:Lemma_Pincing_Eigen_Value_Bounds}
Let $ \Delta_{E^n} $ and $ \Delta_{B^n} $ be the pinching maps \ac{wrt} the spectral decomposition of $\rho_E^{\otimes n}$ and $\rho_B^{\otimes n}$, respectively, and let $ \mu_B^{(n)} $, $ \mu_{B,\calT}^{(n)} $, $ \mu_E^{(n)} $, and $\mu_{E,i}^{(n)}$ be as defined in \eqref{eq:gs_n}. For $t\in[T]$, denote by $ d_{U_t} $, $ d_{X_t} $, and $ d_B $ the dimensions of the Hilbert spaces $\calH_{U_t}$, $\calH_{X_t}$, and $\calH_b$, respectively. Then we have the following bounds
\begin{align*}
    \mu_B^{(n)} &\le (n+1)^{d_B - 1},\\
    \mu_{B,\calT}^{(n)} &\le (n+1)^{\frac{\pr{\prod\limits_{t\in\calT}d_{U_t}d_{X_t}}\pr{d_B + 2}\pr{d_B - 1}}{2}},\\
    \mu_E^{(n)} &\le (n+1)^{d_E - 1},\\
    \mu_{E,i}^{(n)} &\le (n+1)^{\frac{\pr{\prod\limits_{j=1}^{i-1}\prod\limits_{\tilde{x}_t\in\tilde{x}(\calS_j)}d_{U_t}d_{X_t}}\pr{d_E + 2}\pr{d_E - 1}}{2}},
\end{align*}where $\calT\subset[T]$ and $i\in\sbr{\abs{\mathfrak{B}}}$.
\end{lemma}
The proof of Lemma~\ref{eq:Lemma_Pincing_Eigen_Value_Bounds} is similar to the proof of \cite[Proposition~12]{Salek25} and is omitted for brevity.
\begin{lemma}
\label{eq:Ryini_Conditional_MI}
For each $\calT\subset[T]$ and $\alpha\in\pr{\frac{-1}{2}:0}\cup\pr{0:\infty}$, we have
\begin{align*}
    &\ubar{\I}_{1-\alpha}\pr{U_\calT^n,X_\calT^n;B^n\lvert U_{\calT^c}^n,X_{\calT^c}^n}_{\rho_{U_{[T]}X_{[T]}B}^{\otimes n}\lvert\rho_{U_{[T]}X_{[T]}}^{\otimes n}}\nonumber\\
    &=n\ubar{\I}_{1-\alpha}\pr{U_\calT,X_\calT;B\lvert U_{\calT^c},X_{\calT^c}}_{\rho_{U_{[T]}X_{[T]}B}\lvert\rho_{U_{[T]}X_{[T]}}}.
\end{align*}
\end{lemma}The proof Lemma~\ref{eq:Ryini_Conditional_MI} is similar to the proof of \cite[Lemma~4]{AshnuHayashi2020} and is omitted for brevity. 

\setcounter{equation}{52}
From Theorem~\ref{thm:Achievable_One_Shot} when the transmitters use the channel $\calN_{A_{[T]}\to BE}$ $n$ times independently, there exists a code such that we have \eqref{eq:Pe_Covertness_2}, at the bottom of the page. By Lemma~\ref{eq:Lemma_Pincing_Eigen_Value_Bounds}, the quantities $\mu_{B,[T]}^{(n)}$, $\mu_{B,\calT_\calS^c}^{(n)}$, $\mu_{E,\abs{\mathfrak{B}}}^{(n)}$, and $\mu_{E,i}^{(n)}$, for $i\in[\abs{\mathfrak{B}}]$, are polynomial in $n$. Moreover, by Lemma~\ref{eq:Ryini_Conditional_MI}, the sandwiched R\'enyi mutual information quantities are additive under tensor products:
\begin{align}
\ubar{\I}_{1-\alpha}\pr{\tilde{X}_{\calT_\calS}^n;B^n\big\lvert \tilde{X}_{\calT_\calS^c}^n}=n\ubar{\I}_{1-\alpha}\pr{\tilde{X}_{\calT_\calS};B\big\lvert \tilde{X}_{\calT_\calS^c}},\nonumber
\end{align}
and similarly for the sandwiched R\'{e}nyi relative entropy terms. 
Also, by the continuity of the sandwiched R\'{e}nyi relative entropy and sandwiched R\'{e}nyi mutual information, defined in \eqref{eq:Sand_alpha_diver}, at order one, we have $\lim_{\alpha\to0}
\ubar{\D}_{1+\alpha}(\rho\lVert\sigma)=D(\rho\lVert\sigma)$
and $\lim_{\alpha\to0}
\ubar{\I}_{1-\alpha}(U;B|V)=I(U;B|V)$ \cite{Wilde14,Sandwiched10}. Combining these facts with \eqref{eq:Pe_Covertness_2}, it follows that, for each $\calS\subseteq[S]$, when $n\to\infty$ and $\alpha\to0$, there exists a sequence of codes for which the error probability and the covertness metric vanish provided that
\begin{subequations}
\begin{align}
    \sum\limits_{s\in\calS}R_s&\le I\pr{\tilde{X}_{\calT_\calS};B\lvert\tilde{X}_{\calT_\calS^c}},\label{eq:Initial_Bound_Rel}\\
    \sum\limits_{s\in\calS}R_s&\ge I\pr{X(\calS);E},\label{eq:Initial_Bound_Res}
\end{align}
\end{subequations}where $\calT_\calS$ denotes the set of channel inputs for which each channel input has access to at least one message in the message set $M_\calS$, and $\tilde{X}\pr{\calS}$ denotes the set of channel inputs for which each channel input has only access to a subset of messages $M_{\calS'}$, for some $\calS' \subseteq \calS$. 
From Lemma~\ref{lemma:One_Message_One_Vertex}, each message is at most a private message of one vertex; therefore, every subset $\calS$ of the messages corresponds to a subset of the vertices with private messages. Also, as discussed in Section~\ref{sec:Special_Message_Hierarchy}, for every subset $\calS\subset[S]$ of the messages, the set of all vertices that have access to at least one message that belongs to the subset $\calS$ corresponds to a proper ancestral sub-graph. 
Now if we denote all distinct subsets of the messages that lead to the same proper ancestral sub-graph $\Omega\in\mathfrak{O}$ with $\calS_1,\calS_2,\dots,\calS_j$, for some $j\in\bbN^+$, there exists a subset $\calS_{i'}$, for some  $i'\in[j]$, that contains all messages which are associated, as private messages, with the vertices in the proper ancestral sub-graphs $\Omega$. Note that the \ac{RHS} of the rate constraint in \eqref{eq:Initial_Bound_Rel} for all the subsets in $\calS_i$, with $i\in[j]$, are equal. Therefore, the rate constraints corresponding to all the subsets in $\calS_{[j]}\backslash\calS_{i'}$ are redundant because of $\calS_{i'}$. This shows that for each proper ancestral sub-graph, we have only one active rate constraint in \eqref{eq:Initial_Bound_Rel}. 
For example, consider the \ac{MAC} depicted in Fig.~\ref{fig:4M6T_Example}. The set of all vertices that have access to at least one message that belongs to the message sets $\calS_1\triangleq\br{M_2,M_5}$ and $\calS_2\triangleq\br{M_2}$ is the same proper ancestral sub-graph, denoted by $\Omega'\triangleq\br{v_2,v_5,v_4,v_7,v_0}$. Note that the message set $\calS_1$ includes all messages which are associated, as private messages, with the vertices in the proper ancestral sub-graphs $\Omega'$. The rate constraints in \eqref{eq:Initial_Bound_Rel} corresponding to these two subsets are $R_2+R_5<I(X_2,X_4,X_5,X_7;B|X_1,X_3,X_6,X_8)$ and $R_2<I(X_2,X_4,X_5,X_7;B|X_1,X_3,X_6,X_8)$, respectively. Clearly, the latter constraint is redundant because of the former. Now, since $U_\Omega-X_\Omega-B$, $t\in[T]$, forms a Markov chain, we can rewrite \eqref{eq:Initial_Bound_Rel}~as
\begin{align}
    \sum_{v\in\Omega}R(v)&< I\left(X_\Omega;B\lvert X_{\Omega^c},U_{\Omega^c\cap\bbM}\right),\nonumber
\end{align}for all $\Omega\in\mathfrak{O}$ of the associated message graph $\Gamma$. 

We now argue that \eqref{eq:Initial_Bound_Res} is equivalent to
\begin{align}
        \sum\limits_{v\in\Xi}R(v)\ge I\pr{X_\Xi;E},\label{eq:Initial_Bound_Res_4}
    \end{align}for all proper descendant sub-graphs $\Xi\in\mathfrak{X}$ of the associated message graph $\Gamma$. 
We note the following facts observed from the bound in \eqref{eq:Initial_Bound_Res}:
    \begin{enumerate}[{Fact} 1:]
        \item\label{fct:1}If two subsets $\calS, \calS' \subseteq [S]$ satisfy $\calS \subseteq \calS'$ and $\tilde{X}(\calS) = \tilde{X}(\calS')$, then the \ac{LHS} of the rate constraint associated with $\calS'$ in \eqref{eq:Initial_Bound_Res} is bigger than that associated with $\calS$.
        \item\label{fct:2} For every message set $\calS \subseteq [S]$, if all elements of $\calS$ correspond to private messages of non-leaf vertices, then $X(\calS) = \emptyset$. This follows from the fact that the channel input of every non-leaf vertex is a function of the private message of at least one leaf vertex.
        \item\label{fct:3} If a message set $\calS'\subseteq [S]$ includes at least one private message of a vertex whose descendant set contains a vertex that its private message is not in $\calS'$, then there exists a message set $\calS\subset\calS'$ such that the \ac{LHS} of the rate constraint in \eqref{eq:Initial_Bound_Res} corresponding to $\calS'$ is bigger than that of $\calS$ and $\tilde{X}(\calS) = \tilde{X}(\calS')$. This is because we can construct a smaller set $\calS \subset \calS'$ by removing these messages from $\calS'$, which results in $\tilde{X}(\calS) = \tilde{X}(\calS')$. Therefore, from Fact~\ref{fct:1}, the \ac{LHS} of the rate constraint in \eqref{eq:Initial_Bound_Res} corresponding to $\calS'$ is bigger than that of the set $\calS$. For example, consider the \ac{MAC} depicted in Fig.~\ref{fig:4M6T_Example} and the message set $\calS' \triangleq \{M_1, M_2, M_5\}$. Note that $X(\calS') = \{X_1, X_2\}$, and that $M_5$ is the private message of vertex $v_5$, whose descendant set includes vertex $v_3$. The private message of $v_3$ is $M_3$, which does not belong to $\calS'$. Thus, we can construct the reduced set $\calS \triangleq \{M_1, M_2\}$ by removing $M_5$ from $\calS'$, while still preserving the input set: $X(\calS) = X(\calS') = \{X_1, X_2\}$.
        \item\label{fct:4}If two subsets of vertices $\calV,\calV'\subseteq [T]$ include the same set of vertices with private messages, denoted by $\calS, \calS' \subseteq [S]$, with $\calS=\calS'$, and the set of vertices without private messages in $\calV'$ is a subset of that in $\calV$, then $X_{\calV'}\subseteq X_{\calV}$ and the \ac{LHS} of rate constraint in \eqref{eq:Initial_Bound_Res} corresponding to the private messages of $\calS$ and $\calS'$ satisfy $\sum_{s\in\calS} R(s) = \sum_{s\in\calS'} R(s)$.
    \end{enumerate}We now argue that the rate constraints not corresponding to the proper descendant sub-graphs, as defined in Definition~\ref{defi:Proper_Ance_Leaf_Path}, are redundant. Since every subset $\calS$ of the messages corresponds to a subset of vertices with private messages, we can analyze the problem in terms of subsets of these vertices rather than directly considering subsets of messages. From Fact~\ref{fct:2}, the subsets of vertices with private messages that do not include a leaf are redundant. This directly implies the first condition of Definition~\ref{defi:Proper_Ance_Leaf_Path}, which requires that any valid subset must include at least one leaf. Also, following a similar argument as in Fact~\ref{fct:3}, one can show that for any subset $\calV'$ of vertices, with private messages set $\calS'$, that includes a vertex whose descendant set contains a vertex with a private message not included in $\calS'$, the corresponding rate constraint becomes redundant. This directly implies the second condition of Definition~\ref{defi:Proper_Ance_Leaf_Path}. Similarly, following the argument used in Fact~\ref{fct:4}, one can show that Fact~\ref{fct:4} directly implies the third condition of Definition~\ref{defi:Proper_Ance_Leaf_Path}. 
    We now observe that if a subset $\calV'$ of the vertices with private messages, which include a non-empty subset of the leaves, consists only of vertices whose descendant sets do not contain any vertex with a private message outside $\calS'$, which is the private message set of $\calV'$, then $\calV'$, together with the vertices that do not have a private message and whose descendant sets also do not contain any vertex with a private message outside $\calS'$, forms a proper descendant sub-graph.
    Therefore, \eqref{eq:Initial_Bound_Res} is equivalent to \eqref{eq:Initial_Bound_Res_4}.

 \section{Converse Proof of Theorem~\ref{thm:Achievable_Asymp_cq}}
\label{proof:thm:Achievable_Asymp_cq}
\subsection{Entropy Bounds}
We establish a few preliminary bounds to show that quantum states close to the product states in the trace distance exhibit information-theoretic properties approximating those of truly product states. The following lemma is a consequence of \cite[Exercise~9.1.9]{Wilde_Book}, which states that the trace distance between two sequences of quantum states upper bounds the trace distance between the individual states at a random time index $Q$, where $Q$ is independent of the sequences.
\begin{lemma}[Trace Distance of Random Sample]
\label{lemma:Lemma_Product}
    Let $Q\in[n]$ be a classical random time index distributed according to $p_Q$. Also, let $\rho_{E^n}$ and $\sigma_{E^n}$ be two quantum states independent of $Q$, i.e., $\rho_{E^n,Q}=\rho_{E^n}\otimes\sum_{i=1}^np_Q(i)\den{i}{i}$ and $\sigma_{E^n,Q}=\sigma_{E^n}\otimes\sum_{i=1}^np_Q(i)\den{i}{i}$. Then,
    \begin{align}
    \lVert\rho_{E_Q}-\sigma_{E_Q}\rVert_1\le\lVert\rho_{E^n}-\sigma_{E^n}\rVert_1.\nonumber%
    \end{align}
\end{lemma}
\begin{proof}
    The proof follows from \cite[Exercise~9.1.9]{Wilde_Book}, by defining a channel that takes $\rho_{E^n,Q}$ and $\sigma_{E^n,Q}$ as inputs and outputs $\rho_{E_Q}$ and $\sigma_{E_Q}$, and from the fact that $\rVert\rho_{E^n,Q}-\sigma_{E^n,Q}\rVert_1=\rVert\rho_{E^n}-\sigma_{E^n}\rVert_1$.
\end{proof}
Now we build on the fact that we can upper-bound the difference in entropy in terms of trace distance for finite states \cite[Theorem~11.10.2]{Wilde_Book}.
\begin{lemma}[Timing mutual information of nearly product states]
\label{lemma:Timing_Lemma}
    For any sequence of quantum states $\rho_{E^n}$ and a quantum state $\sigma_E$, where $\rho_{E_i},\sigma_E\in\calD(\calH)$, for $i\in[n]$, if $\lVert\rho_{E^n}-\sigma_E^{\otimes n}\rVert_1\le\epsilon$, then
    \begin{subequations}
    \begin{align}
        \frac{1}{n}\sum_{i=1}^nI_\rho\pr{E_i;E^{i-1}}&\le2\epsilon\log\pr{\text{dim}\pr{\calH}}+2h_2(\epsilon),\label{eq:Sum_MI_1}\\
        \frac{1}{n}\sum_{i=1}^nI_\rho\pr{E_i;E_{i+1}^n}&\le2\epsilon\log\pr{\text{dim}\pr{\calH}}+2h_2(\epsilon).\label{eq:Sum_MI_2}
    \end{align}
    \end{subequations}Also, for any classical \ac{RV} $Q\in[n]$ independent of $E^n$ we have
    \begin{align}
        I_\rho\pr{Q;E_Q}\le2\epsilon\log\pr{\text{dim}\pr{\calH}}+2h_2(\epsilon).\label{eq:Tim_Indepen}
    \end{align}
\end{lemma}
\begin{proof}
    From Lemma~\ref{lemma:Lemma_Product} and \cite[Theorem~11.10.2]{Wilde_Book} we have
    \begin{subequations}
    \begin{align}
        \left\lvert H_\rho\pr{E^n} - H_\sigma\pr{E^n} \right\rvert&\le\epsilon\log\pr{\text{dim}\pr{\calH^n}-1}+h_2(\epsilon),\label{eq:First_FAI}\\
        \left\lvert H_\rho\pr{E_Q} - H_\sigma\pr{E} \right\rvert&\le\epsilon\log\pr{\text{dim}\pr{\calH}-1}+h_2(\epsilon),\label{eq:Second_FAI}\\
        \left\lvert H_\rho\pr{E_i} - H_\sigma\pr{E} \right\rvert&\le\epsilon\log\pr{\text{dim}\pr{\calH}-1}+h_2(\epsilon),\label{eq:Third_FAI}
    \end{align}where $i\in[n]$ and $h_2(\epsilon)=-\epsilon\log(\epsilon)-(1-\epsilon)\log(1-\epsilon)$.
    \end{subequations}
Now we have,
    \begin{align}
        &\frac{1}{n}\sum_{i=1}^nI_\rho\pr{E_i;E^{i-1}}\nonumber\\
        &\mathop=\limits^{(a)}\frac{1}{n}\pr{\pr{\sum_{i=1}^nH_\rho(E_i)}-H_\rho(E^n)}\nonumber\\
        &\mathop=\limits^{(b)}\frac{1}{n}\left(\!\pr{\sum_{i=1}^nH_\rho(E_i)}\!-H_\rho(E^n)+\!H_\sigma(E^n)-\!\sum_{i=1}^nH_\sigma(E_i)\right)\nonumber\\
        &\le\frac{1}{n}\left\lvert H_\rho(E^n)-H_\sigma(E^n) \right\rvert + \frac{1}{n}\sum_{i=1}^n\left\lvert H_\rho(E_i)-H_\sigma(E_i) \right\rvert\nonumber\\
        &\mathop\le\limits^{(c)}\frac{1}{n}\epsilon\log\pr{\text{dim}\pr{\calH^n}-1}+h_2(\epsilon)\nonumber\\
        &\quad+\epsilon\log\pr{\text{dim}\pr{\calH}-1}+h_2(\epsilon)\nonumber\\
        &\le2\epsilon\log\pr{\text{dim}\pr{\calH}}+2h_2(\epsilon),\nonumber
    \end{align}where
\begin{itemize}
    \item[$(a)$] follows since
    \begin{align}
        &\frac{1}{n}\sum_{i=1}^nI_\rho\pr{E_i;E^{i-1}}\nonumber\\
        &\quad=\frac{1}{n}\sum_{i=1}^n\pr{H_\rho\pr{E_i}+H_\rho\pr{E^{i-1}}-H_\rho\pr{E^i}},\nonumber\\
        &\sum_{i=1}^nH_\rho\pr{E^{i-1}}\nonumber\\
        &\quad=H_\rho\pr{E^1}+H_\rho\pr{E^2}+\cdots+H_\rho\pr{E^{n-1}},\nonumber\\
        &\sum_{i=1}^nH_\rho\pr{E^i}=H_\rho\pr{E^1}+H_\rho\pr{E^2}+\cdots+H_\rho\pr{E^n};\nonumber
    \end{align}
    \item[$(b)$] follows since $\sigma_{E^n}=\sigma_E^{\otimes n}$ is a product state and therefore $H_\sigma(E^n)=\sum_{i=1}^nH_\sigma(E_i)$; 
    \item[$(c)$] follows from \eqref{eq:First_FAI} and \eqref{eq:Third_FAI}.
    \end{itemize}Similarly, one can prove \eqref{eq:Sum_MI_2}.

    Now let the distribution of $Q$ be denoted by $p_Q$, from the independence of $E^n$ and $Q$ we have
    \begin{align}
        H_\rho\pr{E_Q\lvert Q}=\sum_{i=1}^np_Q(i)H_\rho\pr{E_i}.\nonumber%
    \end{align}
    \begin{figure*}[b!]
    \hrulefill
    \setcounter{equation}{59}
    \begin{align}
        \rho_{U_{[T]}X_{[T]}BE}&\triangleq\sum_{i=1}^n\sum_{u_{[T],i},x_{[T],i}}\prod_{t=1}^T\frac{1}{n}p_{U_{t,i}X_{t,i}\lvert U_{\bbD_t,i}X_{\bbD_t,i}}(u_{t,i},x_{t,i}\lvert u_{\bbD_t,i},x_{\bbD_t,i})\den{u_{[T],i}x_{[T],i}}{u_{[T],i}x_{[T],i}}\otimes\calN_{X_{[T]}\to BE}\pr{x_{[T],i}};\label{eq:Joint_Time_Sharing_State}
    \end{align}
    \setcounter{equation}{57}
\end{figure*}Therefore,
    \begin{align}
        &\left\lvert H_\rho\pr{E_Q\lvert Q} - H_\sigma(E) \right\rvert\nonumber\\
        &=\left\lvert \sum_{i=1}^np_Q(i)\pr{H_\rho\pr{E_i} - H_\sigma(E)} \right\rvert\nonumber\\
        &\le\sum_{i=1}^np_Q(i)\left\lvert H_\rho\pr{E_i} - H_\sigma(E) \right\rvert\nonumber\\
        &\mathop\le\limits^{(a)}\sum_{i=1}^np_Q(i)\pr{\epsilon\log\pr{\text{dim}\pr{\calH}-1}+h_2(\epsilon)}\nonumber\\
        &=\epsilon\log\pr{\text{dim}\pr{\calH}-1}+h_2(\epsilon)\nonumber\\
        &\le\epsilon\log\pr{\text{dim}\pr{\calH}}+h_2(\epsilon),\label{eq:Last_Timing_MI}
    \end{align}where $(a)$ follows from \eqref{eq:Third_FAI}. Now, combining \eqref{eq:Last_Timing_MI} with \eqref{eq:Second_FAI} completes the proof \eqref{eq:Tim_Indepen}.
\end{proof}

\subsection{Converse}
Consider any sequence of $(2^{nR_1},2^{nR_2},\cdots,2^{nR_S},n)$ codes for a \ac{MAC} with a general message set, that simultaneously satisfies the reliability constraint $\bbP\left\{M_{[S]}\ne\hat{M}_{[S]}\right\}\le\delta_n$, where $\lim\limits_{n\to\infty}\delta_n=0$, and the covertness constraint $\left\lVert\rho_{E^n}-\rho_0^{\otimes n}\right\lVert_1\le\epsilon_n$, where $\lim\limits_{n\to\infty}\epsilon_n=0$. 

\textit{Delta Rate Region:} We first define a $\calQ_\delta$ region, which expands the region in \eqref{eq:CJ_cq}, as follows
\begin{align}
  \calQ_\delta= \left.\begin{cases}(R_1,\cdots,R_S)\geq 0: \exists\, p_{\tilde{X}_{[T]}BE}\in\calG_\delta:\\
  \sum\limits_{v\in\Omega}R(v)\le I\left(X_\Omega;B\lvert X_{\Omega^c},U_{\Omega^c\cap\bbM}\right)+\delta\\
\end{cases}\right\},\nonumber
\end{align}
where,
\begin{align}
  \calG_\delta \triangleq \left.\begin{cases}p_{\tilde{X}_{[T]}BE}:\\
p_{\tilde{X}_{[T]}BE}=\sum\limits_{\tilde{x}_{[S]}}\prod_{t=1}^T p_{\tilde{X}_t\lvert \tilde{X}_{\bbD_t}}\pr{\tilde{x}_t\lvert \tilde{x}_{\bbD_t}}\calN_{X_{[T]}\to BE}\\
\sum\limits_{v\in\Xi}R(v)\ge I\pr{X_\Xi;E}-2\delta\\
\left\lVert\rho_E-\rho_0\right\lVert_1\le\delta\\
\end{cases}\hspace{-3mm}\right\}.\nonumber
\end{align}
Now, for any $\delta>0$, we show that when a tuple $(R_1,\cdots,R_S)$ is achievable, then $(R_1,\cdots,R_S)\in\calQ_\delta$. For any $\epsilon_n,\delta_n>0$, and proper ancestral sub-graph $\Omega\in\mathfrak{O}$, we upper-bound the rate as follows,
\begin{align}
    &n\sum\limits_{v\in\Omega}R(v)\nonumber\\
    &\mathop=\limits^{(a)}H\pr{M(\Omega)}\nonumber\\
    &\mathop=\limits^{(b)}H\pr{M(\Omega)\lvert M(\Omega^c),X_{\Omega^c}^n}\nonumber\\
    &=I\pr{M(\Omega);\hat{M}(\Omega)\lvert M(\Omega^c),X_{\Omega^c}^n}\nonumber\\
    &\quad+H\pr{M(\Omega)\lvert \hat{M}(\Omega), M(\Omega^c),X_{\Omega^c}^n}\nonumber\\
    &\mathop\le\limits^{(c)} I\pr{M(\Omega);B^n\lvert M(\Omega^c),X_{\Omega^c}^n}+n\delta_n\nonumber\\
    &=\sum_{i=1}^n I\pr{M(\Omega);B_i\lvert M(\Omega^c),X_{\Omega^c}^n,B^{i-1}}+n\delta_n\nonumber\\
    &\le\sum_{i=1}^n I\pr{M(\Omega),X_{\Omega,i};B_i\lvert M(\Omega^c),X_{\Omega^c}^n,B^{i-1}}+n\delta_n\nonumber\\
    &=\sum_{i=1}^n \big[H\pr{B_i\lvert M(\Omega^c),X_{\Omega^c}^n,B^{i-1}}\nonumber\\
    &\quad-H\pr{B_i\lvert M(\Omega^c),X_{\Omega^c}^n,B^{i-1},M(\Omega),X_{\Omega,i}}\big]+n\delta_n\nonumber\\
    &\mathop=\limits^{(d)}\sum_{i=1}^n \big[H\pr{B_i\lvert M(\Omega^c),X_{\Omega^c}^n,B^{i-1}}\nonumber\\
    &\quad-H\pr{B_i\lvert M_{\Omega^c\cap\bbM}^\bbD,X_{\Omega,i},X_{\Omega^c,i}}\big]+n\delta_n\nonumber\\
    &\le\sum_{i=1}^n \big[H\pr{B_i\lvert M_{\Omega^c\cap\bbM}^\bbD,X_{\Omega^c,i}}\nonumber\\
    &\quad-H\pr{B_i\lvert M_{\Omega^c\cap\bbM}^\bbD,X_{\Omega,i},X_{\Omega^c,i}}\big]+n\delta_n\nonumber\\
    &=\sum_{i=1}^n I\pr{X_{\Omega,i};B_i\lvert M_{\Omega^c\cap\bbM}^\bbD,X_{\Omega^c,i}}+n\delta_n\nonumber\\
    &\mathop=\limits^{(e)}\sum_{i=1}^n I\pr{X_{\Omega,i};B_i\lvert U_{\Omega^c\cap\bbM,i},X_{\Omega^c,i}}+n\delta_n\nonumber\\
    &\mathop=\limits^{(f)}n\!\sum_{i=1}^n p(Q=i)I\pr{X_{\Omega,Q};B_Q\lvert U_{\Omega^c\cap\bbM,Q},X_{\Omega^c,Q},Q=i}\!+\!n\delta_n\nonumber\\
    &=nI\pr{X_{\Omega,Q};B_Q\lvert U_{\Omega^c\cap\bbM,Q},X_{\Omega^c,Q},Q}+n\delta_n\nonumber\\
    &\mathop=\limits^{(g)} nI\pr{X_\Omega;B\lvert U_{\Omega^c\cap\bbM},X_{\Omega^c}}+n\delta_n\nonumber\\
    &\mathop\le\limits^{(h)}nI\pr{X_\Omega;B\lvert U_{\Omega^c\cap\bbM},X_{\Omega^c}}+n\delta,\label{eq:Upper_R_cq}
\end{align}where
\begin{itemize}
    \item[$(a)$] follows by defining $M(\Omega)$ as the set of private messages of the vertices of the proper ancestral sub-graph $\Omega$, i.e., $\br{m(k):v_k\in\Omega}$;
    \item[$(b)$] follows by defining $M(\Omega^c)$ as the complement of the set $M(\Omega)$ with respect to $M_{[S]}$. By definition, the inputs $X_\Omega^n$ are the channel inputs that depend on at least one message in $M(\Omega)$. Therefore, the inputs $X_{\Omega^c}^n$ depend only on messages in $M(\Omega^c)$, and are thus independent of $M(\Omega)$;
    \item[$(c)$] follows from data processing inequality and Fano's inequality;
    \item[$(d)$] follows since $\pr{M_{[S]},X_{\Omega^c}^{i-1},X_{\Omega^c,i+1}^n,B^{i-1}}-\pr{X_{\Omega^c,i},X_{\Omega,i}}-B_i$ forms a Markov chain and defining $M_\calS^\bbD\triangleq\br{M_{\bbD_t}:v_t\in\calS}\cup\br{M_t:v_t\in\calS}$;
    \item[$(e)$] follows by defining $U_{\Omega^c\cap\bbM,i}\triangleq M_{\Omega^c\cap\bbM}^\bbD$;
    \item[$(f)$] follows by defining $Q$ as \ac{RV} with uniform distribution over $[n]$ and independent of all the other involved \acp{RV};
    \item[$(g)$] follows by defining the joint state $\rho_{U_{[T]}X_{[T]}BE}$ by identifying $X_\Omega\triangleq \pr{X_{\Omega,Q},Q},B\triangleq B_Q,E\triangleq E_Q,U_{\Omega^c}\triangleq U_{\Omega^c,Q}$, and $X_{\Omega^c}\triangleq\pr{X_{\Omega^c,Q},Q}$, as \eqref{eq:Joint_Time_Sharing_State}, provided at the bottom of the page;
    \setcounter{equation}{60}
    \item[$(h)$] follows by defining
    \begin{equation}
        \delta\triangleq\max\br{\delta_n,\epsilon'_n,\sqrt{\frac{8}{n\beta}}\epsilon_n},\label{eq:delta}
    \end{equation}
    where $\beta$ is the minimum eigenvalue of $\rho_0^{\otimes n}$ and $\epsilon'_n\triangleq2\epsilon_n\log\pr{\text{dim}\pr{\calH}}+2h_2(\epsilon_n)$.
\end{itemize}
Now for any proper descendant sub-graph $\Xi\in\mathfrak{X}$, we lower-bound the rate as follows,
\begin{align}
    n\sum\limits_{v\in\Xi}R(v)&\mathop=\limits^{(a)} H\pr{M\pr{\Xi}}\nonumber\\
    &\ge I\pr{M(\Xi);E^n}\nonumber\\
    &\mathop=\limits^{(b)} I\pr{M(\Xi),X_\Xi^n;E^n}\nonumber\\
    &\ge I\pr{X_\Xi^n;E^n}\nonumber\\
    &=\sum_{i=1}^nI\pr{X_\Xi^n;E_i\lvert E^{i-1}}\nonumber\\
    &=\sum_{i=1}^n\sbr{H\pr{E_i\lvert E^{i-1}}-H\pr{E_i\lvert E^{i-1},X_\Xi^n}}\nonumber\\
    &\ge\sum_{i=1}^n\sbr{H\pr{E_i\lvert E^{i-1}}-H\pr{E_i\lvert X_{\Xi,i}}}\nonumber\\
    &\mathop\ge\limits^{(c)}\sum_{i=1}^n\sbr{H\pr{E_i}-H\pr{E_i\lvert X_{\Xi,i}}}-\epsilon'_n\nonumber\\
    &=\sum_{i=1}^nI\pr{X_{\Xi,i};E_i}-\epsilon'_n\nonumber\\
    &\mathop=\limits^{(d)}n\sum_{i=1}^np(Q=i)I\pr{X_{\Xi,Q};E_Q\lvert Q=i}-\epsilon'_n\nonumber\\
    &=nI\pr{X_{\Xi,Q};E_Q\lvert Q}-\epsilon'_n\nonumber\\
    &\mathop\ge\limits^{(e)}nI\pr{X_{\Xi,Q},Q;E_Q}-2\epsilon'_n\nonumber\\
    &\mathop=\limits^{(f)}nI_\rho\pr{X_\Xi;E}-2\epsilon'_n\nonumber\\
    &\mathop\ge\limits^{(g)}nI_\rho\pr{X_\Xi;E}-2\delta,\label{eq:Lower_R_cq}
\end{align}where
\begin{itemize}
    \item[$(a)$] follows by defining $M(\Xi)$ as the set of private messages of the vertices of the proper descendant sub-graph $\Xi$, i.e., $\br{m(k):v_k\in\Xi}$;
    \item[$(b)$] follows since $X_\Xi^n$ is a deterministic function of $M(\Xi)$;
    \item[$(c)$] and $(e)$ follow from Lemma~\ref{lemma:Timing_Lemma} and defining $\epsilon'_n\triangleq2\epsilon_n\log\pr{\text{dim}\pr{\calH}}+2h_2(\epsilon_n)$;
    \item[$(d)$] follows by defining $Q$ as \ac{RV} with uniform distribution over $[n]$;
    \item[$(f)$] follows from the definition of the joint state $\rho_{U_{[T]}X_{[T]}BE}$ as in \eqref{eq:Joint_Time_Sharing_State}, and defining $X_\Xi\triangleq \pr{X_{\Xi,Q},Q}$, and $E\triangleq E_Q$ 
    \item[$(g)$] follows from \eqref{eq:delta}.
\end{itemize}

Now, we have 
\begin{align}
\left\lVert\rho_E-\rho_0\right\rVert_1&\mathop\leq\limits^{(a)}\sqrt{2\bbD\pr{\rho_E\lVert\rho_0}}\nonumber\\
&=\sqrt{2\bbD\pr{\rho_{E_Q}\lVert\rho_0}}\nonumber\\
&=\sqrt{2\bbD\pr{\frac{1}{n}\sum\limits_{q=1}^n\rho_{E_q}\Big\lVert \rho_0}}\nonumber\\
&\mathop\leq\limits^{(b)}\sqrt{\frac{2}{n}\sum\limits_{q=1}^n\bbD\pr{\rho_{E_q}\lVert \rho_0}}\nonumber\\
&\mathop\leq\limits^{(c)}\sqrt{\frac{2}{n}\bbD\pr{\rho_{E^n}\lVert\rho_0^{\otimes n}}}\nonumber\\
&\mathop\leq\limits^{(d)}\sqrt{\frac{8}{n\beta}}\left\lVert\rho_{E^n}-\rho_0^{\otimes n}\right\rVert_1\nonumber\\
&\leq\sqrt{\frac{8}{n\beta}}\epsilon_n\nonumber\\
&\mathop\leq\limits^{(e)}\delta.\label{eq:Covertness_C_Degraded}
\end{align}where
\begin{itemize}
    \item[$(a)$] follows from \cite[Theorem~1.15]{Ohya_Petz_Book};
    \item[$(b)$] follows from the convexity of quantum relative entropy \cite[Corollary~11.9.2]{Wilde_Book} and Jensen's Inequality;
    \item[$(c)$] follows from the superadditivity of the relative entropy \cite[Eq.~(1)]{SuuperAdditivity_RE};
    \item[$(d)$] follows from \cite[Theorem~2]{Audenaert05}, stating that, for $\rho,\sigma\in\calD(\calH)$, we have $\bbD(\rho\lVert\sigma)\le\frac{4}{\beta}\left\lVert\rho-\sigma\right\lVert_1^2$, where $\beta$ is the minimum eigenvalue of $\sigma$;
    \item[$(e)$] follows from \eqref{eq:delta}. 
\end{itemize}

Combining \eqref{eq:Upper_R_cq} and \eqref{eq:Lower_R_cq}, and~\eqref{eq:Covertness_C_Degraded} shows that $\forall \delta_n,\delta>0$, $R\leq \max\{a:a\in\mathcal{Q}_{\delta}\}$. Therefore,
\begin{align}
  R\leq \max\left\{r:r\in\bigcap_{\delta>0}\mathcal{Q}_{\delta}\right\}.\nonumber
\end{align}

The proof of continuity at zero of $\calQ_\delta$ is similar to that of \cite[Appendix~F]{Keyless22} and is omitted.

\end{appendices}
\section{Acknowledgment}
The authors would like to thank R\'{e}mi A. Chou (The University of Texas at Arlington), Zhaoyou Wang (Columbia University), and anonymous reviewers for the helpful discussions and comments.

\bibliographystyle{IEEEtran}
\bibliography{IEEEabrv,bibfile}

@STRING{IEEE_J_STSP       = "{IEEE} J. Sel. Topics Signal Process."}

@STRING{IEEE_J_JSAC       = "{IEEE} J. Sel. Areas Commun."}

@STRING{IEEE_J_WCOM       = "{IEEE} Trans. Wireless Commun."}

@STRING{IEEE_J_IFS        = "{IEEE} Trans. Inf. Forensics Security"}

@STRING{IEEE_J_IT         = "{IEEE} Trans. Inf. Theory"}

@BOOK{Wilde_Book,

AUTHOR    = {Wilde, M. M},
TITLE     = {Quantum Information Theory},
PUBLISHER = {Cambridge University Press},
YEAR      = 2017,
EDITION   = {Second},
ADDRESS   = {\hspace{-0.2cm}Cambridge, U.K},
}

@BOOK{QIT_Hayashi,

AUTHOR    = {Hayashi, M},
TITLE     = {Quantum Information Theory: An Introduction},
PUBLISHER = {\hspace{-0.24cm}Berlin, Germany: Springer},
YEAR      = 2006,
}

@BOOK{Tomamichel16,

AUTHOR    = {Tomamichel, M},
TITLE     = {Quantum Information Processing with Finite Resources},
PUBLISHER = {\hspace{-0.24cm}Springer International Publishing},
YEAR      = 2016,
}

@BOOK{Cover_Book,

AUTHOR    = {Cover, T. M and Thomas, J. A}, 
TITLE     = {Elements of Information Theory},
PUBLISHER = {John Wiley \& Sons, Inc.},
YEAR      = 2001,
EDITION   = {2nd},
}

@ARTICLE{Han_MAC79, 
author={Han, T.-S}, 
journal={Inf. and Contr.}, 
title={The capacity region of general multiple-access channel with certain correlated sources}, 
year={1979}, 
volume={40}, 
pages={37-60}, 
month=Jan,
}

@ARTICLE{Prelov84, 
author={Prelov, V.~V}, 
journal={Probl. Peredachi Inf.}, 
title={Transmission over a multiple access channel with a special source hierarchy}, 
year={1984}, 
volume={20}, 
number={4},
pages={3-10},
}

@ARTICLE{Uhlmann76, 
author={Uhlmann, A}, 
journal={Rep. Math. Phys.}, 
title={The ``transition probability'' in the state space of a *-algebra}, 
year={1976}, 
volume={9}, 
number={2}, 
pages={273-279}, 
}

@ARTICLE{Branum96, 
author={Barnum, H and Caves, C. M and Fuchs, C. A and Jozsa, R and Schumacher, B}, 
journal={Phys. Rev. Lett.}, 
title={Noncommuting mixed states cannot be broadcast}, 
year={1996}, 
volume={76}, 
number={15}, 
pages={2818-2821},
month=apr,
}

@article{Frank13,
    author = {Frank, R. L and Lieb, E. H},
    title = "{Monotonicity of a relative R\'{e}nyi entropy}",
    journal = {J. Math. Phys.},
    volume = {54},
    number = {12},
    pages = {122201},
    year = {2013},
    month = dec,
}

@article{Audenaert05,
    author = {Audenaert, K. M. R and Eisert, J},
    title = {Continuity bounds on the quantum relative entropy},
    journal = {J. Math. Phys.},
    volume = {46},
    number = {10},
    pages = {102104},
    year = {2005},
    month = oct,
}

@article{Sandwiched10,
    author = {M\"uller-Lennert, M and Dupuis, F and Szehr, O and Fehr, S and Tomamichel, M},
    title = {On quantum R\'enyi entropies: A new generalization and some properties},
    journal = {J. Math. Phys.},
    volume = {54},
    number = {12},
    pages = {122203},
    year = {2013},
    month = dec,
}

@ARTICLE{Hayashi02, 
author={Hayashi, M}, 
journal={J. Phys. A Math Gen.}, 
title={Optimal sequence of quantum measurements in the sense of Stein's lemma in quantum hypothesis testing}, 
year={2002}, 
volume={35}, 
number={50}, 
pages={10759}, 
}

@ARTICLE{Salek25, 
author={Salek, F and Hayden, P and Hayashi, M}, 
journal={Ann. Henri Poincar\'{e}}, 
title={Three-Receiver Quantum Broadcast Channels: Classical Communication with Quantum Non-unique Decoding}, 
year={2025},
month = Jun, 
pages={1-73}, 
}

@ARTICLE{Wilde14, 
author={Wilde, M. M and Winter, A and Yang, D}, 
journal={Commun. Math. Phys.}, 
title={Strong converse for the classical capacity of entanglement-breaking and {H}adamard channels via a sandwiched {R}\'{e}nyi relative entropy}, 
year={2014}, 
volume={331}, 
pages={593-622}, 
}

@article{SuuperAdditivity_RE,

author  = {Caper, \'{A} and Lucia, A},
journal = IEEE_J_IT,
title   = {Superadditivity of Quantum Relative Entropy for General States},
volume  = {64},
number  = {7},
month   = jul,
year    = {2018},
pages   = {4758-4765},
}

@ARTICLE{Jozsa94, 
author={Jozsa, R}, 
journal={J. Mod. Opt.}, 
title={Fidelity for mixed quantum states}, 
year={1994}, 
volume={41},
number  = {12},
month   = dec,
pages={2315-2323}, 
}

@ARTICLE{Rastegin02, 
author={Rastegin, A. E}, 
journal={Phys. Rev. A, Gen. Phys.}, 
title={Relative error of state-dependent cloning}, 
year={2002}, 
volume={66},
number  = {4},
month   = oct,
pages={042304}, 
}

@article{Tomamichel10,

author  = {Tomamichel, M and Colbeck, R and Renner, R},
journal = IEEE_J_IT,
title   = {Duality between smooth min- and max-entropies},
volume  = {56},
number  = {9},
month   = sep,
year    = {2010},
pages   = {4674-4681},
}

@article{Hayashi03,

author  = {Hayashi, M  and Nagaoka, H},
journal = IEEE_J_IT,
title   = {General formulas for capacity of classical-quantum channels},
volume  = {49},
number  = {7},
month   = jul,
year    = {2003},
pages   = {1753-1768},
}

@article{Helal_Cribbibg,

author  = {Helal, N and Bloch, M. R and Nosratinia, A},
journal = IEEE_J_IT,
title   = {Cooperative resolvability and secrecy in the cribbing multiple-access channel},
volume  = {66},
number  = {9},
month   = may,
year    = {2020},
pages   = {5429-5447},
}

@article{Bullock_25,

author  = {Bullock, M.~S and Sheikholeslami, A and Tahmasbi, M and Macdonald, R.~C and Guha, S and Bash, B.~A},
journal = IEEE_J_IT,
title   = {Fundamental Limits of Covert Communication Over Classical-Quantum Channels},
volume  = {71},
number  = {4},
month   = apr,
year    = {2025},
pages   = {2741-2762},
}

@article{Bash_15,

author  = {Bash, B.~A and Gheorghe, A.~H and Patel, M and Habif, L.~H and Goeckel, D and Towsley, D and Guha, S},
journal = {Nature Commun.},
title   = {Quantum-secure covert communication on bosonic channels},
volume  = {6},
number  = {1},
month   = oct,
year    = {2015},
pages   = {1-9},
}

@article{MAC_LPD,

author  = {Arumugam, K. S. K and Bloch, M. R},
journal = IEEE_J_IT,
title   = {Covert Communication Over a {$K$}-User Multiple-Access Channel},
volume  = {65},
number  = {11},
month   = jul,
year    = {2019},
pages   = {7020 - 7044},
}

@article{QMAC,

author  = {Winter, A},
journal = IEEE_J_IT,
title   = {The capacity of the quantum multiple access channel},
volume  = {47},
number  = {7},
month   = nov,
year    = {2001},
pages   = {3059 - 3065},
}

@article{Keyless22,

author = {ZivariFard, H and Bloch, M. R and Nosratinia, A},
title  = {Keyless Covert Communication via Channel State Information},
journal = IEEE_J_IT,
volume  = {68},
number  = {8},
month   = Aug,
year    = {2022},
pages   = {5440-5474},
}

@article{4M6Tx,

author  = {ZivariFard, H and Wang, X},
title   = {Covert Communication Over a Quantum MAC with 4 Messages and 6 Transmitters},
journal = {available at \url{https://drive.google.com/file/d/1llHGDhTbfeTTsHpaVMM-VJspvqHBh6Fw/view?usp=sharing}},
month   = Sep,
year    = {2025},
}

@article{2M3Tx,

author  = {ZivariFard, H and Chou, R and Wang, X},
title   = {Covert Communication Over a Quantum MAC with a Helper},
journal = {available at \url{https://arxiv.org/abs/2504.18747}},
month   = Jul,
year    = {2026},
}

@inproceedings{2M3Tx_ISIT,

author    = {ZivariFard, H and Chou, R and Wang, X},
title     = {Covert Communication Over a Quantum MAC with a Helper},
booktitle = {Proc. {IEEE} Int.  Symp. on Info. Theory (ISIT)},
address   = {Ann Arbor, MI, USA},
month     = Jun,
year      = {2025}
}

@article{CsiszarKorner,

author  = {Csisz\'{a}r, I and  K\"{o}rner, J},
title   = {Broadcast Channels with Confidential Messages},
journal = IEEE_J_IT,
volume  = {24},
number  = {3},
month   = May,
year    = {1978},
pages   = {339-348},
}

@inproceedings{Wang16,

author    = {Wang, L},
title     = {Optimal throughput for covert communication over a classical-quantum channel},
booktitle = {Proc. {IEEE} Info. Theory Workshop (ITW)},
address   = {Cambridge, UK},
month     = Sep,
year      = {2016},
pages     = {1-5},
}

@inproceedings{GunduzSimeone10,

author    = {G\"{u}nd\"{u}z, D and Simeone, O},
title     = {On the Capacity Region of a Multiple Access Channel with Common Messages},
booktitle = {Proc. {IEEE} Int.  Symp. on Info. Theory (ISIT)},
address   = {Austin, TX USA},
month     = Jun,
year      = {2010},
pages     = {470-474}
}

@article{Entanglement_MAC,

author  = {Hsieh, M.-H and  Devetak, I and Winter, A},
title   = {Entanglement-Assisted Capacity
of Quantum Multiple-Access Channels},
journal = IEEE_J_IT,
volume  = {54},
number  = {7},
month   = Jul,
year    = {2008},
pages   = {3078-3090},
}

@article{AshnuHayashi2020,

author  = {Anshu, A and Hayashi, M and Warsi, N.~A},
title   = {Secure Communication Over Fully Quantum {G}el'fand-{P}insker Wiretap Channel},
journal = IEEE_J_IT,
volume  = {66},
number  = {9},
month   = Sep,
year    = {2020},
pages   = {5548-5566},
}

@article{UninformedJammer,

author  = {Sobers, T. V and Bash, B.~A and Guha, S and Towsley, D and Goeckel, D},
title   = {Covert Communication in the Presence of an Uninformed Jammer},
journal = IEEE_J_WCOM,
volume  = {16},
number  = {9},
month   = Sep,
year    = {2017},
pages   = {6193-6206},
}

@inproceedings{ISIT21,

author    = {ZivariFard, H and Bloch, M. R and Nosratinia, A},
title     = {Covert Communication via Non-Causal Cribbing from a Cooperative Jammer},
booktitle = {Proc. {IEEE} Int.  Symp. on Info. Theory (ISIT)},
address   = {Melbourne, Australia},
month     = Jul,
year      = {2021},
pages     = {202-207},
}

@inproceedings{ISIT22,

author    = {ZivariFard, H and Bloch, M. R and Nosratinia, A},
title     = {Covert Communication in the Presence of an Uninformed, Informed, and Coordinated Jammer},
booktitle = {Proc. {IEEE} Int.  Symp. on Info. Theory (ISIT)},
address   = {Melbourne, Australia},
month     = Jul,
year      = {2022},
pages     = {306-311},
}

@phdthesis{MyDissertation,
  author       = {ZivariFard, H}, 
  title        = {Secrecy and covertness in the presence of multi-casting, channel state information, and cooperative jamming},
  school       = {Univ. Texas at Dallas, TX, USA},
  year         = 2021,
  month        = Dec,
}

@phdthesis{RennerDissertation,
  author       = {Renner, R}, 
  title        = {Security of quantum key distribution},
  school       = {ETH, Zurich, Switzerland},
  year         = 2005,
  month        = Sep,
}

@article{SlepianWol_MAC,

author  = {Slepian, D and Wolf, J. K},
title   = {A Coding Theorem for Multiple Access Channels with Correlated Sources},
journal = {Bell System Technical Journal},
volume  = {52},
number  = {7},
month   = Sep,
year    = {1973},
pages   = {1037-1076},
}

@BOOK{ElGamalKim,

AUTHOR    = {El Gamal, A and Kim, Y.-H},
TITLE     = {Network Information Theory},
PUBLISHER = {Cambridge University Press},
YEAR      = {2012},
EDITION   = {First},
ADDRESS   = {\hspace{-0.2cm}Cambridge, U.K},
}

@BOOK{Ohya_Petz_Book,

AUTHOR    = {Ohya, M and Petz, D},
TITLE     = {Quantum Entropy and its Use},
PUBLISHER = {Springer-Verlag},
YEAR      = {1993},
ADDRESS   = {\hspace{-0.2cm}Heidelberg, Germany},
}

@BOOK{Bhatia_Book,

AUTHOR    = { Bhatia, R},
TITLE     = {Matrix Analysis},
PUBLISHER = {Springer-Verlag},
YEAR      = {1996},
ADDRESS   = {New York, NY, USA},
}

@article{WillemsCribbing,

author  = {Willems, F. M. J and van der Meulen, E. C},
title   = {The discrete memoryless multiple-access channel with cribbing encoders},
journal = IEEE_J_IT,
volume  = {31},
number  = {3},
month   = May,
year    = {1985},
pages   = {313-327},
}

@article{UziCribbing,

author  = {Pereg, U and Deppe, C and Boche, H},
title   = {The Quantum Multiple-Access Channel
With Cribbing Encoders},
journal = IEEE_J_IT,
volume  = {68},
number  = {6},
month   = Jun,
year    = {2022},
pages   = {3965 - 3988},
}

@article{QMAC_Security,

author  = {Chou, R},
title   = {Private Classical Communication over Quantum Multiple-Access Channels},
journal = IEEE_J_IT,
volume  = {68},
number  = {3},
month   = Mar,
year    = {2022},
pages   = {1782 - 1794},
}

@article{Romero17,

author  = {Romero, H. P and Varanasi, M. K},
title   = {A Unifying Order-Theoretic Framework for Superposition Coding: Polymatroidal Structure and Optimality in the Multiple-Access Channel With General Message Sets},
journal = IEEE_J_IT,
volume  = {63},
number  = {1},
month   = Jan,
year    = {2017},
pages   = {21 - 37},
}

@article{Lee15,

author  = {Lee, S and Baxley, R. J and Weitnauer, M. A and Walkenhorst, B},
journal = IEEE_J_STSP,
title   = {Achieving undetectable communication},
volume  = {9},
number  = {7},
month   = Oct,
year    = {2015},
pages   = {1195-1205},
}

@article{Action_Covert,

author  = {ZivariFard, H and Wang, X},
title   = {Covert Communication via Action-Dependent States},
journal = IEEE_J_IT,
volume  = {71},
number  = {4},
month   = Apr,
year    = {2025},
pages   = {3100-3128},
}

@article{Covert_With_State,

author  = {Lee, S.-H. and  Wang, L and Khisti, A and Wornell, G. W.},
journal = IEEE_J_IFS,
title   = {Covert Communication With Channel-State Information at the Transmitter},
volume  = {13},
number  = {9},
month   = Sep,
year    = {2018},
pages   = {2310-2319},
}

@article{LPD_on_AWGN,

author  = {Bash, B.~A and  Goeckel, D and Towsley, D.},
journal = IEEE_J_JSAC,
title   = {Limits of Reliable Communication with Low Probability of Detection on {AWGN} Channels},
volume  = {31},
number  = {9},
month   = Sep,
year    = {2013},
pages   = {1921-1930},
}

@article{LPD_on_bosonic,

author  = {Bullock, M. S and Gagatsos, C. N and Guha, S and Bash, B.~A},
journal = IEEE_J_JSAC,
title   = {Fundamental limits of quantum-secure covert communication over bosonic channels},
volume  = {38},
number  = {3},
month   = Jan,
year    = {2020},
pages   = {471-482},
}

@inproceedings{Wang23,

author    = {Wang, S.-Y and Erdo\u{g}an, T and Bloch, M. R},
title     = {Towards a characterization of the covert capacity of bosonic channels under trace distance},
booktitle = {Proc. {IEEE} Int.  Symp. on Info. Theory (ISIT)},
address   = {Espoo, Finland},
month     = Jul,
year      = {2023},
pages     = {318-323}
}

@inproceedings{Entangled_Bosonic,

author  = {Bullock, M. S and Gagatsos, C. N and Guha, S and Bash, B.~A},
title   = {Fundamental limits of quantum-secure covert communication over bosonic channels},
booktitle = {OSA Quantum 2.0 Conf.},
year    = {2020},
pages   = {QM6B.5},
}

@article{LPD_over_DMC,

author  = {Wang, L and  Wornell, G. W. and Zheng, L.},
journal = IEEE_J_IT,
title   = {Fundamental limits of communication with low probability of detection},
volume  = {62},
number  = {6},
month   = Jun,
year    = {2016},
pages   = {3493-3503},
}

@article{MAC_Bosonic_05,
  title = {Multiple-access bosonic communications},
  author  = {Yen, B.~J and Shapiro, J.~H},
  journal = {Phys. Rev. A},
  volume = {72},
  issue = {6},
  pages = {062312},
  numpages = {10},
  year = {2005},
  month = {Dec},
}

@article{Guha07,
  title = {Classical capacity of bosonic
broadcast communication and a minimum output entropy conjecture},
  author  = {Guha, S and Shapiro, J.~H and Erkmen, B.~I},
  journal = {Phys. Rev. A, Gen. Phys.},
  volume = {76},
  issue = {3},
  pages = {032303},
  numpages = {11},
  year = {2007},
  month = {Sep},
}

@article{Holevo01,
  title = {Evaluating capacities of bosonic Gaussian channels},
  author  = {Holevo, A.~S and Werner, R.~F},
  journal = {Phys. Rev. A, Gen. Phys.},
  volume = {63},
  issue = {3},
  pages = {032312},
  numpages = {11},
  year = {2001},
  month = {Feb},
}

@article{Eisert07,
  title = {Gaussian Quantum Channels},
  author  = {Eisert, J. and Wolf, M.~M},
  journal = {Quantum Information with Continuous Variables of Atoms and Light. Singapore: World Scientific Report},
  volume = {63},
  issue = {3},
  pages = {23-42},
  year = {2007},
}

@article{RevModPhys12,
  title = {Gaussian quantum information},
  author = {Weedbrook, C and Pirandola, S and Garc\'{\i}a-Patr\'on, R and Cerf, N.~J and Ralph, T.~C and Shapiro, J.~H. and Lloyd, S},
  journal = {Rev. Mod. Phys.},
  volume = {84},
  issue = {2},
  pages = {621-669},
  numpages = {0},
  year = {2012},
  month = {May},
}

@article{Covert_Quantum16,
  title = {Covert Quantum Communication},
  author = {Arrazola, J. M and Scarani, V},
  journal = {Phys. Rev. Lett.},
  volume = {117},
  issue = {25},
  pages = {250503},
  numpages = {5},
  year = {2016},
  month = {Dec},
  publisher = {American Physical Society},
}

@inproceedings{Reliable_Deniable_Comm,

author    = {Che, P. H. and  Bakshi, M. and Jaggi, S.},
title     = {Reliable deniable communication: hiding messages in noise},
booktitle = {Proc. {IEEE} Int.  Symp. on Info. Theory (ISIT)},
address   = {Istanbul, Turkey},
month     = Jul,
year      = {2013},
pages     = {2945-2949},
}

@inproceedings{Deniable_ITW14,

author    = {Che, P. H. and  Bakshi, M. and  Chan, C. and Jaggi, S.},
title     = {Reliable deniable communication with channel uncertainty},
booktitle = {Proc. {IEEE} Info. Theory Workshop (ITW)},
address   = {Hobart, TAS, Australia},
month     = Nov,
year      = {2014},
pages     = {30-34},
}

@inproceedings{ITW24,

author  = {ZivariFard, H and Chou, R. A and Wang, X},
title     = {Covert Communication with Positive Rate Over State-Dependent Quantum Channels},
booktitle = {Proc. {IEEE} Info. Theory Workshop (ITW)},
address   = {Shenzhen, China},
month     = Nov,
year      = {2024},
pages     = {711-716},
}

@article{LPD_by_Resolvability,

author  = {Bloch, M. R.},
journal = IEEE_J_IT,
title   = {Covert communication over noisy channels: A resolvability perspective},
volume  = {62},
number  = {5},
month   = May,
year    = {2016},
pages   = {2334-2354},
}

\end{document}